\documentclass[12pt,a4paper]{article}
\usepackage{amsmath}
\usepackage{graphics}
\usepackage{graphicx}
\usepackage{epsfig}
\usepackage{floatrow}
\usepackage{float}
\usepackage{amssymb}
\usepackage{amsthm}
\usepackage{graphicx}
\usepackage{epstopdf}
\usepackage{anysize}
\usepackage{geometry}
\usepackage{setspace}
\usepackage{longtable}
\usepackage{verbatim}

\usepackage{subfig}
\usepackage{xcolor}
\usepackage{color,soul}
\usepackage{hyperref}
\usepackage{cite}
\hypersetup{colorlinks,linkcolor={blue},citecolor={blue},urlcolor={black}}

\marginsize{1cm}{1cm}{1cm}{1cm}
\input epsf.sty
\usepackage{lineno}
\usepackage[english]{babel}

\newtheorem{theorem}{Theorem}[section]

\newtheorem{proposition}[theorem]{Proposition}

\newtheorem{remark}[theorem]{Remark}

\newtheorem{numerical example}[theorem]{Numerical example}

\begin{document}
	\begin{center}
		{\bf Assessing potential insights of an imperfect testing strategy: Parameter estimation and practical identifiability using early COVID-19 data in India}
	\end{center}
\begin{center}
	\baselineskip .2in {\bf Sarita Bugalia$^1$, Jai Prakash Tripathi$^{1}$}
	
	{\small $^1$ Department of Mathematics, Central University of Rajasthan,
		\\ Bandar Sindri, Kishangarh-305817, Ajmer, Rajasthan, India}
\end{center}

\begin{abstract}
	A deterministic model with testing of infected individuals has been proposed to investigate the potential consequences of the impact of testing strategy. The model exhibits global dynamics concerning the disease-free and a unique endemic equilibrium depending on the basic reproduction number when the recruitment of infected individuals is zero; otherwise, the model does not have a
	disease-free equilibrium, and disease never dies out in the community. Model parameters have been estimated using the maximum likelihood method with respect to the data of early COVID-19 outbreak in India. The practical identifiability analysis shows that the model parameters are estimated uniquely. The consequences of the testing rate for the weekly new cases of early COVID-19 data in India tell that if the testing rate is increased by 20\% and 30\% from its baseline value, the weekly new cases at the peak are decreased by 37.63\% and 52.90\%; and it also delayed the peak time by four and fourteen weeks, respectively. Similar findings are obtained for the testing efficacy that if it is increased by 12.67\% from its baseline value, the weekly new cases at the peak are decreased by 59.05\% and delayed the peak by 15 weeks. Therefore, a higher testing rate and efficacy reduce the disease burden by tumbling the new cases, representing a real scenario. It is also obtained that the testing rate and efficacy reduce the epidemic's severity by increasing the final size of the susceptible population. The testing rate is found more significant if testing efficacy is high. Global sensitivity analysis using partial rank correlation coefficients (PRCCs) and Latin hypercube sampling (LHS) determine the key parameters that must be targeted to worsen/contain the epidemic.    
\end{abstract}

%\noindent {$\dag$Corresponding author's email: jtripathi85@gmail.com}

\section{Introduction}
Mathematical models have an extensive history of being utilized to explain the transmission of communicable diseases from plague pandemics more than a century back \cite{Kermack1927} to the recent SARS outbreak \cite{GUM2004}, and Ebola epidemics \cite{BRO2015, WHI2017}, and from decision making about distinct vaccination strategies for influenza \cite{BAG2013, HOD2017}, to modeling HIV \cite{DOR2017, STU2018}, and from modeling the influenza pandemic \cite{GRI2020} to presently assisting making decisions around the COVID-19 pandemic \cite{HEL2020, KUC2020, BUG2020, GOS2020, DAV2020, BAJ2020, BUG2022}. There are numerous general approaches to modeling, each with different disadvantages and advantages \cite{HOL2009, MAN2016}. Compartmental modeling approach using differential equations \cite{AndersonRM1991, HEE2015, BUG2021} separate the whole population into various compartments of subpopulations, for instance, susceptible, exposed to the disease but not infectious, infectious (capable of transmitting the disease) and recovered. This particular approach also tracks the transfers of individuals among different compartments. Mathematical modeling based on differential equations offers an extensive mechanism for the disease dynamics and assesses the control strategies' efficacy in lowering down the disease load. It has also participated significantly in the control of disease spread, involving the assessment of social distancing measures \cite{FER2020}. It offers a reasonable framework for comprehension of the spread of a transmittable disease throughout a population and permits various intervention strategies to be investigated, involving contact tracing and testing of infected persons as feasible methods to relieve social distancing constraints. Such models are also inevitably understood, and simplifying their hypotheses and what they do and do not signify are essential to interpret them accurately.     

The modern literature on mathematical modeling of the infectious diseases and diagnosis is broad. Diagnostics may contract much exact and correct at the expenditure of a long period to produce the test outcomes. Commencing treatment immediately after the detection of a case is usually expected to be amongst the finest procedures, but the test outcomes may deteriorate from test limitations or deficiencies.  
A few studies \cite{VIL2017,CHI2020,STU2021,GRI2021,SAL2014} have assessed the impact of testing via mathematical models. Chirove et al. \cite{CHI2020} proposed an epidemic model with isolation and mass testing  to study the COVID-19 disease in Nigeria. The authors concluded that increment in the mass testing and isolation could lower the peak of symptomatic cases and decrease the cumulative deaths. Griette et al. \cite{GRI2021} proposed an epidemic model with the daily number of tests of New York state as an input. The authors infer that an increment in the number of tests could reduce the number of daily cases in the starting, but after escalating the number of tests ten times, there is no substantial variation in the number of reported cases. The study \cite{GRI2021} suggests that there must be an optimum between being effectual to slow down the epidemic and the number of tests; the optimal strategy may depend on the factors such as the monetary cost of the tests and other limited resources.   

Testing drives have been crucial in hitting the transmission of infectious diseases, such as HIV and COVID-19. However, an imperfect test might wrongly count infected individuals as susceptible (false negatives) or susceptible to be infected (false positives). Test imperfections reduce the likelihood of obtaining diagnostic outcomes as true negatives or positives. The primary impact arises from the test  specificity, i.e., the possibility that the test result specifies a negative diagnosis, agreed that the tested person does not have the disease. Evading specificity indicates that a considerable portion of test outcomes may specify sick individuals while those are not diseased. In many circumstances, initiating medical care is the primary concern over expecting for more correct test outcomes. Consequently, imperfections here may enhance the cost \cite{VIL2017}. Higher sensitivity is also vital for infected persons since the rate of confirmed diagnosis in the situation of positive tested persons must be increased. This requirement again falls at an expenditure of fast diagnosis. A result of accurately diagnosed individuals has a twofold advantage, to commence medical care and decrease the time to possibly transmit disease to other susceptible persons because of quarantine/isolation and other control interventions. Villela \cite{VIL2017} proposed a general mathematical model to describe the dynamics of test imperfections and diagnostics for a contagious disease. In \cite{VIL2017}, the author presented different scenarios to obtain the disease eradication and concluded that epidemics are less expected to occur under a larger testing rate.
The advantage of specificity and sensitivity has shown in some recent studies investigating the transmission of HIV \cite{EAT2014}, Ebola \cite{NOU2015}, Malaria \cite{BIS2009}, and COVID-19 \cite{GRI2021}. These studies, however, introduce models that demonstrate complexity, for instance, investigating the case fatality ratio for Ebola \cite{NOU2015}. 
%Testing people for a particular illness allows health diagnostic units to efficiently provide treatment amenities if tests' outcomes become positive, which hurries up their isolation/treatment and may also reduce people's interactions with others. 

There might be more than one test for different diseases. For instance, for tuberculosis (TB), the Xpert (R) MTB/R if assay test for diagnosis of TB established for the GeneXpert platform \cite{WHOTB2011}. Such types of molecular testing could be performed on request and nearer to people in require. Walusimbi et al. \cite{WAL2013} have done a meta-analysis of these reports. The authors in \cite{WAL2013} observed pooled specificity and sensitivity anticipated at 0.98 (CI: 0.97-0.99) and 0.67 (CI: 0.62-0.71), respectively. For hepatitis, rapid test outcomes were reported, which could hold both high sensitivity and specificity but can fluctuate to a huge extent reliant on the pathogens \cite{AMA2006} and the type of sample (whole blood, serum, saliva) \cite{PAU2014}. James J. Cochran \cite{Cochran2020} conferred the significance of testing a random sample for COVID-19. In the case of COVID-19, there is strong evidence that many people develop no symptoms, mild symptoms, convey the virus without realizing it, and infect other individuals. 
There is a significant need for random testing to recognize how deep the virus has moved into the community. In order to break the transmission chain of COVID-19, testing of cases, tracing, and isolation of their contacts have been implemented as critical non-pharmaceutical intervention strategies (NPIs) in several countries \cite{Quilty2020}. 

There are some evidences of imperfect testing for the COVID-19 \cite{falsenegative1, falsenegative2, WOLO2020}. For instance, the real-time reverse transcription PCR (RT-PCR) tests, like the individuals used to diagnose COVID-19, false negatives appear for many reasons, such as the level of viral RNA being below the limit of detection of the test. Also, if the limit of detection for a test is too low, the test will detect the smallest amount of viral RNA, leading to false-positive test results \cite{falsenegative1}. Thus, RT-PCR also has an imperfect sensitivity \cite{ZHA2021,YAN2020,ARE2020,ZHAO2020}. Fang et al. \cite{FANG2020} compared the sensitivity of Chest CT with RT-PCR and concluded that chest CT has greater sensitivity than RT-PCR, 98\% vs. 71\%, respectively. The execution of testing outcomes in asymptomatic individuals has indicated lower sensitivity \cite{BRI2021}. Though, PCR testing continues the gold standard for SARS-CoV-2. Whereas lack of PCR tests and testing deliveries along with an extreme requirement for testing have directed to continued delays in test results and turnaround reporting \cite{falsenegative2}. 
Some studies shows that RT-PCR has greater sensitivity than other antigen tests for COVID-19. For instance, the Quidel rapid antigen test has 72.1\% sensitivity in symptomatic and 60.5\% in asymptomatic \cite{BRI2021}. The false-negative outcome may ruin the disease prevention and control and cause lags in treatment. It may also involve needless contact tracing and quarantine/isolation. False-negative outcomes are more subsequent since infected individuals--who may be asymptomatic--might not be quarantined or isolated and may communicate a disease to others \cite{WOLO2020}.  %Diagnostic tests (usually comprising a nasopharyngeal swab) can be incorrect in two directions. 
%A false positive outcome incorrectly identifies an individual infected, with outcomes involving needless contact tracing and quarantine. 
%When the COVID-19 pandemic started, reverse-transcriptase PCR (RT-PCR) tests were the first to be established and extensively used. 

Despite developments in medicine, there is still no vaccine for several communicable diseases, and even when a vaccine is identified, it is usually impossible to preserve an ample stock \cite{Young2019}. Sometimes new variants of the virus also escape the vaccine-induced immunity. Preventive measures may fall out of trend when the infectious disease is out for a long time, only to hit more stringent when it reoccurs due to the poorer level of immunity. A standard recommendation of health administrations in the appearance of an epidemic is the isolation of infected persons. This control policy has been implemented for centuries, and its effectiveness has not faded with time, as evidenced in the recent outbreaks of Ebola in West Africa, SARS in Asia, and COVID-19 worldwide. Research has advanced in the meantime, compartmental models to comprehend the effectiveness and impact of testing and isolation have been anticipated and analyzed. It has been revealed that isolation protocols can be boosted by contact tracing and isolation of people earlier to symptoms. Even detecting diseased hosts, a critical primary step in any isolation policy, can be challenging. At the initial stages of an epidemic, even local health authorities may flop to identify a disease's symptoms (which are often vague). They may also not be aware of the potential risks of runaway the exponential growth of contagious diseases. However, for centuries, isolation of infected individuals has been an effective control strategy for unexpected epidemic outbreaks. 

Cox \cite{COX2014} in a latest editorial expressed the requirement for naive models to explain findings, particularly for the community beyond modeling. In this work, we propose a simple compartmental model that ponders the least number of variables to acquire significant insights and describes the dynamics of a disease with imperfect testing and isolation/quarantine. Our proposed ODE model system, leveraged from a standard SIR model, has two additional compartments: $T_n$, to designate the number of infected persons that undergo the testing and are false negative, and $J$, that expresses the number of infected persons that move into the quarantine or initiate treatment upon test--diagnosed positive. Infected individuals test negative but are actually infected (false-negative) at a rate committed with testing efficacy/sensitivity ($1-\sigma$). It is provided by the rate of true positive, i.e., diseased and tested positive. Also, the infected persons are tested and moved to quarantine/isolation/treatment compartment if their tests' outcomes turn out positive. Since an imperfect testing can also give the false-positive result when testing the susceptibles, the false-positive results do not increase the disease's basic reproduction number or endemicity. Thus, our model does not consider the testing of susceptibles. Our main goal is to propose a model that designates the threshold of an epidemic state as a function of parameters carrying imperfections of testing that could be supportive while attempting to control an epidemic. Consequently, we obtain the basic reproduction number $R_0$ as a function of the parameter, $\sigma$, where $1-\sigma$ is the testing efficacy. The model also includes the immigration of infected individuals, particularly a fraction $q$ of the recruitment rate. We present analytical and numerical results that address fundamental questions relating to isolation and imperfections in the implementation of testing strategies to control the disease. 

Considering the demographic effects zero, we apply the model to the data of the initial outbreak of COVID-19 in India. Whereas the model has been established in the general context, it is general enough to be applicable with such changes as necessary to other diseases as well. Mathematical modeling in disease dynamics has a vast literature nowadays, and there is rising identification that can help uncover the mechanisms of the response of many intervention strategies. However, making quantitative predictions with such models often requires parameter estimation from actual data, raising concerns about parameter estimability and identifiability. We explored the identifiability of our model by using the Fisher Information Matrix (FIM) and profile likelihood, exposing that the estimated parameters are practically identifiable.     

The remaining paper is systematized as follows. The model construction and its biological properties, such as positivity and boundedness are presented in the following Section \ref{sec2}. In Section \ref{rganl}, a rigorous mathematical analysis of the model is presented, including computation of the basic reproduction number, local and global asymptotic stability of the equilibria (disease-free and endemic), transcritical bifurcation. By avoiding the demographic effect and waning immunity, a single outbreak model is analyzed, including the epidemic's peak and final size relations, parameter estimation and practical identifiability, global sensitivity analysis, and some quantitative results are presented in Section \ref{peakandfinal}. The paper ends with a thorough discussion and conclusions in Section \ref{diss}.

\section{Model construction}\label{sec2}
In order to develop the mathematical model, we separate the total population, $N(t)$, of a community into six different states: susceptible, infected (both symptomatic and asymptomatic), infected but false negative, isolated/quarantined, and recovered; the numbers of individuals in these compartments are symbolized by $S(t),$ $I(t),$ $T_n(t),$ $J(t),$ and $R(t),$ respectively.
The susceptible population increases by the recruitment of persons (either by immigration or birth), by losing the immunity attained naturally or from  infection. The susceptibles reduce through infection (transferring to compartment $I$) and natural death.  
The number of diseased persons increases by the immigration of infected persons from outside the population and by the disease transmission to susceptibles. It diminishes by natural death, disease-induced death, and testing (moving to compartment $T_n$ and $J$). 
The population of false tested individuals increases by infected individuals that are infected but tested negative. It decreases by natural death and disease-induced death. Since an imperfect test may examine an infected person  to be false negative. In this case, the false-negative infected individuals can make contact with susceptible individuals, and ultimately the imperfect testing increases the epidemic.
The population of isolated individuals increases by testing infected individuals and are isolated. It diminishes by natural death, disease-induced death, and recovery from the infection (moving to the compartment $R$). Due to   isolation/quarantine, infectives move to the compartment $J$ with no isolated infectives making contact with susceptibles. 
Since the acquired immunity (from infection or naturally) fades with time, the recuperated persons become susceptible to the infection again. Thus the recovered individuals increase by the persons who recovers from their infection and gaining immunity. The recovered class decreases via waning immunity and natural death rate. 

The transition rates from one compartment to another are described in Table \ref{table_1}. We consider the disease spread in human populations; the parameters are assumed to be nonnegative. The schematic diagram of the disease spread is portrayed in Figure \ref{schematic}. 

\begin{figure}[H]
	\includegraphics[scale=0.8]{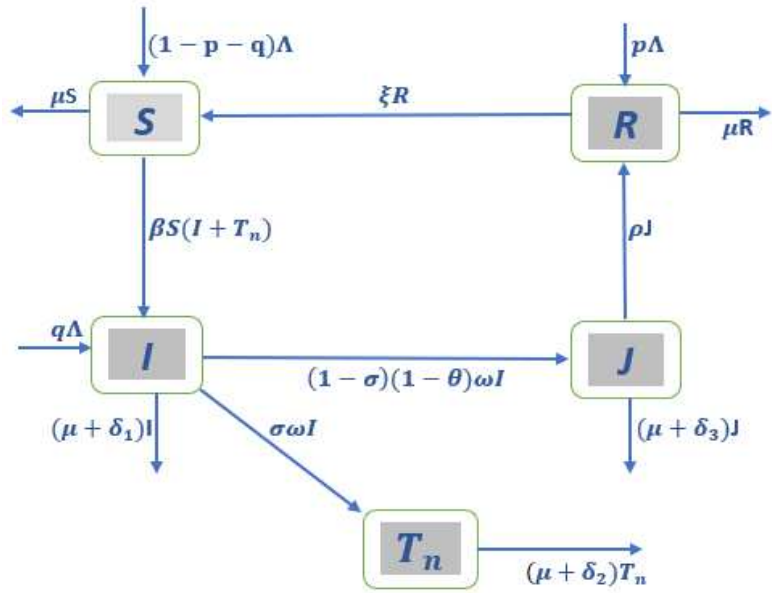}
	\caption{Schematic diagram describing the interaction between individuals in different compartments of system \eqref{model}.}\label{schematic}
\end{figure}

The ordinary differential equations for the transitions among the compartments can be articulated as follows: 
\begin{equation} \label{model}
\begin{aligned}
\frac{dS}{dt} &= (1-p-q)\Lambda - \beta S(I+T_n) - \mu S + \xi R, \\
\frac{dI}{dt} &= q\Lambda + \beta S(I+T_n) - (\mu + \omega + \delta_1)I, \\
\frac{dT_n}{dt} & = \sigma \omega I - (\mu + \delta_2)T_n, \\
\frac{dJ}{dt} & = (1- \sigma) (1-\theta) \omega I - (\mu + \delta_3 + \rho)J, \\
\frac{dR}{dt} & = p\Lambda + \rho J - (\mu + \xi) R,
\end{aligned}
\end{equation}
with nonnegative initial conditions $S(0)>0, I(0) \geq 0, T_n(0) \geq 0, J(0) \geq 0, R(0) \geq0$ and $N(0)>0.$

\begin{table}[H]
	\caption{Biological interpretations of model parameters.} \label{table_1}
	\begin{tabular}{p{2cm}p{10cm}}
		\hline
		Parameters	& Biological explanations \\
		\hline
		$\Lambda$	& The birth/immigration rate at which new persons join to the susceptible class \\
		
		$p$ & The fraction of recruited persons who are having immunity against the disease \\
		
		$q$ & The fraction of newly joined individuals who are already diseased \\
		
		$\beta$	& The transmission rate between infected and susceptible persons \\
		
		$1/\mu$	& Average life-span \\
		
		$\delta_1$	& Disease related death rate of infected individuals \\
		
		$\omega$ & The rate at which infected individuals are tested \\
		
		$\delta_2$	& Disease related death rate of individuals who are infected but tested negative \\
		
		$\delta_3$	& Disease related death rate of individuals who are isolated \\
		
		$\sigma$ & A scaling factor ($0\leq \sigma \leq 1,$ where $1-\sigma$ describes the efficacy/sensitivity of testing i.e. $\sigma = 0$ represents a testing that offer $100\%$ perfection, while $\sigma = 1$ models a testing that offers imperfection at all) \\
		
		$\theta$ &  A fraction of tested individuals that are waiting for tests' results \\
		
		$1/\rho$ & Average length of infection \\
		
		$1/\xi$	& Average time of losing immunity attained from infection or naturally \\
		
		\hline
	\end{tabular}
\end{table}

\subsection{Cost of infection}
We utilize here the model in which the treatment cost changes with three quantities, the quantity $I$ of diseased individuals, the quantity $T_n$, and the quantity $J$ of isolated diseased individuals. The cost is expressed as a function of the time $\tau$ and testing rate $\omega$, specified by $C_a(\omega, \tau) = \int_{0}^{\tau} (\alpha_{I}I(\omega, t) + \alpha_{T_n}T_n(\omega, t) + \alpha_{J}J(\omega, t))dt $. The associated cost $C(\omega, \tau)$ is described by the ratio between the cost under the rate $\omega$ and without treatment, i.e., $\omega = 0,$ given by: 
\begin{equation}
C(\omega, \tau) = \frac{\int_{0}^{\tau} (\alpha_{I}I(\omega, t) + \alpha_{T_n}T_n(\omega, t) + \alpha_{J}J(\omega, t))dt}{\int_{0}^{\tau} (\alpha_{I}I(0, t) + \alpha_{T_n}T_n(0, t) + \alpha_{J}J(0, t))dt},
\end{equation}
where $\alpha_{I}, \alpha_{T_n},$ and $\alpha_{J}$ are weights associated to the variables $I, T_n,$ and $J,$ respectively. In applications, these weights should be specified by measurements, for instance, economic cost per person case or time utilization from health amenities or others. After a control measure or through an initial epidemic time period, such methods may compare the treatment cost. 

Another opportunity is to compare such costs experienced by testing after attaining the endemic state when $R_0 > 1$. For such comparison, we consider the endemic cost specified by weights $\alpha_{I}, \alpha_{T_n},$ and $\alpha_{J}$ with respect to the variables $I, T_n,$ and $J,$ in the endemic state, respectively. The associated endemic cost $C(\omega)$ is the ratio between the endemic cost with the testing $(\omega)$ and deprived of testing:
\begin{equation}
C(\omega) = \frac{\alpha_{I}I^*(\omega) + \alpha_{T_n}T_n^*(\omega) + \alpha_{J}J^*(\omega)}{\alpha_{I}I^*(\omega = 0) + \alpha_{T_n}T_n^*(\omega = 0) + \alpha_{J}J^*(\omega = 0)}.
\end{equation}
\subsection{Well-posedness}\label{wpsd}
We know that the population cannot be negative and unbounded for $t\geq 0$. Hence it is essential to prove the non-negativity and boundedness of the solutions of system \eqref{model}. The following result assures this property:
\begin{theorem}\label{thm3.1}
	Any solution $(S,I,T_n,J,R)$ of system \eqref{model} with non-negative initial conditions is non-negative for all $t\geq 0$ and bounded in the following region
	\begin{equation}
	\Omega = \left\lbrace (S,I,T_n,J,R) \in \mathbb{R}_+^5 : 0< S+I+T_n+J+R \leq \frac{\Lambda}{\mu}  \right\rbrace.
	\end{equation}
\end{theorem}
\begin{proof}
	The proof is provided in \nameref{apndxa}.
\end{proof}

\section{Rigorous analysis}\label{rganl}
In this section, we establish the restrictions for the existence and stability of the possible equilibria of system \eqref{model}. In order to perform this, we first ponder the case $q=0$, i.e., the system \eqref{model} does not recruit newly diseased persons. The results achieved in this section will afterwards be adapted and utilized in subsection \ref{end_eq_qnotzero} to examine the existence of the equilibria of the system \eqref{model} with $q>0$.
\subsection{Threshold of the epidemic}\label{dfe}
When $q=0,$ the system \eqref{model} admits that all new individuals are only recruited into the susceptible population. In this case, the system \eqref{model} has a disease-free equilibrium (DFE) (i.e., $I(t) = T_n(t) = J(t) =0$) given by
\begin{equation*}
E_0 = (S_0, 0, 0,0, R_0) =\Bigg(\frac{\Lambda (\mu (1-p) + \xi)}{\mu (\mu + \xi)}, 0 ,0,0, \frac{p\Lambda}{\mu + \xi}\Bigg). 
\end{equation*}
It is convenient to write $N_0 = \frac{\Lambda}{\mu},$ $S_0 = (1-\chi)N_0,$ and $R_0 = \chi N_0$, where $\chi = \frac{p\mu}{\mu + \xi}$ is the fraction of the population recovered (due to having immunity against the disease) at the disease-free equilibrium. It is well-known that the stability of the DFE determines the threshold of an epidemic.

To determine the restrictions for the existence of endemic equilibria for the case $q=0,$ it is helpful to investigate the stability of $E_0$ and computation of the basic reproduction number. This analysis offers a fundamental threshold measure that will be utilized for stability analysis of the system \eqref{model} throughout the paper.
The basic reproduction number is characterized as the number of secondary infections made by a single infection throughout his/her whole infectious period. Mathematically, the basic reproduction number is expressed as a spectral radius $R_0$
that describes the number of novel infectious produced by a single infection in an entirely susceptible population (which is a threshold quantity for disease control) \cite{Driessche2002}. Let $x = (I,T_n, J)^T$, then system \eqref{model} can be rewritten in the following form 
\begin{equation*}
\frac{dx}{dt}= f(x)-v(x),
\end{equation*}
 where 
\begin{align*}
f(x)=\left( \begin{array}{c}
\beta S(I+T_n)\\ 
0\\ 
0
\end{array} \right),\quad v(x)= \left(\begin{array}{c}
(\mu + \omega + \delta_1)I \\ 
-\sigma \omega I + (\mu + \delta_2)T_n \\ 
-(1-\sigma) (1-\theta) \omega I + (\mu + \delta_3 + \rho)J
\end{array} \right).  
\end{align*} 
The Jacobian matrices of $f(x)$ and $v(x)$ at the disease-free equilibrium $(E_0)$ are given by
\begin{align*}
F = \left(\begin{array}{ccc}
\beta S_0& \beta S_0 & 0  \\ 
0&  0&0  \\ 
0& 0 &  0
\end{array}  \right), \quad
V = \left(\begin{array}{ccc}
\mu + \omega + \delta_1 & 0  & 0  \\ 
-\sigma \omega & \mu + \delta_2 & 0 \\ 
-(1-\sigma) (1-\theta) \omega & 0 & \mu + \delta_3 + \rho 
\end{array}  \right),
\end{align*}
respectively, where $F=Df(E_0)$ and $V=Dv(E_0)$. The next generation matrix of system \eqref{model} is given by $FV^{-1}.$ It follows that the spectral radius of matrix $FV^{-1}$ is 
\begin{align*}
\sigma(FV^{-1}) = \frac{\beta S_0}{(\mu+\omega + \delta_1)}\Bigg[1+\frac{\sigma \omega}{\mu + \delta_2}\Bigg].
\end{align*}
In consonance with Theorem 2 in \cite{Driessche2002}, the basic reproduction number for system \eqref{model} is 
\begin{align*}
R_0 = \sigma(FV^{-1}) = \frac{\beta \Lambda (\mu (1-p) + \xi)}{\mu (\mu+\xi)(\mu+\omega + \delta_1)}\Bigg[1+\frac{\sigma \omega}{\mu + \delta_2}\Bigg].
\end{align*}
\textbf{Interpretation of the basic reproduction number ($R_0$):} 
The infection rate into the susceptible individuals by infectious individuals (near the disease-free equilibrium) is [$\beta S_0$], and the average time spent by one individual in the infectious class ($I$) is [$\frac{1}{\mu + \omega + \delta_1}$]. The proportion of the infectious individuals that survived in the compartment ($T_n$) is [$\frac{\sigma \omega}{\mu + \delta_2}$]. The number of average secondary infection caused by the infected persons who are in class $I$ is [$\frac{\beta S_0}{\mu + \omega + \delta_1}$]. The number of average secondary infection caused by diseased persons who are in class $T_n$ is [$\frac{\beta S_0}{\mu + \omega + \delta_1} \frac{\sigma \omega}{\mu + \delta_2}$].

By setting $\sigma = 0$, the model \eqref{model} reduces to a SIJR model with perfect testing and isolation, where the basic reproduction number denoted by $R_0^{PT}$ as: $$R_0^{PT} = \frac{\beta S_0}{\mu + \omega + \delta_1} = \frac{\beta \Lambda (\mu (1-p) + \xi)}{\mu (\mu+\xi)(\mu+\omega + \delta_1)}.$$ 
Thus, $R_0$ can be express in terms of $R_0^{PT}$ as $$R_0 = R_0^{PT}\Big[1+\frac{\sigma \omega}{\mu + \delta_2}\Big].$$
\begin{remark}
Note that $R_0 \geq R_0^{PT}$ with equality only if $\sigma = 0$, i.e., imperfect testing always increases the basic reproduction number of the disease. 
\end{remark}
Additional simulations and arguments to expose that a decrease in $R_0$ usually infers a delayed and lessened peak caseload, and prevalence are summarized in Section \ref{5.5}.  

Since the isolation rate depends on the testing rate $(\omega)$ and test sensitivity $(1-\sigma)$. It is interesting to observe that if $\omega \rightarrow \infty$, then 
$R_0 \rightarrow \frac{\beta \Lambda (\mu (1-p) + \xi)}{\mu (\mu + \xi)}\frac{\sigma}{\mu + \delta_2}.$ When $\omega = 0$, we have a classical SIR model (without testing, i.e. $T_n = J = 0$) such that $R_0 = \frac{\beta \Lambda (\mu (1-p) + \xi)}{\mu (\mu + \xi) (\mu + \delta_1)}.$ Therefore, these are the two extreme values for the basic reproduction number $(R_0)$ when varying the testing rate $(\omega).$ 

Further, for the global dynamics near the disease-free equilibrium, we have the following results: 
\begin{theorem}\label{dfe_loc_stab}
	DFE $(E_0)$ of system \eqref{model} is locally asymptotically stable if $R_0 <1$ and unstable if $R_0>1$.
\end{theorem}
\begin{proof}
	The Jacobian matrix is evaluated at $E_0$ given by
	\begin{align*}
J_{E_0} =	\left( \begin{array}{ccccc}
	-\mu	& -\beta S_0 &  -\beta S_0 & 0 & \xi  \\ 
	0	& \beta S_0 - (\mu+\omega + \delta_1) & \beta S_0 & 0 & 0 \\ 
	0	& \sigma \omega  & -(\mu+\delta_2)  & 0 & 0 \\ 
	0	& (1-\sigma)(1-\theta)\omega  & 0 & -(\mu + \delta_3 + \rho) & 0 \\ 
	0	& 0 & 0 & \rho & -(\mu + \xi)
	\end{array} \right). 
 \end{align*}
 The eigenvalues of $J_{E_0}$ are 
 $\lambda_1= -\mu, \lambda_2 = -(\mu + \xi), \lambda_3 = -(\mu + \delta_3 + \rho),$ and the eigenvalues of the submatrix  
 \begin{align*}
 J_{E_0}^{'} = \left( \begin{array}{cc}
 \beta S_0 - (\mu+\omega + \delta_1) & \beta S_0 \\ 
 \sigma \omega  & -(\mu+\delta_2)
 \end{array} \right).
 \end{align*}
 The trace and determinant of $J_{E_0}^{'}$  are given by $Tr(J_{E_0}^{'})     = \beta S_0 - (\mu+\omega + \delta_1) - (\mu+\delta_2),$ and
 $Det(J_{E_0}^{'})     = (\mu+\delta_2)(\mu+\omega + \delta_1)(1-R_0)$, respectively. Clearly $\lambda_1 < 0, \lambda_2 < 0, \lambda_3 <0,$ and the eigenvalues of $J_{E_0}^{'}$ are negative, if $Tr(J_{E_0}^{'}) < 0$ and $Det(J_{E_0}^{'}) > 0$. Noting that $Tr(J_{E_0}^{'}) < 0$ and $Det(J_{E_0}^{'}) > 0$ if and only if $R_0 < 1,$ the proof is accomplished.
\end{proof}

\textbf{Biologically speaking,} Theorem \ref{dfe_loc_stab} infers that the infection can be removed from the population (when $R_0 < 1$) if the initial conditions of the system \eqref{model} lie in the basin of attraction of $E_0$. Though, this condition is insufficient for disease elimination since for random initial conditions of the system \eqref{model}, the local stability of $E_0$ does not assure community-wide elimination of the infection. It must be proven that $E_0$ is globally asymptotically stable (when $R_0 <1$) to eliminate the disease in the community. Global asymptotic stability (GAS) result ensures that the removal of disease is free from the initial sizes of the sub-populations of the model. This is established in the following theorem. 

\begin{theorem}
	The DFE ($E_0$) is globally asymptotically stable (GAS) whenever $R_0 <1$. 
\end{theorem} 
\begin{proof}
	In order to prove the global asymptotic stability of $E_0$, we go along with the method given in Castillo-Chavez et al. \cite{Castillo-Chavez}. The system \eqref{model} can be rewritten as follows 
	\begin{equation}
	\begin{aligned}
	\frac{dY}{dt} &= G(Y,Z),\\
	\frac{dZ}{dt} &= H(Y,Z), \quad H(Y,0)=0,
	\end{aligned}
	\end{equation}
	where $Y = (S,R)\in \mathbb{R}^2$ indicates the number of uninfected persons and $Z = (I, T_n, J) \in \mathbb{R}^3 $ represents the number of diseased individuals. DFE $(E_0)$ is globally asymptotically stable if the ensuing two conditions are followed: 
	\begin{itemize}
		\item[(H1)] For $\frac{dY}{dt} = G(Y,Z),$ $Y^*$ is globally asymptotically stable,
		\item[(H2)] $H(Y,Z) = MZ - \hat{H}(Y,Z),~ \hat{H}(Y,Z)>0$ for $(Y,Z)\in \Omega$,
	\end{itemize}
	where $M = D_Z H(Y^*, 0)$ is a $M$-matrix. For the system \eqref{model}, we have 
	\begin{equation}\label{model_9}
	G(Y,0) = \left(\begin{array}{c}
	(1-p)\Lambda - \mu S + \xi R\\ 
	p\Lambda - (\mu + \xi)R 
	\end{array}  \right). 
	\end{equation}
	It is evident that the equilibrium $Y^* = \left(\frac{\Lambda (\mu (1-p) + \xi)}{\mu (\mu + \xi)}, \frac{p\Lambda}{\mu + \xi} \right) $ of system \eqref{model_9} is globally asymptotically stable. Further, for system \eqref{model}, we obtain 
	\begin{align*}
	M &= \left(\begin{array}{cccc}
	\beta S_0 -(\mu + \omega + \delta_1) & \beta S_0 & 0 \\ 
	\sigma \omega & -(\mu + \delta_2) & 0 \\
	(1-\sigma)(1-\theta)\omega & 0 & -(\mu + \rho + \delta_3) 
	\end{array} \right), \\
	\hat{H}(Y,Z) &= \left(\begin{array}{c}
	\beta (I + T_n)(S_0 - S) \\
	0 \\
	0
	\end{array}  \right). 
	\end{align*}
	Clearly $\hat{H}(Y,Z) \geq 0$. Thus, $E_0$ is globally asymptotically stable, i.e. every trajectory in the region $\Omega$ approaches the DFE ($E_0$) as $t \rightarrow \infty$ for $R_0 <1.$ Hence, the infection will be eradicated from the population if $R_0 <1.$      
\end{proof}

Further, we see that $J_{E_0}$ has a zero eigenvalue when $R_0=1$. Consequently, the system \eqref{model} would possibly experience a transcritical bifurcation at $E_0$ when $R_0=1$. By using Theorem 4.1 from Castillo-Chavez and Song \cite{Castillo2004} and center manifold theory \cite{Gukenheimer1983}, we determine certain conditions on the parameters for the transcritical bifurcation. We prove the subsequent theorem:

\begin{theorem}\label{fortrans}
	System \eqref{model} experiences a transcritical bifurcation at $E_0$, when $R_0=1$.
\end{theorem}
\begin{proof}
	We consider $\beta$ as a bifurcation parameter. By manipulating $R_0=1,$ we have 
	\begin{equation*}
	\beta = \beta^* = \frac{(\mu + \omega + \delta_1)(\mu + \delta_2)}{S_0 (\mu + \delta_2 + \sigma \omega)}.
	\end{equation*}
	It can simply be achieved that the Jacobian $J_{(E_0, \beta^*)}$ calculated at $\beta = \beta^*$ and $E_0$ has a simple zero eigenvalue and other eigenvalues have negative real parts. Hence $E_0$ is a non-hyperbolic equilibrium for $\beta = \beta^*.$ Further, we compute a left eigenvector $V = (v_1, v_2, v_3, v_4, v_5)$ and a right eigenvector $W = (w_1, w_2, w_3, w_4, w_5)$ associated to the zero eigenvalue, where 
	\begin{align*}
	w_1 &= \frac{\xi}{\mu} - \frac{(\mu + \xi)(\mu + \omega + \delta_1)(\mu + \rho + \delta_3)}{\mu \rho (1-\sigma)(1-\theta)\omega} < 0,\quad
	w_2 = \frac{(\mu + \xi)(\mu + \rho + \delta_3)}{\rho (1-\sigma)(1-\theta)\omega}>0,\\
	w_3 & = \frac{\sigma (\mu + \xi)(\mu + \rho + \delta_3)}{\rho (1-\sigma)(1-\theta)(\mu + \delta_2)(\mu + \delta_2)}>0,\quad
	w_4 = \frac{\mu + \xi}{\rho}>0, \quad w_5 = 1, \\
	v_1 & = 0, \quad v_2 = \frac{\mu + \sigma \omega + \delta_2}{\mu + \omega + \delta_1}>0, \quad
	v_3 = 1,\quad v_4 = 0, \quad
	v_5 = 0.
	\end{align*}
	Furthermore, we require to evaluate the two bifurcation constants $a$ and $b$, as shown in Theorem 4.1 of \cite{Castillo2004} via the related non-zero partial derivatives of $f$ (calculated at $E_0$, $x_1 = S, x_2 = I, x_3 = T_n, x_4 = J, x_5 = R$), $a$ and $b$ are given by   
	\begin{align*}
	a &= 2v_2 w_1 w_2 \frac{\partial^{2}f_2}{\partial S \partial I} + 2v_2 w_1 w_3 \frac{\partial^{2}f_2}{\partial S \partial T_n} = 2 v_2 w_1(w_2 + w_3)\beta^* <0,\\
	b & = 2v_2 w_2 \frac{\partial^{2}f_2}{\partial \beta \partial I} + 2v_2 w_3 \frac{\partial^{2}f_2}{\partial \beta \partial T_n} = 2 v_2 S_0 (w_2 + w_3) > 0.
	\end{align*}
	Since the constants $b$ is positive and $a$ is negative, the system \eqref{model} follows a forward transcritical bifurcation at $\beta = \beta^*$.  
\end{proof}
\subsection{Different scenarios involving different parameters}\label{difsce}
From the above investigation, we realize that the system dynamics is governed by the basic reproduction number ($R_0$), and the infection dies out whenever $R_0 < 1$. We notice that if $R_0 = 1,$ then we have $$\sigma_{crit} \equiv \frac{(\delta _2+\mu)}{\omega }\Big(\frac{1}{R_0^{PT}}-1\Big).$$
	Since all the parameters in the model \eqref{model} are positive, it ensues that 
	$$\frac{dR_0}{d \sigma} = \frac{\beta  \Lambda  \omega  (\xi +\mu  (1-p))}{\mu  \left(\delta _2+\mu \right) (\mu +\xi ) \left(\delta _1+\mu +\omega \right)}>0.$$
	Thus, $R_0$ is always a continuous increasing function of $\sigma$ for $\sigma > 0,$ and if $\sigma < \sigma_{crit},$ then $R_0 < 1.$ It can be also observed that if $\sigma$ increases, i.e. the testing efficacy $(1-\sigma)$ decreases and thus $R_0$ increases with the decreasing testing efficacy and vice versa. Thus, the imperfect testing may be harmful to the community. It further indicates that model \eqref{model} has an endemic equilibrium for $\sigma > \sigma_{crit}$, i.e. $R_0 > 1,$ which is shown in Section \ref{endeq}. Thus, the above mentioned condition on $\sigma$ is also sufficient and necessary for disease control.
	
	Furthermore, from the following expression,
	$$\frac{dR_0}{d \xi} = \frac{\beta  \Lambda  p \left(\delta _2+\mu +\sigma  \omega \right)}{\left(\delta _2+\mu \right) (\mu +\xi )^2 \left(\delta _1+\mu +\omega \right)} > 0,$$ 
	it is easy to see that the waning immunity always increases the basic reproduction number $R_0$, as expected. Thus, the waning immunity results the detrimental consequences in the community. 
	
	In addition, we obtain that $R_0$ is a decreasing function of the testing rate $(\omega)$, i.e.,  
	$$\frac{dR_0}{d \omega} = -\frac{\beta  \Lambda  \left(\delta _2 -\delta _1 \sigma +\mu  (1-\sigma )\right) (\mu  (1-p)+\xi )}{\mu  \left(\delta _2+\mu \right) (\mu +\xi ) \left(\delta _1+\mu +\omega \right){}^2} < 0,$$ if and only if either $\delta_2 > \delta_1$ or $\delta_2 < \delta_1$ and $\sigma < \frac{\mu + \delta_2}{\mu + \delta_1}$, which implies testing of infected persons is beneficial in case: either (i) $\delta_2 > \delta_1$, or (ii) $\delta_2 < \delta_1$ and $\sigma < \frac{\mu + \delta_2}{\mu + \delta_1}$.

\subsection{Endemic equilibria ($q=0$)}\label{endeq}
The endemic equilibria of the system \eqref{model} with $q=0$ can clearly be computed in closed form. To find the certain conditions for the existence of the equilibrium point, we utilize the equations of the right-hand side of system \eqref{model} to express the variables in terms of the parameters. It gives the equilibrium $E^* (S^*, I^*, T_n^*, J^*, R^*),$ where

\begin{equation}
\begin{aligned}\label{end_eq}
S^* &= \frac{\Lambda}{\mu}\Big(\frac{(1-p)\mu + \xi}{\mu + \xi} \Big) - \frac{1}{\mu}(\mu + \omega + \delta_1) \Big(\frac{R_0 -1}{1+\frac{\sigma \omega}{\mu + \delta_2}} \Big)I^*, \\
I^* & = \frac{R_0 -1}{\Big(1+\frac{\sigma \omega}{\mu + \delta_2} \Big) \Big(1- \frac{\rho \xi (1-\sigma)(1-\theta)\omega}{(\mu + \delta_3 + \rho)(\mu + \xi)(\mu + \omega + \delta_1)} \Big)}, \quad
T_n^* = \frac{\sigma \omega I^*}{(\mu + \delta_2)},\\
J^* &= \frac{(1-\sigma)(1-\theta) \omega I^*}{(\mu + \delta_3 + \rho)}, \quad R^* = \frac{p \Lambda}{(\mu + \xi)} + \frac{\rho (1-\sigma)(1-\theta)\omega I^*}{(\mu + \delta_3 + \rho)(\mu + \xi)}. 
\end{aligned}
\end{equation}
The above expressions represent that $E^*$ exists if and only if $R_0>1$; otherwise, there does not exist any positive equilibria. 
Consequently, $R_0$ stands for a threshold value for the existence of endemic equilibrium of the system \eqref{model}. Additionally, the basic reproduction number $(R_0)$ is a well-known measure that provides the information about average number of secondary infections caused by a single infectious person in the entire susceptible populace. Consequently, if $R_0 < 1$, each infectious person will generate less than one diseased persons on average in the total infectious time duration, implying that the infection will fade out. Though, if $R_0 > 1,$ each infectious person in the total infectious time duration would produce more than one diseased individuals; this indicates to the disease persisting in the community.  

\subsection{Endemic equilibria ($q \neq 0$)}\label{end_eq_qnotzero}
Since diseases like influenza and COVID-19 could also be commenced into the populace by immigration of diseased persons, it is more practical to count $q > 0$ in the system \eqref{model}. Mathematically speaking, if immigration/recruitment of diseases persons is admitted, the system \eqref{model} does not allow to exist a DFE, and eradicating the infection is impossible. In this scenario, the public health aim is to curtail the level of endemicity.

Equilibrium point $E^*(S^*, I^*, T_n^*, J^*, R^*)$ for $q \neq 0,$ for which the disease is endemic in the community, the expressions of $S^*, T_n^*, J^*, R^*$ remain same as defined in \eqref{end_eq}. However, their numerical values change with the corresponding component $I^*$, and $I^*$ can be obtained from the roots of the following quadratic equation: 
\begin{equation}\label{endeq_qnot0}
P(I,q) = A_1 I^2 + A_2 I + A_3 = 0,
\end{equation}
where 
\begin{align*}
A_1 = & \frac{\beta}{\mu}\Bigg(\frac{\rho \xi (1-\sigma) (1-\theta) \omega}{(\mu + \delta_3 + \rho)(\mu + \xi)(\mu + \omega + \delta_1)} -1 \Bigg) \Bigg(1 + \frac{\sigma \omega}{\mu + \delta_2} \Bigg)(\mu + \omega + \delta_1) < 0, \\
A_2 = & (\mu + \omega + \delta_1)(R_0 - 1), \quad A_3 = q \Lambda > 0.
\end{align*}
Note that for $q=0,$ $P(I,q)$ has roots which corresponds to equilibria given  in \eqref{end_eq} (along with $I = 0$). The negative equilibria are feasibly unrealistic, the restrictions for $P(I,q)$ with $q \neq 0$ to have positive real solutions are examined. 

Clearly $A_1 <0$, $A_3 >0$ and so the quadratic $P(I,q)$ is concave down and the vertical intercept of $P(I,q)$ is positive. It implies that $P(I,q)$ always has two real roots with opposite signs. However, to determine expression of the roots $P(I,q)$, we have following cases: 
\begin{itemize}
	\item[Case 1.] Assume $R_0 > 1$. Then $A_2 > 0$ and expression of the roots of $P(I,q)$ are given by 
	\begin{equation}
	I^*_{1,2} = \frac{-A_2 \pm \sqrt{A_2^2 - 4 A_1 A_3}}{4 A_1}.
	\end{equation} 
	It can be clearly seen that $I^*_1$ is always positive and $I^*_2$ is negative, i.e.
	$$I^*_1  = \frac{-A_2 + \sqrt{A_2^2 - 4 A_1 A_3}}{4 A_1} > 0 \quad \text{and} \quad I^*_2  = \frac{-A_2 - \sqrt{A_2^2 - 4 A_1 A_3}}{4 A_1} < 0.$$
	\item[Case 2.] Assume $R_0 = 1$. Then $A_2 = 0$ and the quadratic equation $P(I,q) = A_1 I^2 + A_3,$ with roots $I^* = \pm \frac{A_3}{A_1}$. However, it is clear that $A_1 < 1$, then $I^* = - \frac{A_3}{A_1}$ is a positive real root. Thus, the model \eqref{model} has a unique positive equilibrium when $R_0 = 1.$
	\item[Case 3.] Suppose $R_0 < 1$. Then $A_2 < 0$ and expression of the roots of $P(I,q)$ are given by 
	\begin{equation}
	I^*_{1,2} = \frac{-A_2 \pm \sqrt{A_2^2 - 4 A_1 A_3}}{4 A_1}.
	\end{equation} 
	It can be clearly seen that $I^*_1$ is always positive and $I^*_2$ is negative, i.e.
	$$I^*_1  = \frac{-A_2 + \sqrt{A_2^2 - 4 A_1 A_3}}{4 A_1} > 0 \quad \text{and} \quad I^*_2  = \frac{-A_2 - \sqrt{A_2^2 - 4 A_1 A_3}}{4 A_1} < 0.$$
\end{itemize}
 This implies that Eq. \eqref{endeq_qnot0} always has a unique positive solution. Thus, a unique positive equilibrium exists for the case $q \neq 0$. The above discussion briefs in the following result:
 \begin{proposition}\label{pre}
 	The model \eqref{model} always has a unique endemic equilibrium for $q \neq 0$.
 \end{proposition}
We now claim the global asymptotic stability of the unique endemic equilibrium in the following result:
\begin{theorem}\label{gbstabend}
 	The endemic equilibrium of system \eqref{model} is globally asymptotically stable whenever it exists.
 \end{theorem}
\begin{proof}
	The proof follows by constructing a suitable Lyapunov function and is provided in \nameref{apndxa}. The case $q=0$, could be dealt similarly.
\end{proof}

To observe the equilibria for the model \eqref{model}, let $\Lambda = 10^6$, $\mu=0.0002$, $p=0.1$, $\xi=0.05$ and other values from Table \ref{fitval}. 
By choosing the above hypothetical parameter values and MATLAB R2021a, we can describe the forward transcritical bifurcation (given in Theorem \ref{fortrans}) diagram for model \eqref{model} at $R_0=1$ for the case $q=0$ (see, Figure \ref{eqlibrm}(a)). By plotting $I^*$ with respect to $R_0$, it is clear that there is one threshold $R_0 =1$ at which forward transcritical bifurcation occurs. In the region $R_0<1,$ only DFE exists and in $R_0>1$, a unique endemic equilibrium exists. 

Let all parameters values be same as in Figure \ref{eqlibrm}(a), except $q$. For five different values of $q$, the endemic equilibrium is plotted in Figure \ref{eqlibrm}(b). It can be seen that the threshold $R_0 =1$ does not work for $q \neq 0.$ The model \eqref{model} always has a unique endemic equilibrium for $q \neq 0$ in the whole region $R_0<1$ as well as $R_0>1$. Moreover, the endemic level increases with the increasing value of $q$.  

\begin{figure}[H]
	(a)\includegraphics[scale=0.5]{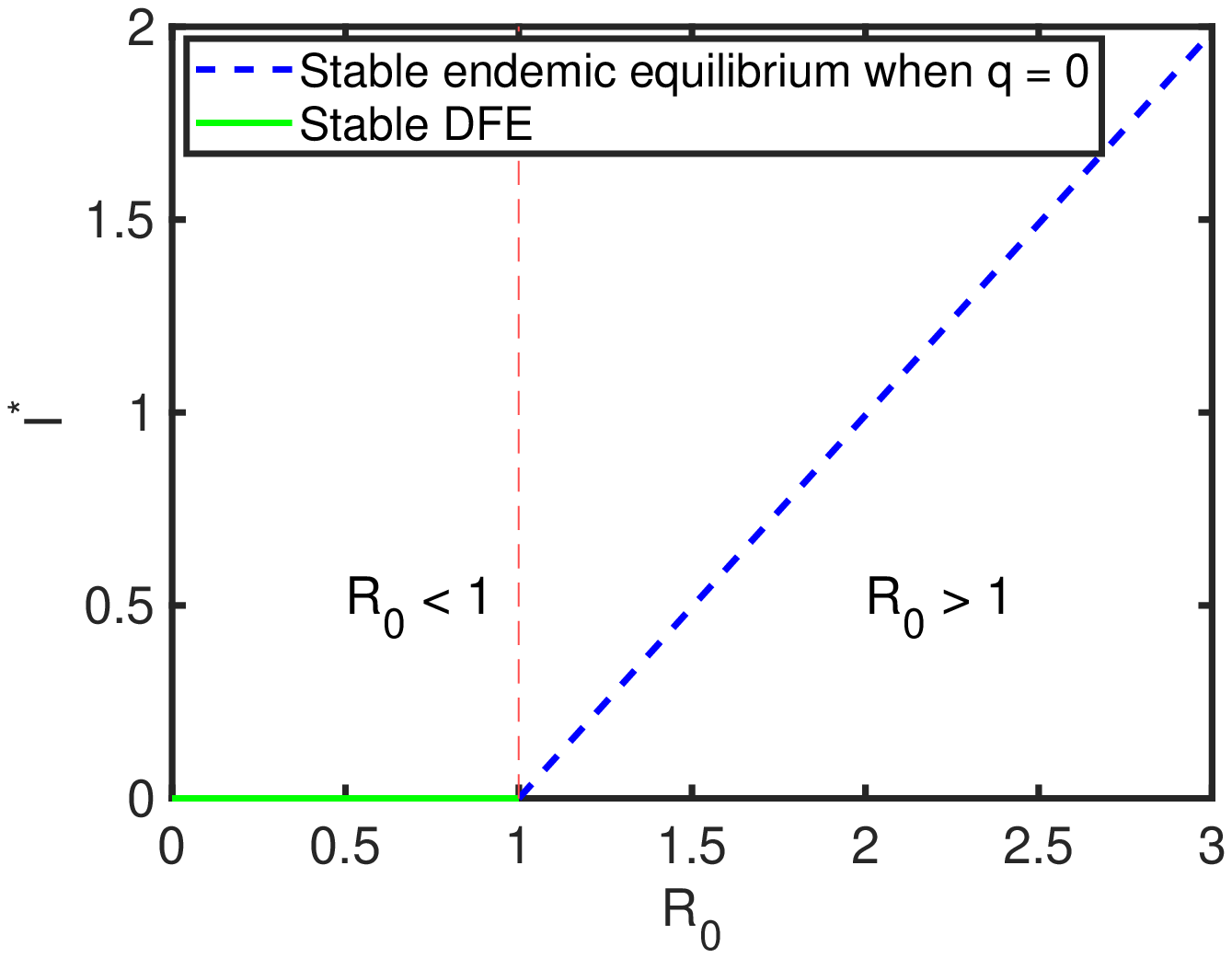}
	(b)\includegraphics[scale=0.5]{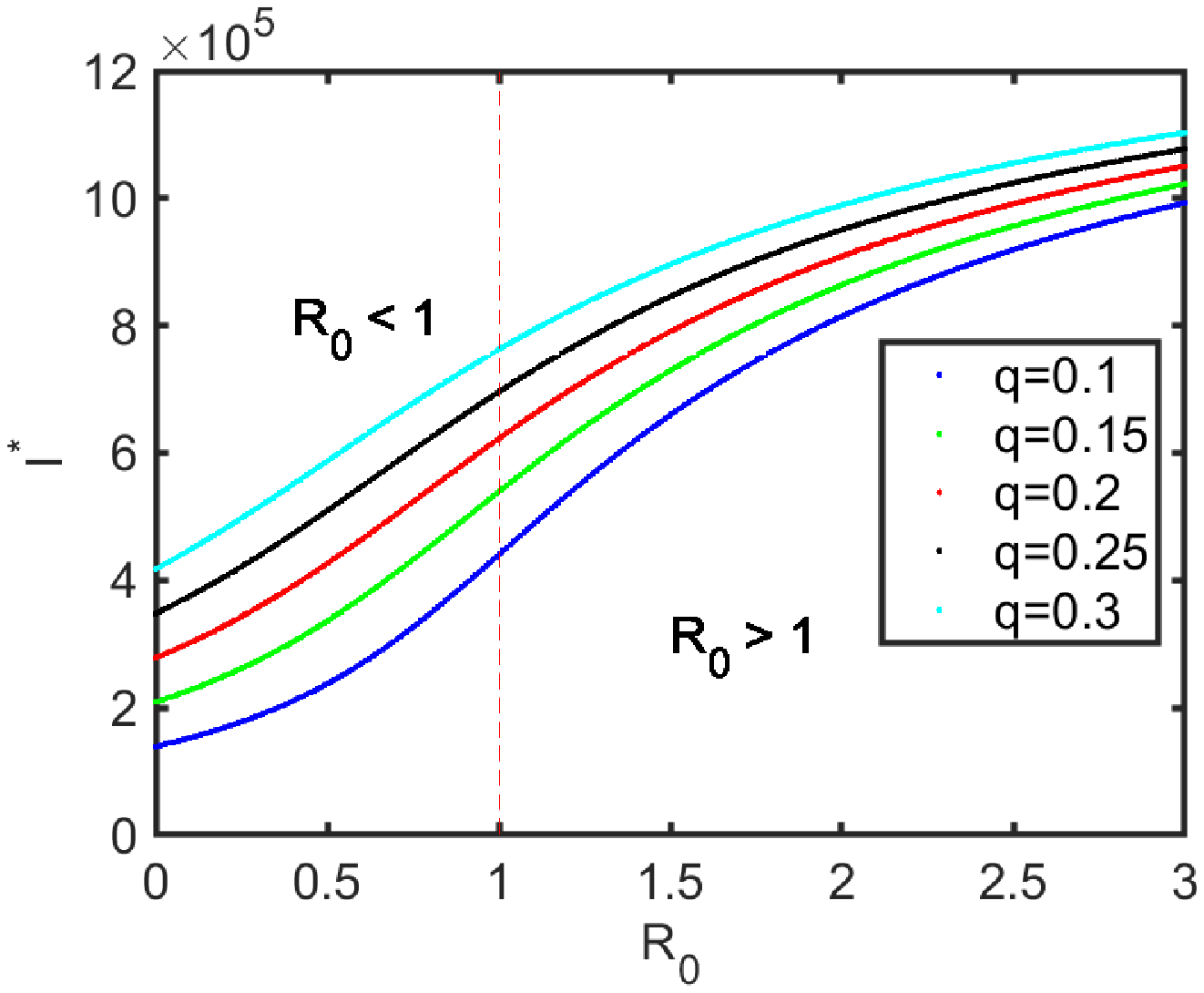}
	\caption{The figure depicts the equilibria of model \eqref{model} for both the cases $q=0$ and $q \neq 0$. (a) Forward transcritical bifurcation diagram shows the existence of the DFE for $R_0 < 1$ (green line) and a unique endemic equilibrium $R_0 > 1$ (blue dashed curve) for the case $q=0$. (b) A unique endemic equilibrium for the case $q \neq 0$, different colored curves show the level of endemicity for different values of the parameter  $q$.}\label{eqlibrm}
\end{figure}

\section{The peak and final size relation of an epidemic}\label{peakandfinal}
The final size relations of an epidemic or outbreak are biological quantities allied with the dynamics of epidemic models (without demographic effects), permit for the accurate quantification of infection load in the community and can be benefitted to evaluate the effectiveness and impact of numerous mitigation and interventions strategies, for instance, the SARS epidemic in 2002-2004, COVID-19. An epidemic model without any recruitment of individuals (from either births/immigration or loss of immunity) can be benefited to designate the short-term disease transmissions with a short period of infection and permanent immunity (e.g., a specific strain of SARS-CoV-2 or influenza). In these cases, births and deaths may be ignored because of the short time period. For a novel disease, loss of immunity may also be neglected because of the interest in the initial phase of the epidemic, at that time the quantity of infected persons is small. The above factors are also not evident at the single outbreak disease scale. The final size relations are relations comprising the number of the population remain in each disease-free class throughout the epidemic and the basic reproduction number. 

The original model \eqref{model} with no demographic effect, i.e. $\Lambda = \mu = 0$, and without loss of immunity ($\xi = 0$) reduces to the following: 
\begin{equation} \label{modeloutbreak}
\begin{aligned}
\frac{dS}{dt} &= - \beta S(I+T_n), \\
\frac{dI}{dt} &=  \beta S(I+T_n) - (\omega + \delta_1)I, \\
\frac{dT_n}{dt} & = \sigma \omega I - \delta_2 T_n, \\
\frac{dJ}{dt} & = (1- \sigma) (1-\theta) \omega I - (\delta_3 + \rho)J, \\
\frac{dR}{dt} & = \rho J,
\end{aligned}
\end{equation}
with initial conditions $S(0)>0, I(0)\geq 0, T_n(0)\geq 0, J(0)\geq 0, R(0)\geq 0$. 

The model \eqref{modeloutbreak} has a disease-free equilibrium given by $(N_0, 0,0,0,0)$, where $N_0 \sim N(0)$ is the initial size of total susceptible population in the absence of disease.

Using the notations from Arino et al. \cite{Arino2007}, let $x \in \mathbb{R}^3$, $y \in \mathbb{R}^1$, and $z \in \mathbb{R}^1$ represent the set of diseased compartments, the set of susceptible compartments, and the set of compartments removed from the disease, respectively. Therefore, it ensues from the model \eqref{modeloutbreak}, that $x=(I, T_n, J)^T$, $y=S$, and $z=R$. Further, let $D$ be the $m \times m$ diagonal matrix whose diagonal entries represent the relative susceptibilities of the associated susceptible compartments. It is appropriate to specify $\Pi$ to be an $n \times m$ matrix with the feature that the $(i,j)$ entry signifies the fraction of the $j^{th}$ susceptible class that moves into the $i^{th}$ diseased class upon getting infection. Let $b$ be an $n$-dimensional row vector of relative horizontal transmissions. It follows, in the context of the model \eqref{modeloutbreak}, that 
\begin{align*}
D=1, \quad \Pi = \left[\begin{array}{c}
1 \\ 
0 \\ 
0
\end{array}  \right], \quad b = \left[\begin{array}{ccc}
1& 1 & 0
\end{array}  \right]. 
\end{align*}
Using the above notations, definitions and variables, the system \eqref{modeloutbreak} rewrites to:
\begin{equation}\label{eqmodel}
\begin{aligned}
\frac{dx}{dt} &= \Pi Dy\beta bx - Vx,\\
\frac{dy}{dt} &= -\Pi Dy\beta bx, \\
\frac{dz}{dt} &= Wx,
\end{aligned}
\end{equation}
where $W$ is a $k \times n$ matrix with the characteristic that the entry $(i,j)$ denotes the rate at which individuals of the $j^{th}$ diseased class
moves into the removed ($i^{th}$) class upon recovery and the matrix $V$ is same as described in Section \ref{dfe} (with $\mu = 0$). It is significant saying that the basic reproduction number, $R_{01}$, (of the system \eqref{modeloutbreak} or equivalently \eqref{eqmodel}) could be obtained utilizing the definition $R_0 = \beta(0,y_0,z_0)bV^{-1}\Pi Dy_0$ provided in Theorem 2.1 of \cite{Arino2007}. It must be noticed that this theorem also claims the local asymptotic stability for the disease-free equilibrium of the system \eqref{modeloutbreak}. The basic reproduction number, $R_{01}$, for system \eqref{modeloutbreak} is given by 
\begin{equation}
R_{01} = \frac{\beta N_0}{(\omega + \delta_1)}\Bigg[1+\frac{\sigma \omega}{\delta_2}\Bigg].
\end{equation}

\subsection{Peak size relation} 
To investigate the peak size of an epidemic, we use the approach given in Feng \cite{FEN2007}. We define a weighted sum of disease variables given by $Y(t) = \frac{1}{R_{01}}\beta b V^{-1}x$, which yields 
\begin{equation}\label{wght}
Y(t) =  I + \frac{\delta_1 + \omega}{(\delta_2 + \sigma \omega)} T_n.
\end{equation}
The infected compartments $I(t)$ and $T_n(t)$ are only counted since they participate in the disease spread. Further, by differentiating \eqref{wght} with respect to the time $t$, we obtain
\begin{equation}
\frac{d Y(t)}{dt} = \frac{d I}{dt} + \frac{\delta_1 + \omega}{(\delta_2 + \sigma \omega)} \frac{dT_n}{dt},
\end{equation}
and substituting $\frac{d I}{dt}$ and $\frac{dT_n}{dt}$ from the system \eqref{modeloutbreak}, yields
\begin{equation}
\frac{d Y(t)}{dt} = \beta S\Big(1-\frac{1}{S R_{01}}\Big)(I+T_n).
\end{equation}
Further, we have
\begin{equation}\label{dyds}
\frac{d Y}{dS} = -\Big(1-\frac{1}{S R_{01}}\Big).
\end{equation}
Integrating \eqref{dyds} and using initial conditions $S(0) = S_0$ and $Y(0) = Y_0$ yields 
\begin{equation}
Y + S - \frac{\ln (S)}{R_{01}} = Y_0 + S_0 - \frac{\ln (S_0)}{R_{01}},
\end{equation}
where $Y_0 = I_0 + \frac{\delta_1 + \omega}{(\delta_2 + \sigma \omega)} T_{n_0}$. The maximum value of $Y(t)$ at any time $t$ is the number of infectives when $\frac{dY}{dt} = 0,$ i.e. when $S =  \frac{1}{R_{01}}$. It is specified by  
\begin{equation}\label{pksz}
Y_{max} = Y_0 + S_0 - \frac{1}{R_{01}} + \frac{\ln (\frac{1}{R_{01}})}{R_{01}} - \frac{\ln (S_0)}{R_{01}}.
\end{equation}
Hence Eq. \eqref{pksz} provides the peak size of an epidemic.

\subsection{Final size relation}
To analyze the final size of an epidemic, we consider $N = S_0$ and $S_{\infty}$ to be a non-negative smooth decreasing function that approaches to a limit as $t \rightarrow \infty$, i.e. $S_{\infty} > 0.$ Utilizing the method \cite{BRA2019}, let $I_{\infty} \rightarrow 0, T_{n_{\infty}} \rightarrow 0$ and $J_{\infty}\rightarrow 0$.

By addition of the first two equations of system \eqref{modeloutbreak}, we acquire
\begin{equation*}
S_{\infty} + I_{\infty} - S_0 - I_0 = -(\omega + \delta_1)\int_{0}^{\infty} I(s)ds,
\end{equation*}
but $S_0 = N, I_{\infty} = 0$ gives
\begin{equation}\label{eq20}
\int_{0}^{\infty} I(s)ds = \frac{N - S_{\infty}}{(\omega + \delta_1)} + \frac{I_0}{(\omega + \delta_1)}.
\end{equation}
Now integrating the third equation of system \eqref{modeloutbreak}, we obtain
\begin{equation*}
T_{n_{\infty}} - T_{n_{0}} = \sigma \omega \int_{0}^{\infty} I(s)ds - \delta_2 \int_{0}^{\infty} T_n(s)ds,
\end{equation*}
which yields
\begin{equation}\label{eq21}
\int_{0}^{\infty} T_n(s)ds = \frac{\sigma \omega}{\delta_2} \frac{N - S_{\infty}}{(\omega + \delta_1)} + \frac{\sigma \omega}{\delta_2} \frac{I_0}{(\omega + \delta_1)} + \frac{1}{\delta_2} T_{n_{0}}.
\end{equation}
Similarly, we have
\begin{equation}\label{eq22}
\int_{0}^{\infty} J(s)ds = \frac{(1-\sigma)(1-\theta)\omega}{(\delta_3 + \rho)(\omega + \delta_1)}(N - S_{\infty}) + \frac{(1-\sigma)(1-\theta)\omega}{(\delta_3 + \rho)(\omega + \delta_1)} I_0 + \frac{1}{(\delta_3 + \rho)} J_0.
\end{equation}
From the first equation of system \eqref{modeloutbreak}, we have
\begin{equation}\label{eq23}
\frac{1}{S}\frac{dS}{dt} = - \beta (I + T_n).
\end{equation}
Therefore, integrating Eq. \eqref{eq23} on $[0, \infty)$ gives
\begin{equation}
\ln S_{\infty} - \ln S_0 = -\beta \Big[\int_{0}^{\infty}I(s)ds + \int_{0}^{\infty}T_n(s)ds \Big],
\end{equation}
and substituting the values from \eqref{eq20}-\eqref{eq22} into \eqref{eq23} and simplifying gives 
\begin{equation}\label{eq25}
\ln \Big(\frac{S_{\infty}}{S_0}\Big) = -R_{01}(N - S_{\infty}) - R_{01} Y_0,
\end{equation}
where $Y_0 =  I_0 + \frac{\delta_1 + \omega}{(\delta_2 + \sigma \omega)} T_{n_0}.$ Therefore, Eq. \eqref{eq25} provides the final size relation with the initial infected populations $I_0$ and $T_{n_0}$. If the initially infected individuals are assumed to be zero, i.e., $I_0=T_{n_0}=J_0=0,$ and if a less number of infected individuals are introduced into the community then we have $S_0 \sim N$ such that the final size relation becomes in the form 
\begin{equation}
\ln \Big(\frac{S_{\infty}}{S_0}\Big) = -R_{01}(N - S_{\infty}),
\end{equation}
which yields
\begin{equation}
S_{\infty} = S_0 e^{-R_{01}N\Big(1 - \frac{S_{\infty}}{N}\Big)},
\end{equation}
thus, $\Big(1 - \frac{S_{\infty}}{N}\Big)$ is the clinical attack rate and $S_0 - S_{\infty}$ denotes the epidemic size, i.e. the number of infected individuals throughout the course of the epidemic.  

%Furthermore, the final size relation of the model \eqref{modeloutbreak} (or, equivalently, \eqref{eqmodel}) can be written directly from Theorem 5.1 in \cite{Arino2007} as follows:
%\begin{equation*}
%\ln \Bigg[ \frac{S(0)}{S_\infty}\Bigg] = \frac{R_0^{'}}{S(0)}\Big[S(0)-S_\infty\Big] + \beta b V^{-1}x(0),
%\end{equation*}
%which yields
%\begin{equation}
%\ln \Bigg[ \frac{S(0)}{S_\infty}\Bigg] = \frac{R_0^{'}}{S(0)}\Big[S(0)-S_\infty\Big] + \frac{R_0^{'}}{N(0)} I(0) + \frac{\beta}{\delta_2}T_n(0) + \frac{\beta}{(\rho + \delta_3)}J(0),
%\end{equation}
%with $I(0)>0, T_n(0)>0,$ and $J(0)>0.$ 
%
%If we consider the initial conditions $I(0)>0, T_n(0) = J(0) = 0$, then the final size relation given by 
%\begin{equation}
%\ln \Bigg[ \frac{S(0)}{S_\infty}\Bigg] = \frac{R_0^{'}}{S(0)}\Big[S(0)-S_\infty\Big] + \frac{R_0^{'}}{N(0)} I(0). 
%\end{equation}

\subsection{Infection and Case Fatality Rates}
Here, we compute infection fatality rate (IFR) and case fatality rate (CFR) based on the model \eqref{modeloutbreak}. Infection fatality rate (IFR) depends on the total infected population, i.e., involving the negative tested and asymptomatic individuals. In terms of the fatality and recovery rates, we have
\begin{equation}\label{ifr1}
\text{IFR} = \frac{D_{\infty}}{D_{\infty} + R_{\infty}}.100 (\%),
\end{equation}
where $D_{\infty}$ and $R_{\infty}$ signify the dead and recovered individuals at end of the epidemic ($t \rightarrow \infty$), respectively. The number of infected individuals is a sum of dead and recovered individuals. It can straightforwardly be revealed that using the equations of model \eqref{modeloutbreak}, we obtain 
\begin{equation}\label{ifr2}
\text{IFR} = \frac{\delta_1 I_{\infty} + \delta_2 T_n + \delta_3 J_{\infty}}{\delta_1 I_{\infty} + \delta_2 T_n + \delta_3 J_{\infty} + \rho J_{\infty}}.100 (\%).
\end{equation}
Eq. \eqref{ifr2} holds the infection fatality rate (IFR) at the epidemic's end. In addition, the case fatality rate (CFR) contains the number of deaths related to the diagnosed individuals, and IFR can not exceed the CFR since the number of undetected cases is added to its denominator.

\subsection{Parameter estimation and practical identifiability} 
First, we neglect the demographic effect in the model ($\Lambda = 0$ and $\mu = 0$) and waning immunity rate ($\xi = 0$) because we collect the data of early COVID-19 outbreak in India for 55 weeks only. The data of COVID-19 in India has been collected from 27 January 2020 to 8 February 2021 from the official website of the World Health Organization (WHO) \cite{whocovid}. We grasp some parameter values from the literature, such as disease-induced death rate, it is necessary because we do not fit the death cases for the model. The recovery rate is fixed as 1/2 because the average recovery time of infected individuals is two weeks \cite{WANG2022}. However, sometimes it varies between 1-4 weeks. The testing efficacy is assumed 71\% \cite{FANG2020}. The testing return time is assumed to be one week for the early outbreak; it may also vary from 1-10 days. However, it depends on the type of tests and is also improved for further outbreaks of COVID-19. The total population size of India is 138e+07. The initial size of susceptible population is given by $S(0) = N(0) - I(0) -T_{n}(0) - J(0) - R(0).$ 

We use the maximum likelihood method \cite{EIS2014} and MATLAB function \textit{fminsearch} to estimate the parameters. The Figures \ref{figfit}(a) and \ref{figfit}(b) show the weekly new cases and cumulative cases of the early outbreak of COVID-19 in India and the model output curves. The residuals in Figure \ref{figfit}(c) reveal that the model captures the data very well.  
The parameter values and initial values are given in Table \ref{fitval} and \ref{estmdini}, respectively. From the data and model output, we observe that the epidemic peak at approximately 35 weeks and a level near infected individuals. 

%In the next section \ref{5.5}, we observe the impact of parameters on the peak size either it would be delayed or reduced. 

We further analyze the practical identifiability of the model using Fisher Information Matrix (FIM) and profile likelihoods \cite{EIS2014, BRO2017, KAO2018}. The FIM method is a numerical approach to analyzing the identifiability of the model. The FIM is evaluated by $F = Y^T Y$, where $Y$ is the sensitivity matrix described by $Y_{(i,j)} = \frac{\partial x}{\partial q_i}(t_j)$ for parameters $q_1, q_2, \cdots , q_n$ and time points $t_1, t_2, \cdots, t_m$. The above form is an abridged form of the FIM for normally distributed measurement error. The rank of $F$ reveals the number of identifiable parameters and parameter combinations. The model is structurally unidentifiable if $F$ is singular; the model may be practically unidentifiable if $F$ is close to singular. The FIM also determines which parameter subsets are unidentifiable or identifiable. To maximize a likelihood function, we further compute profile likelihoods for the estimated parameters. The likelihood is defined by letting a constant, normally distributed measurement error with a standard deviation equal to 10\% of the mean of the data and mean equal to the model trajectory. Maximizing the likelihood function is similar to minimizing a cost function based on the negative log likelihood identical to least squares. The profile likelihood for the parameters is identified by the maximum values of the likelihood function across the range of the parameter values. The parameter is structurally unidentifiable if the profile likelihood is flat; it is practically unidentifiable if the profile likelihood is sufficiently shallow. However, structural identifiability is an analytical method, but FIM gives the idea of locally structural identifiability. For more details and a deep understanding of the profile likelihood, one can refer to the paper Raue et al. \cite{RAU2019}.  

While estimating the parameters and computing FIM for our model \eqref{modeloutbreak}, we find that the rank of FIM is two, inferring that there should be two identifiable parameters and parameter combinations, which also implies that the model is locally structurally identifiable. From profile likelihoods in Figure \ref{costfun}, it is clear that the profile likelihoods are not flattened and shallow around the estimated parameters; the bowl-shaped curves imply that the parameters $\beta$ and $I_0$ are practically identifiable. The red dashed lines are the thresholds for the approximate 95\% confidence bound of the profile likelihood. The profile likelihood curves of identifiable parameters should cross the thresholds on either side of the minimum (red dot), and the parameter values where they cross would be the confidence bounds. The Figure \ref{costfun}(a) and (b) show that the 95\% confidence bound for the parameters $\beta$ and $I_0$ are [1.012e-09, 1.145e-09] and [1.677e+04, 2.625e+04], respectively.  
 
\begin{table}[H]
	\caption{Estimated values of the model parameters \eqref{modeloutbreak}.}\label{fitval}
	\begin{tabular}{p{3cm}p{3cm}p{3cm}p{3cm}}
		\hline
		Parameters	& Fitted values & Units & References \\
		\hline
		
		$\beta$	& 1.0860e-09 & Per week & Estimated \\
		
		$\delta_1 = \delta_2 = \delta_3$	&  0.7 & Per week & \cite{BUG2020} \\
		
		$\omega$ & 1 & Per week & Assumed \\
		
		$\sigma$ & 0.29 & Dimensionless & \cite{FANG2020}  \\
		
		$\theta$ & 0.95 & Dimensionless & Assumed \\
		
		$1/\rho$ & 2 & Per week & \cite{WANG2022}  \\
		\hline
	\end{tabular}
\end{table}

\begin{table}[H]
	\caption{Initial conditions with respect to the model \eqref{modeloutbreak}.}\label{estmdini}
	\begin{tabular}{ccc}
		\hline 
		Initial conditions & Values & References \\ 
		\hline
	$S(0)$	& 137e+07 & $N(0) - I(0) - T_n(0) - J(0)$ \\ 
	$I(0)$ & 2.1296e+04 & Estimated \\ 
	$T_n(0)$  & e+04 & Assumed \\ 
	$J(0)$  & 2 & Data \\ 
	$R(0)$ & 0 & Data \\ 
		\hline 
	\end{tabular} 
\end{table}

For the baseline parameter values given in Table \ref{fitval}, the value of the basic reproduction number is computed as $R_{01} = 1.4378>1.$

\begin{figure}[H]
	(a)\includegraphics[scale=0.4]{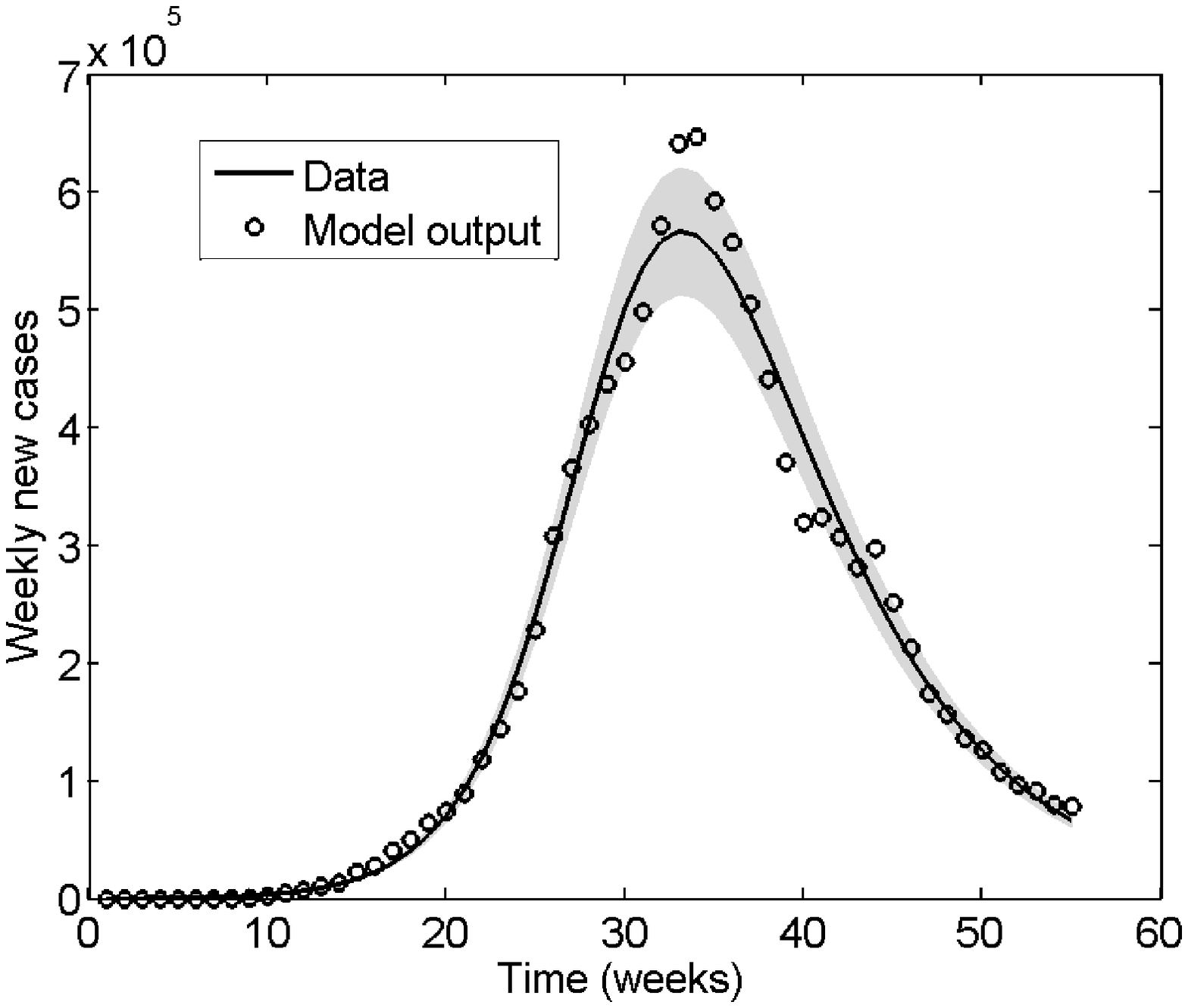}
	(b)\includegraphics[scale=0.4]{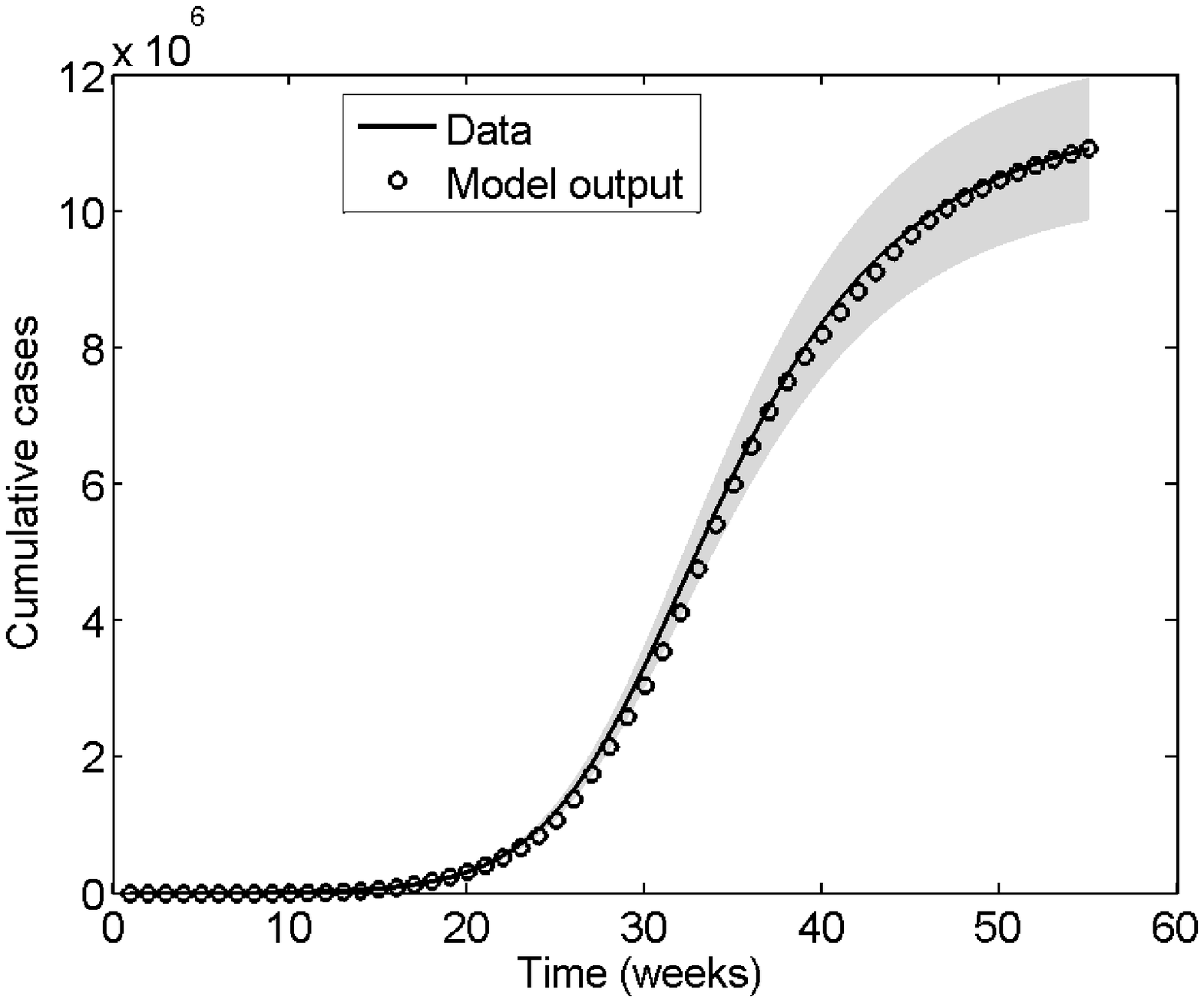}
	(c)\includegraphics[scale=0.5]{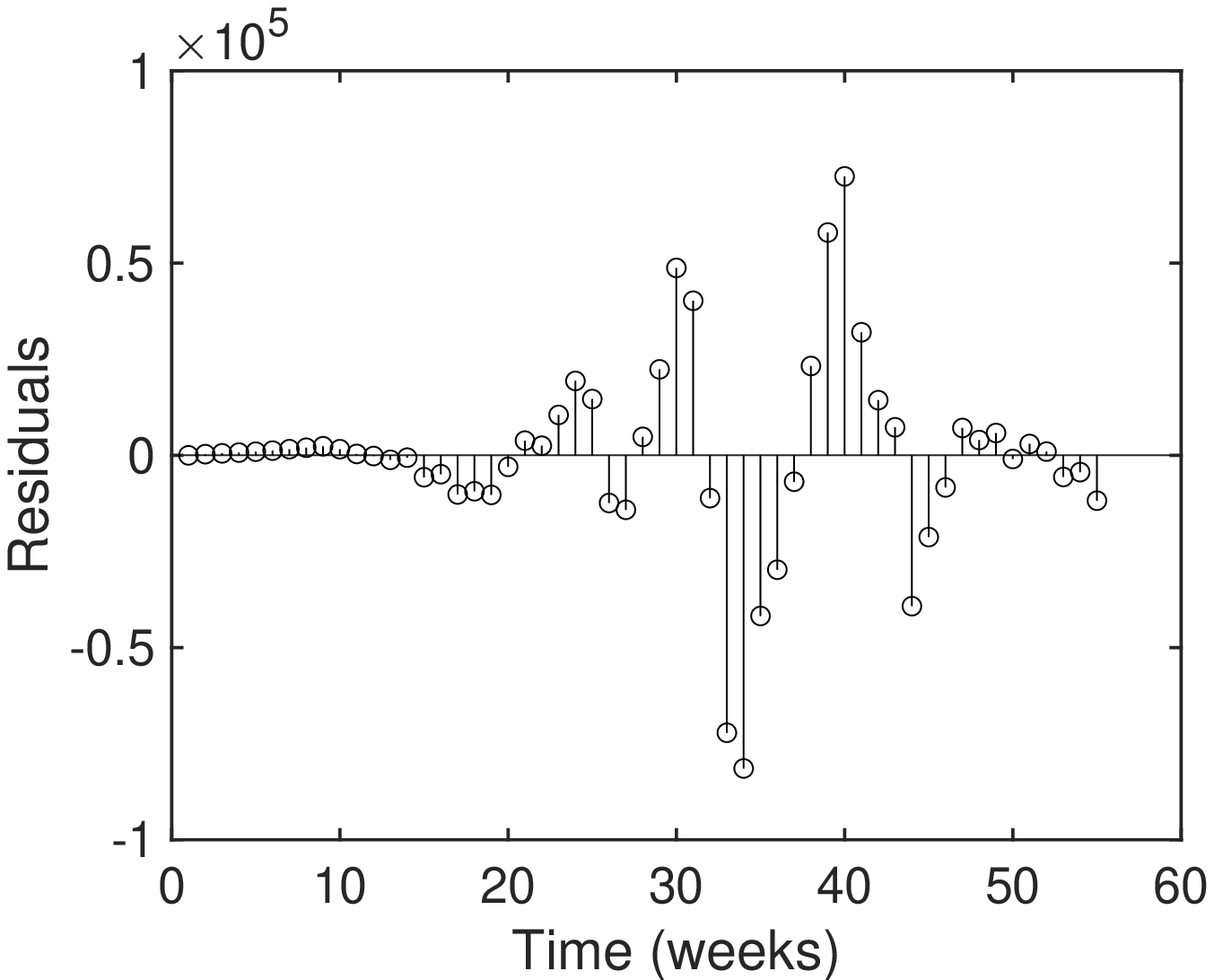}
	\caption{Model \eqref{modeloutbreak} is fitted to weekly reported and cumulative cases of the early COVID-19 outbreak in India. (a) The black circles show the data of weekly new cases, and the black curve demonstrates the model output. (b) The black circles show the data of cumulative cases, and the black curve illustrates the model output. (c) Residuals of the fit.}\label{figfit}
\end{figure}

\begin{figure}[H]
	(a)\includegraphics[scale=0.45]{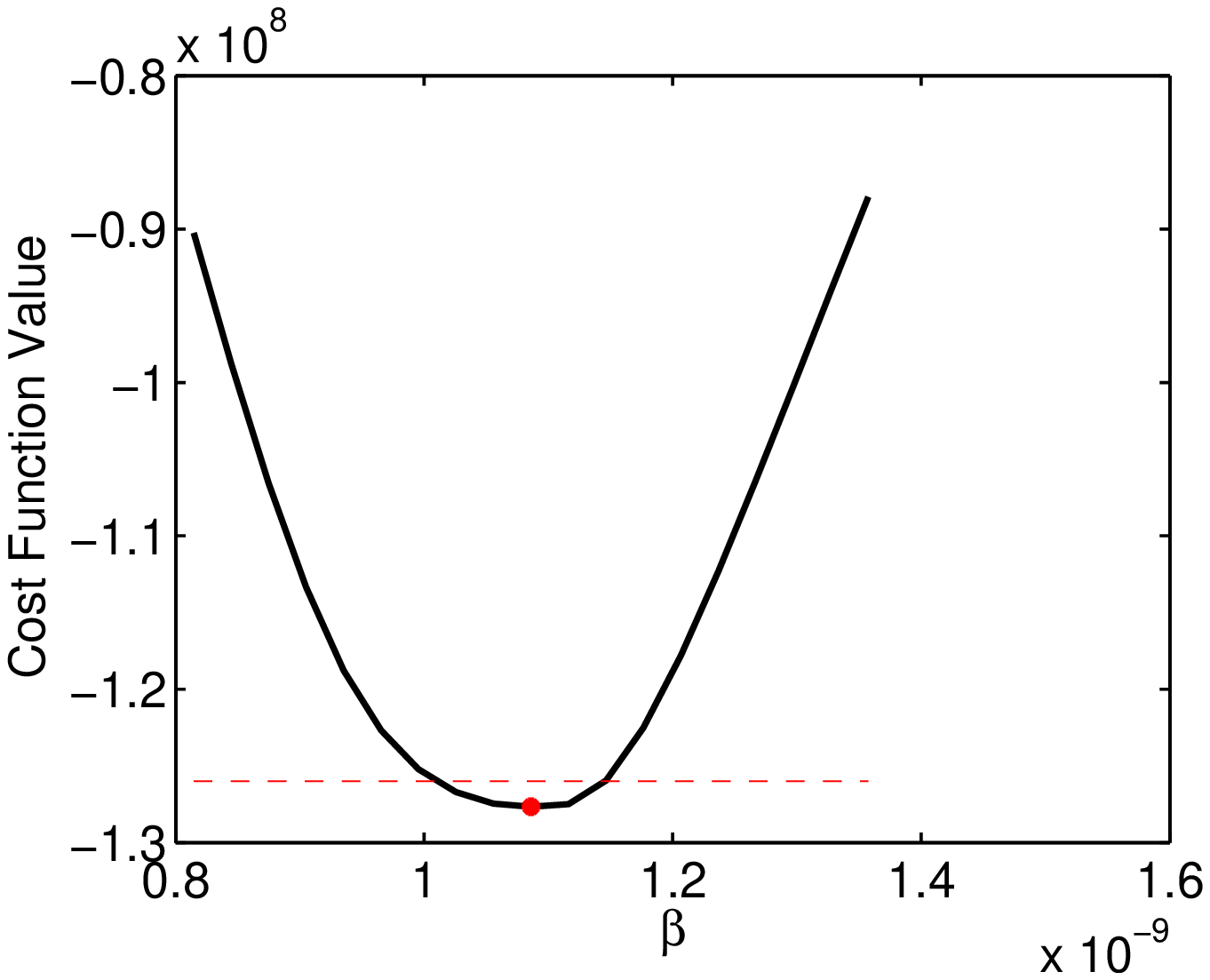}
	(b)\includegraphics[scale=0.45]{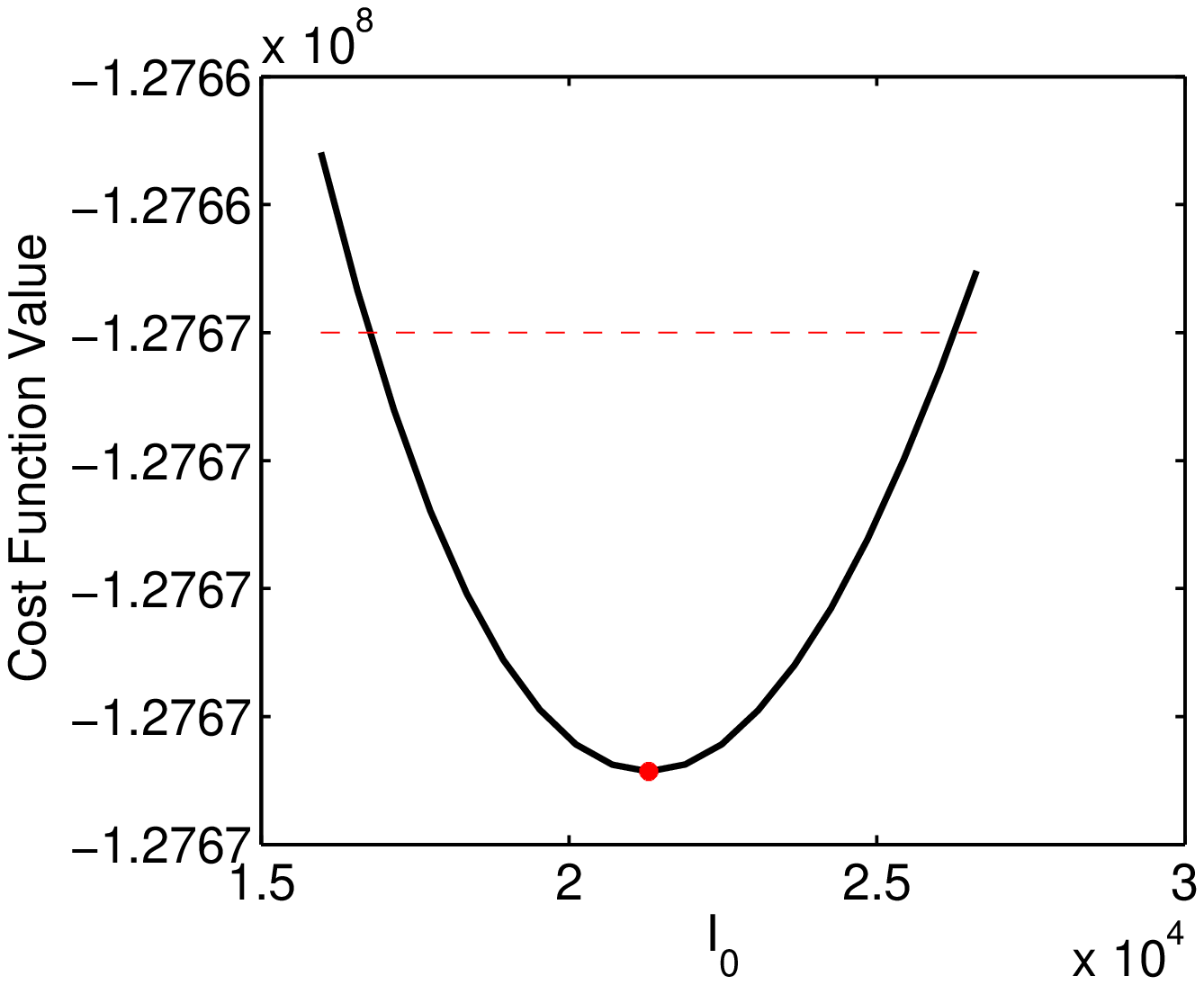}
	\caption{Cost functions based on profile likelihoods are shown for the parameters $\beta$ and $I_0$. Black curves show the cost function value, red dots denote the local minima, and the red dashed line signifies the 95 \% confidence thresholds.}\label{costfun}
\end{figure}

\subsection{Sensitivity analysis}\label{sens}
This section inspects how a slight variation in the critical model parameters influences (changes) weekly new and cumulative cases of infection as well as the basic reproduction number. The most important parameter to reduce in the system is the basic reproduction number of the disease transmission. The infection in the population will be removed if the basic reproduction number is reduced to less than unity. Even if the basic reproduction number cannot be reduced to lower than one, sensitivity analysis may benefit to uncover which parameters, if performed upon, would fetch the most significant decrease in the basic reproduction number. Sensitivity analysis gives the idea about what measures may degrade the weekly new, cumulative cases of infection and basic reproduction number.  

We perform the global sensitivity analysis of model \eqref{modeloutbreak} using the methodology of Latin Hypercube Sampling (LHS) and partial rank correlation coefficients (PRCCs) \cite{MAR2008} to investigate and recognize parameters accountable for most affecting the basic reproductive number as well as weekly new and cumulative cases of infection. The model parameters are randomly sampled using uniform distribution to execute the global sensitivity analysis. Then utilizing the baseline values from Table \ref{fitval}, the total sample size is set to 1000 simulations per LHS run. Partial rank correlation coefficients (PRCCs) values could be positive or negative. The parameter values with negative PRCCs signify that the processes described by such parameters can potentially control the epidemic when increased. On the other hand, the parameters with positive PRCCs specify that the processes described by such parameters can potentially make the epidemic worse if enhanced. PRCCs and their P-values for weekly and cumulative cases of infection are given in Figure \ref{prccWandC}. PRCC results illustrate that weekly and cumulative cases of infection are most sensitive to transmission rate $(\beta)$, testing rate $(\omega)$ and disease-induced death rate $(\delta_1)$ of infectives in compartment $I$. As observed from Figure \ref{prccWandC}(b) and (d), $\beta$, $\omega$, $\sigma$, $\delta_1$ and $\delta_2$ has significant impact on the output of weekly cases (P-value $< 0.05$) and $\beta$, $\omega$, $\sigma$, $\delta_1$, $\rho$ and $\delta_3$ has significant impact on the output of cumulative cases (P-value $< 0.05$). Respective PRCCs and P-values are also given in Table \ref{tblprcandp}.  

\begin{table}[H]
	\caption{PRCC and P-values for cumulative and weekly new cases.}\label{tblprcandp}
	\begin{tabular}{cp{3cm}cp{3cm}c}
		\hline
		Parameters	& PRCC values for cumulative cases & P-values & PRCC values for weekly new cases & P-values \\ 
		\hline
		$\beta$	& 0.8308 & 0 & 0.8277 & 0 \\ 
		$\omega$ & 0.3881 & 0 & 0.3544 & 0 \\ 
		$\delta_1$	& -0.8658 & 0 & -0.8593 & 0 \\ 
		$\sigma$	& -0.1799 & 0 & -0.1004 & 0.0015 \\ 
		$\delta_2$	& -0.0001 & 0.9965 & 0.0229 & 0.4718 \\ 
		$\delta_3$	& -0.2530 & 0 & -0.0041 & 0.8974 \\ 
		$\rho$	& 0.2586 & 0 & -0.0191 & 0.5472 \\
		$\theta$  & -0.6985 & 0 & -0.6860 & 0 \\
		\hline
	\end{tabular} 
\end{table}

Sensitivity analysis of $R_{01}$ of the model \eqref{modeloutbreak} based on LHS shows that the parameters $\beta$, $N_0$ and $\delta_1$ are highly sensitive, have PRCC values less than -0.5 or greater than 0.5 (see Figure \ref{prccR0}(a)). It could also be observed that the basic reproduction number is directly proportional to the transmission rate $(\beta)$ and total population size $(N_0)$; hence these parameters affect most. 
Figure \ref{prccR0}(b) represents the five-number summary (maximum value, upper quartile, median, lower quartile, and minimum value) for $R_{01}$. The median is around 1.2039, the lower quartile is around 0.8519, and the upper quartile is around 1.6733. The minimum and maximum values of $R_{01}$ are 0.3130 and 4.067, respectively. Figure \ref{prccR0}(c) depicts the histogram which tells the uncertainty in $R_{01}.$ The box plot and histogram plot are also produced via Latin hypercube sampling with sample sizes of 1000. Furthermore, scatter plots of the basic reproduction number $(R_{01})$ against the parameter $\beta$, $N_0$ and $\delta_1$ are also shown in Figure \ref{sct} for Latin hypercube sampling with sample sizes 1000. These scatter plots demonstrate the linear relationships (monotonicity) between outcomes of the basic reproduction numbers and input parameters.  

\begin{figure}[H]
	(a)\includegraphics[scale=0.5]{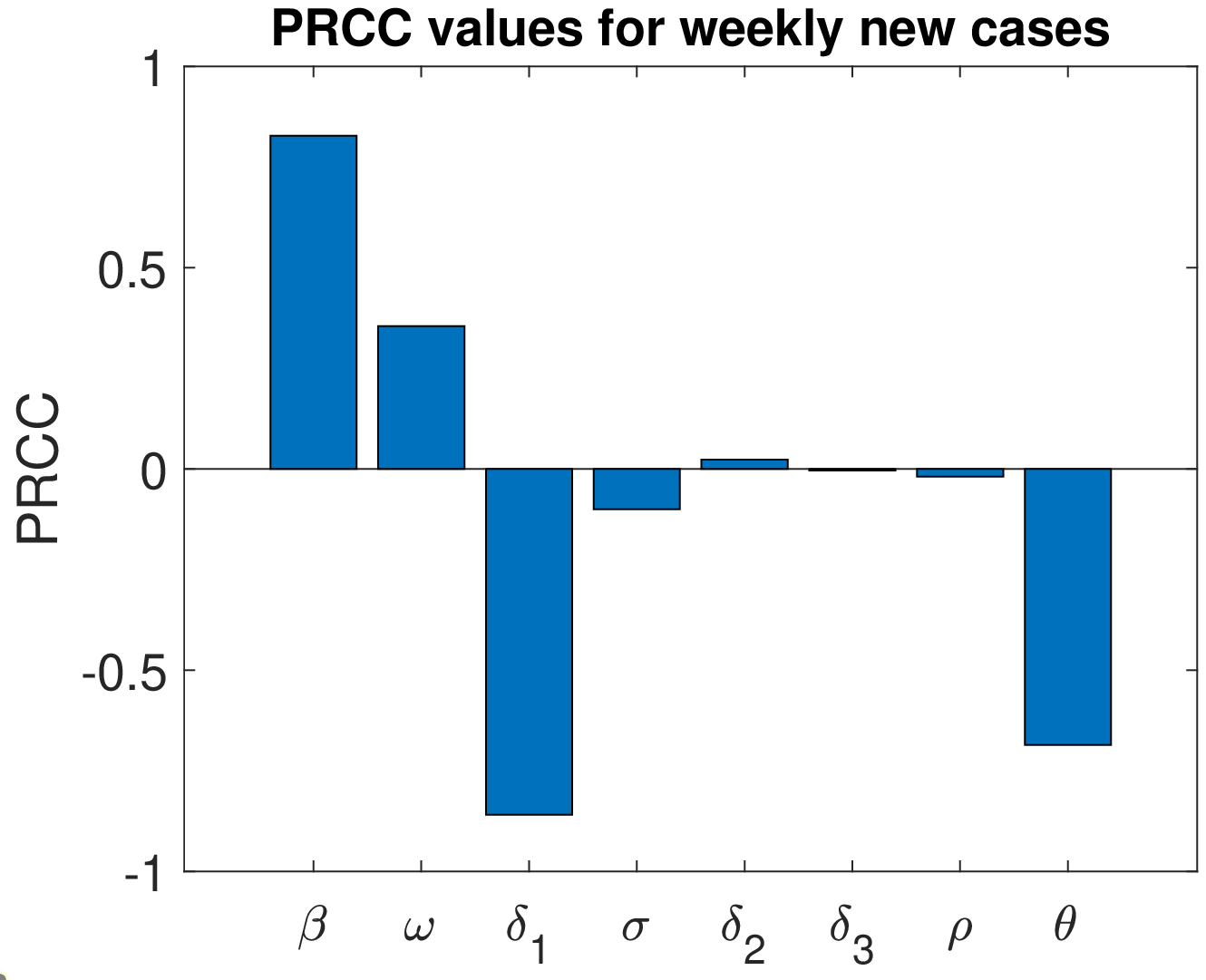}
	(b)\includegraphics[scale=0.5]{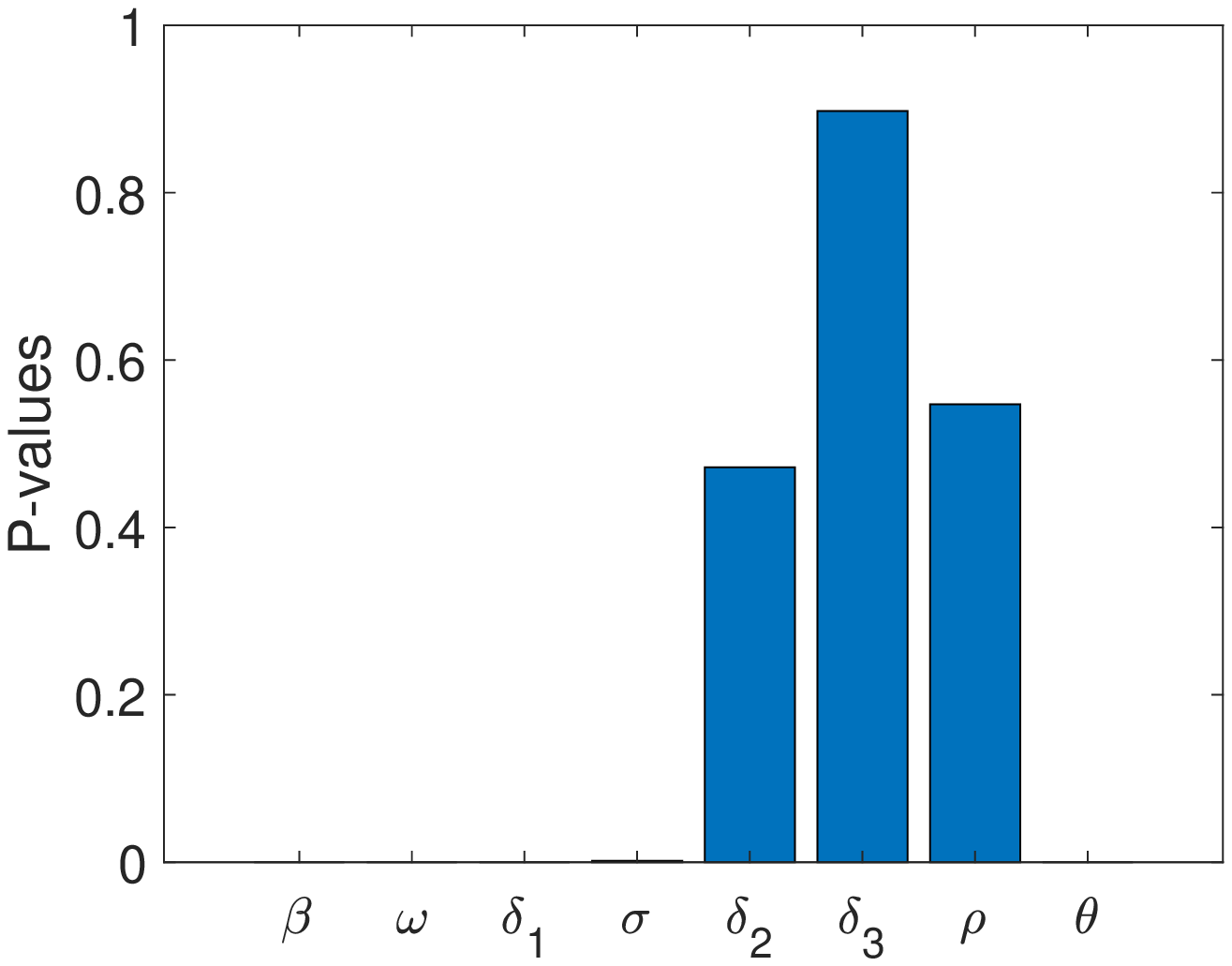}
	(c)\includegraphics[scale=0.5]{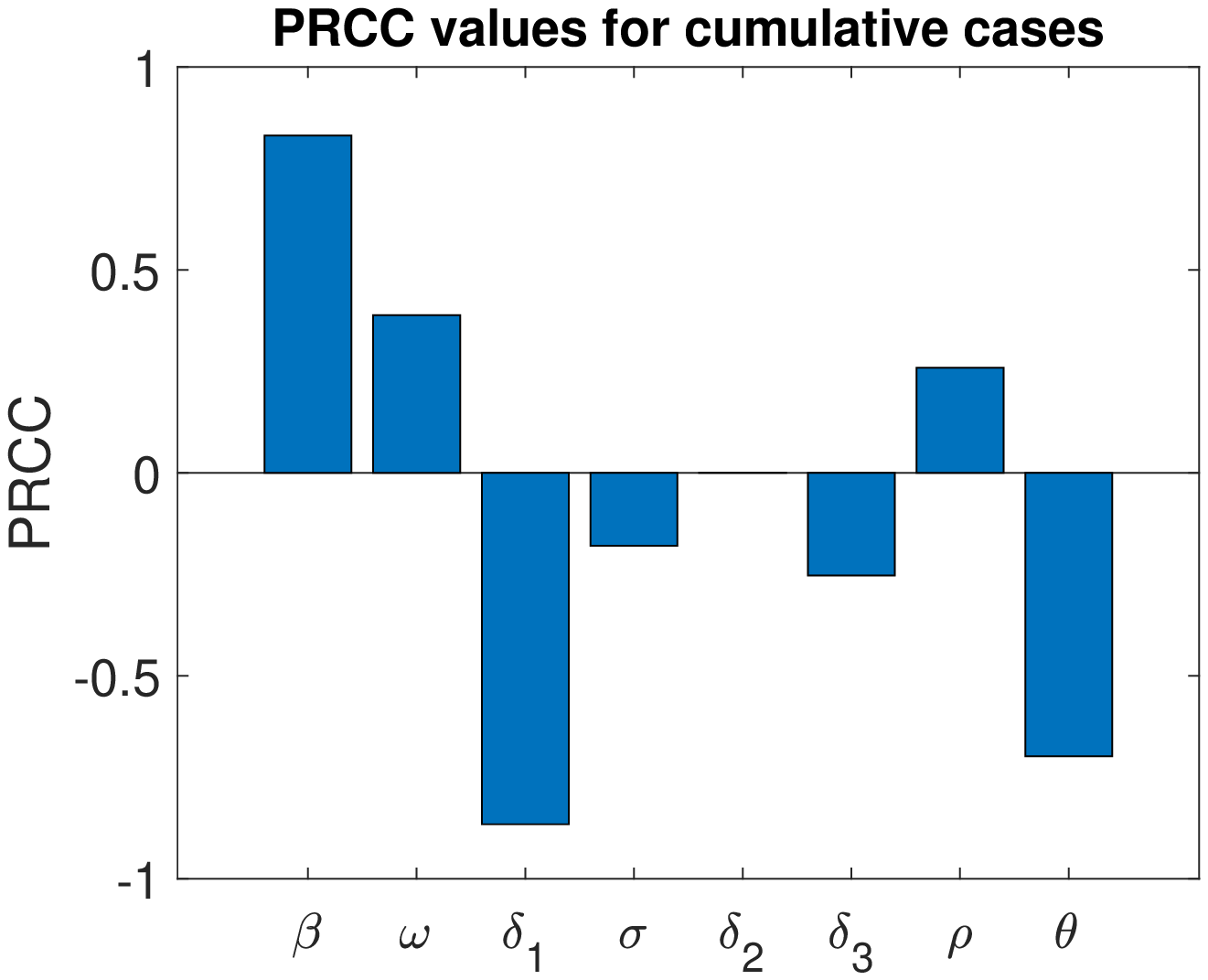}
	(d)\includegraphics[scale=0.5]{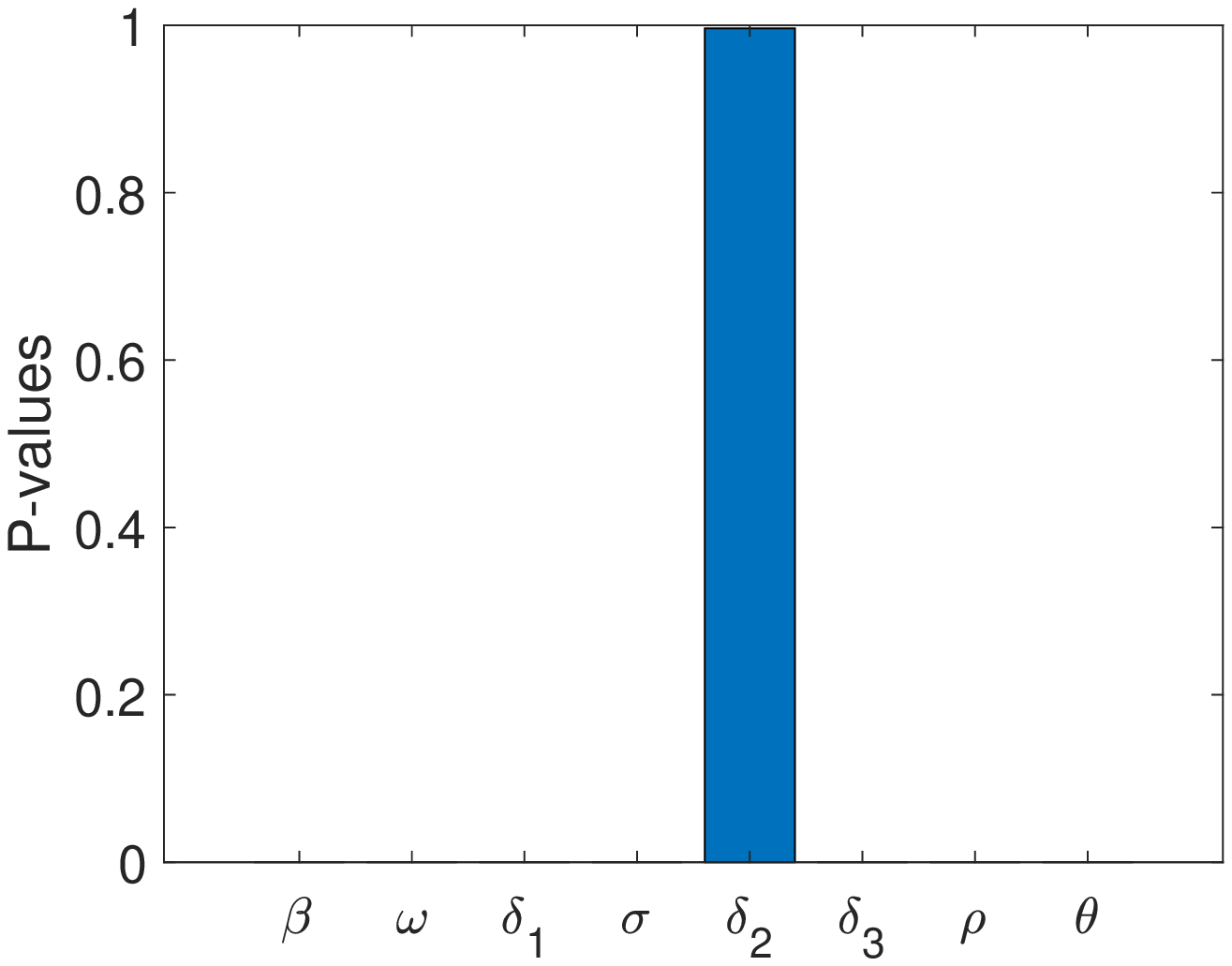}
	\caption{PRCC plots and P-values for weekly new and cumulative cases.}\label{prccWandC}
\end{figure}

\begin{figure}[H]
	(a)\includegraphics[scale=0.5]{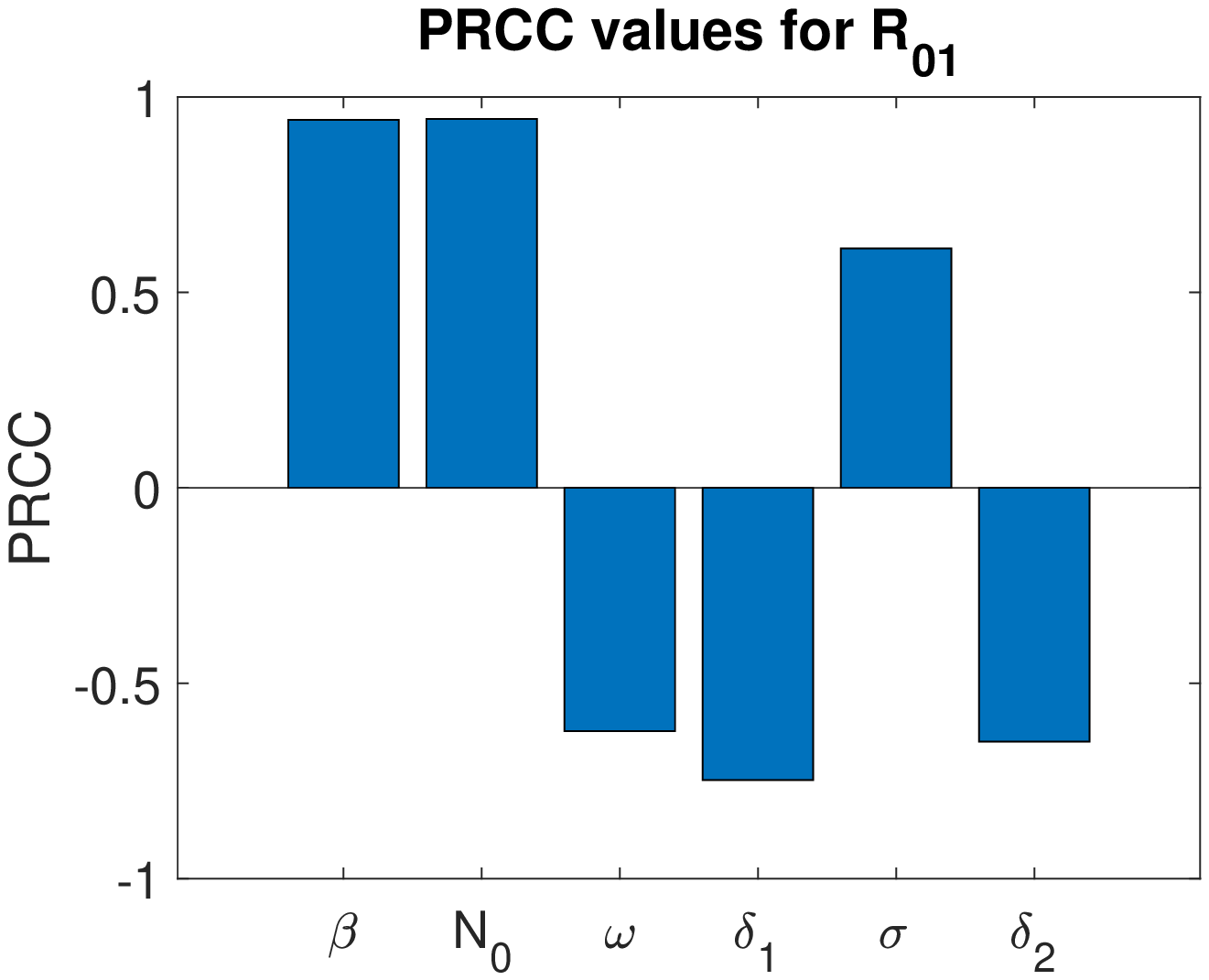}
	(b)\includegraphics[scale=0.5]{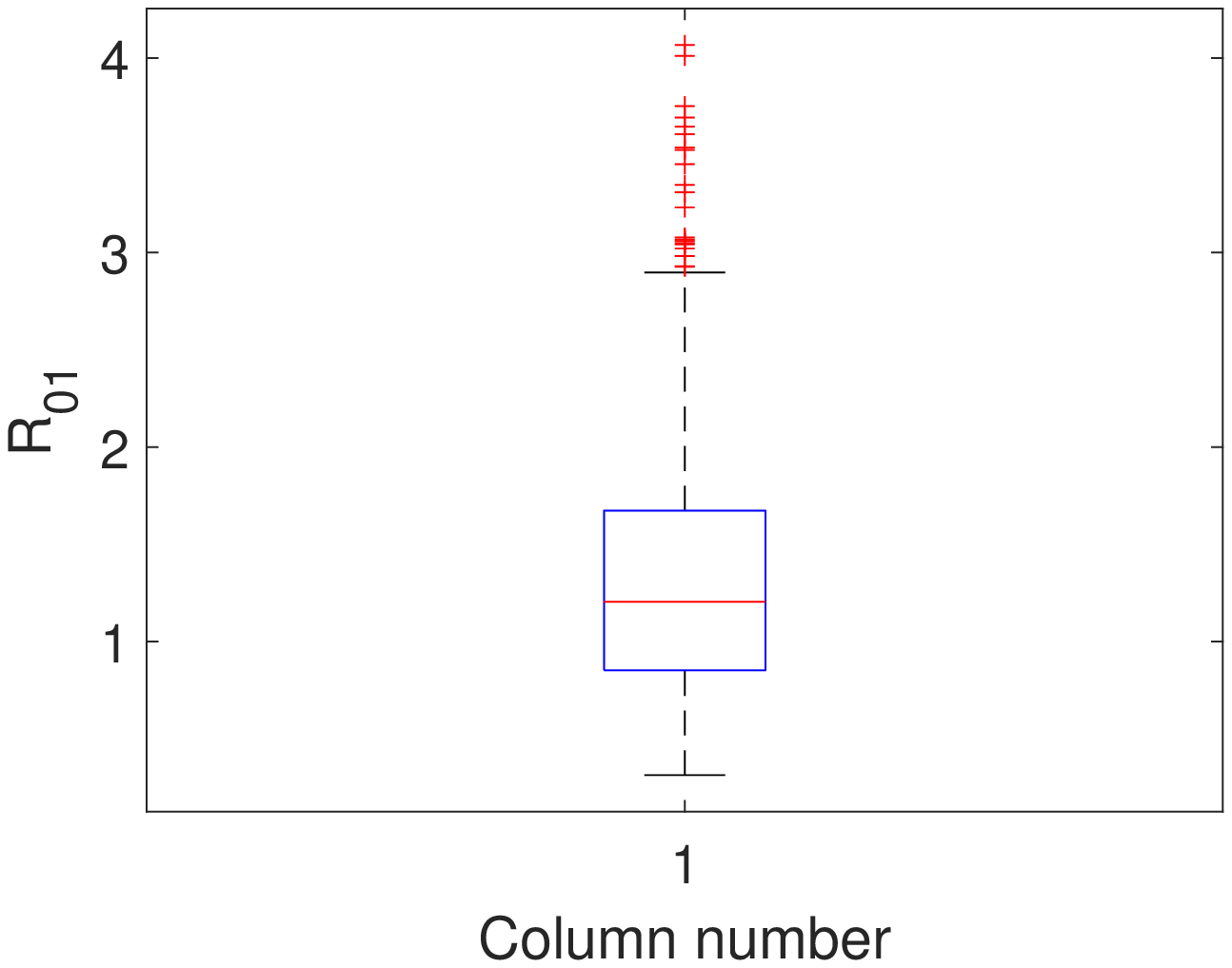}
	(c)\includegraphics[scale=0.5]{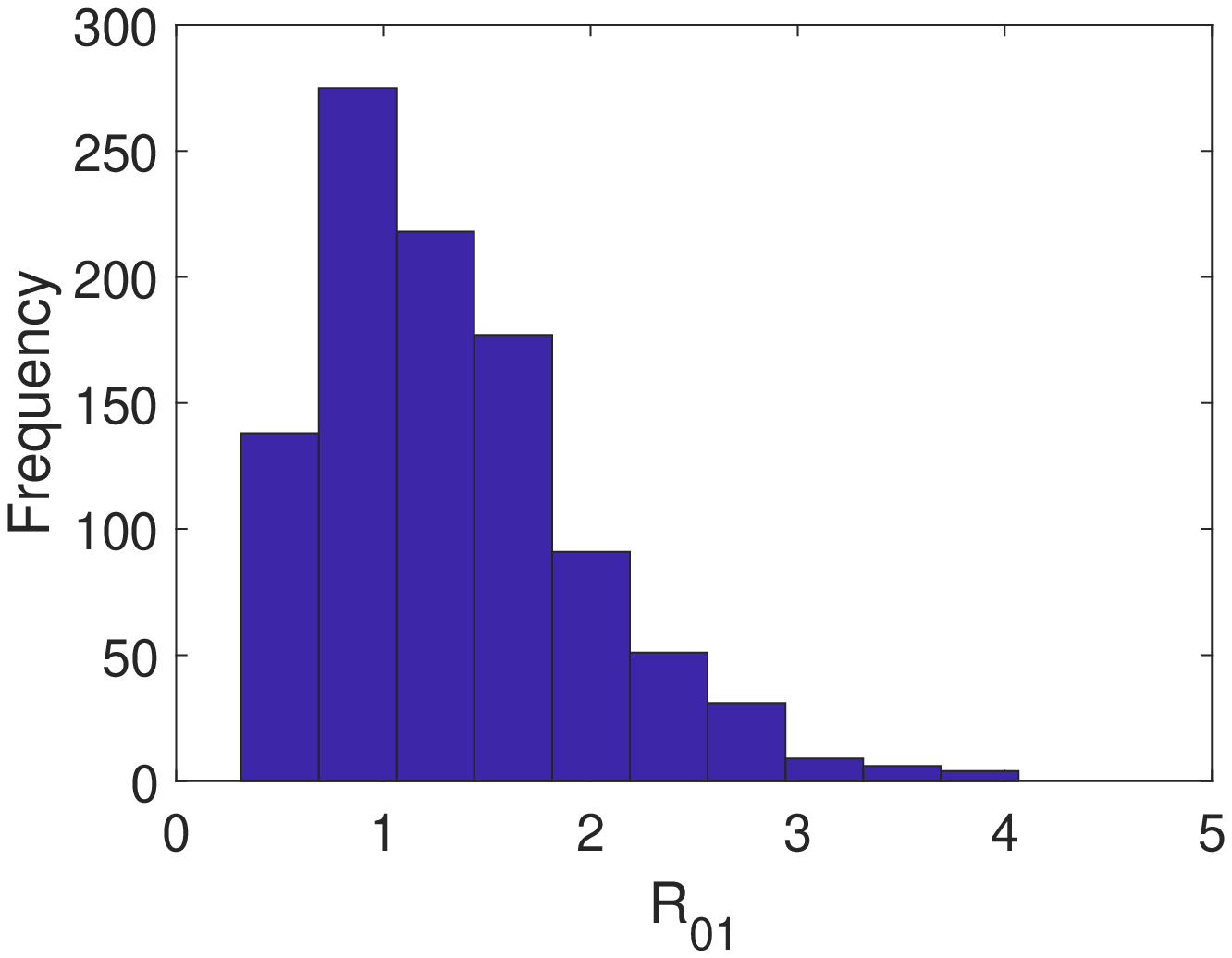}
	\caption{(a) PRCC plots for the basic reproduction number $(R_{01})$. (b) Box plot for $R_{01}.$ (c) Histogram of the distribution of $R_{01}.$}\label{prccR0}
\end{figure}

\begin{table}[H]
	\caption{PRCC values for the basic reproduction number $(R_{01})$.}
	\begin{tabular}{cc}
		\hline
		Parameters	& PRCC values for $R_{01}$ \\
		\hline
		$\beta$	& 0.9415 \\ 
		$N_0$	& 0.9441 \\ 
		$\omega$	& -0.6223 \\ 
		$\delta_1$	& -0.7475 \\ 
		$\sigma$	& 0.6124 \\ 
		$\delta_2$	& -0.6493 \\
		\hline
	\end{tabular} 
\end{table}

\begin{figure}[H]
	(a)\includegraphics[scale=0.5]{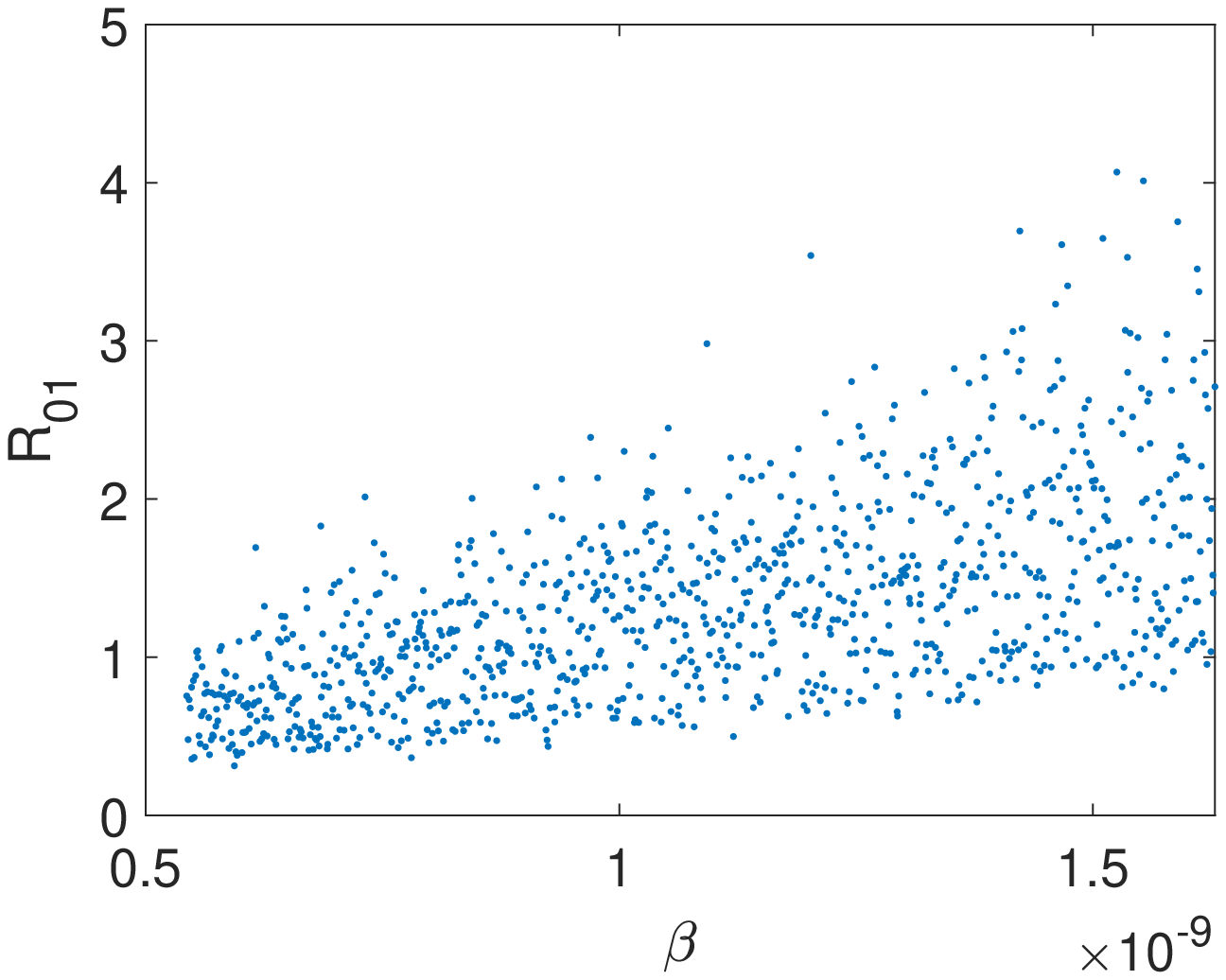}
	(b)\includegraphics[scale=0.5]{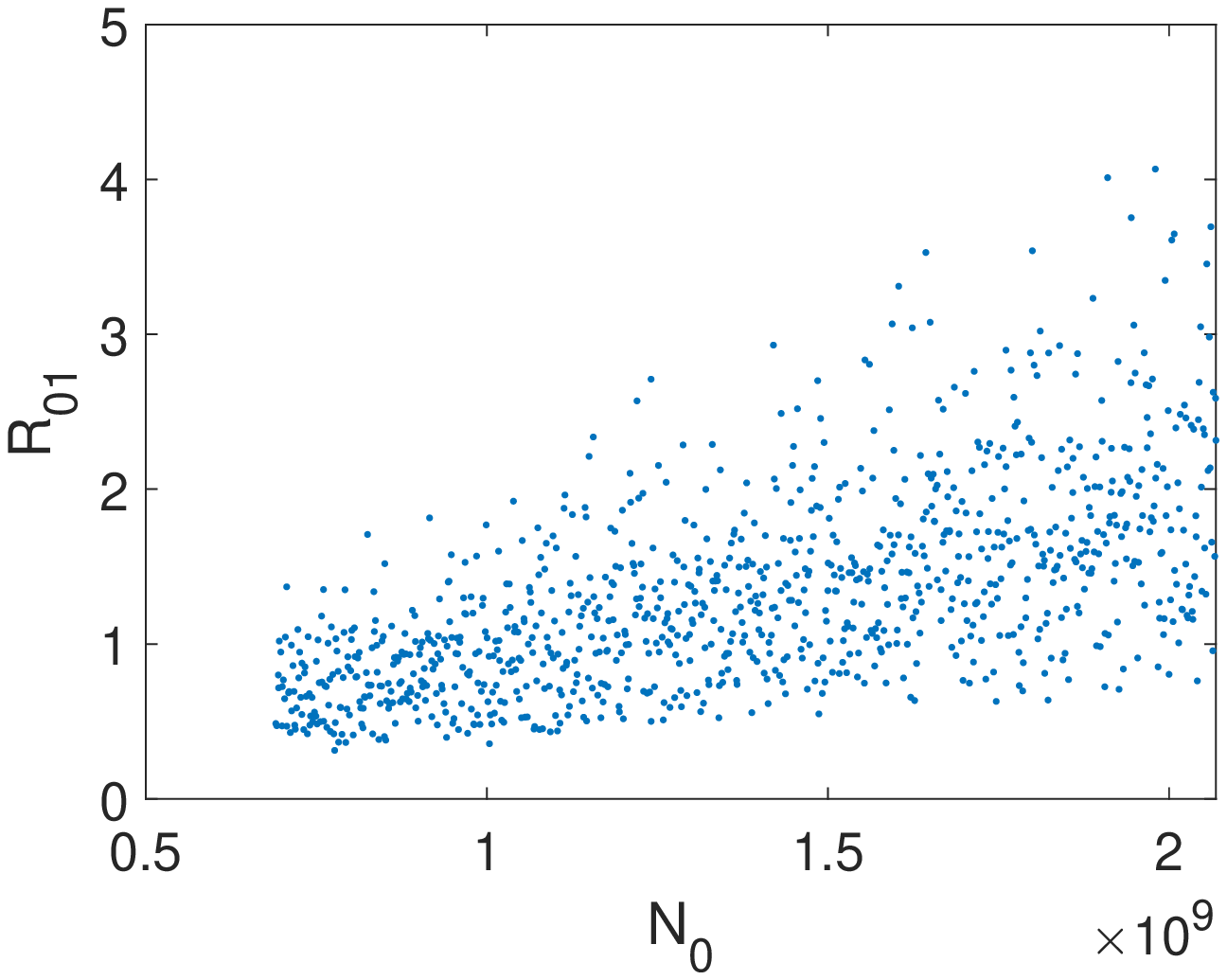}
	(c)\includegraphics[scale=0.5]{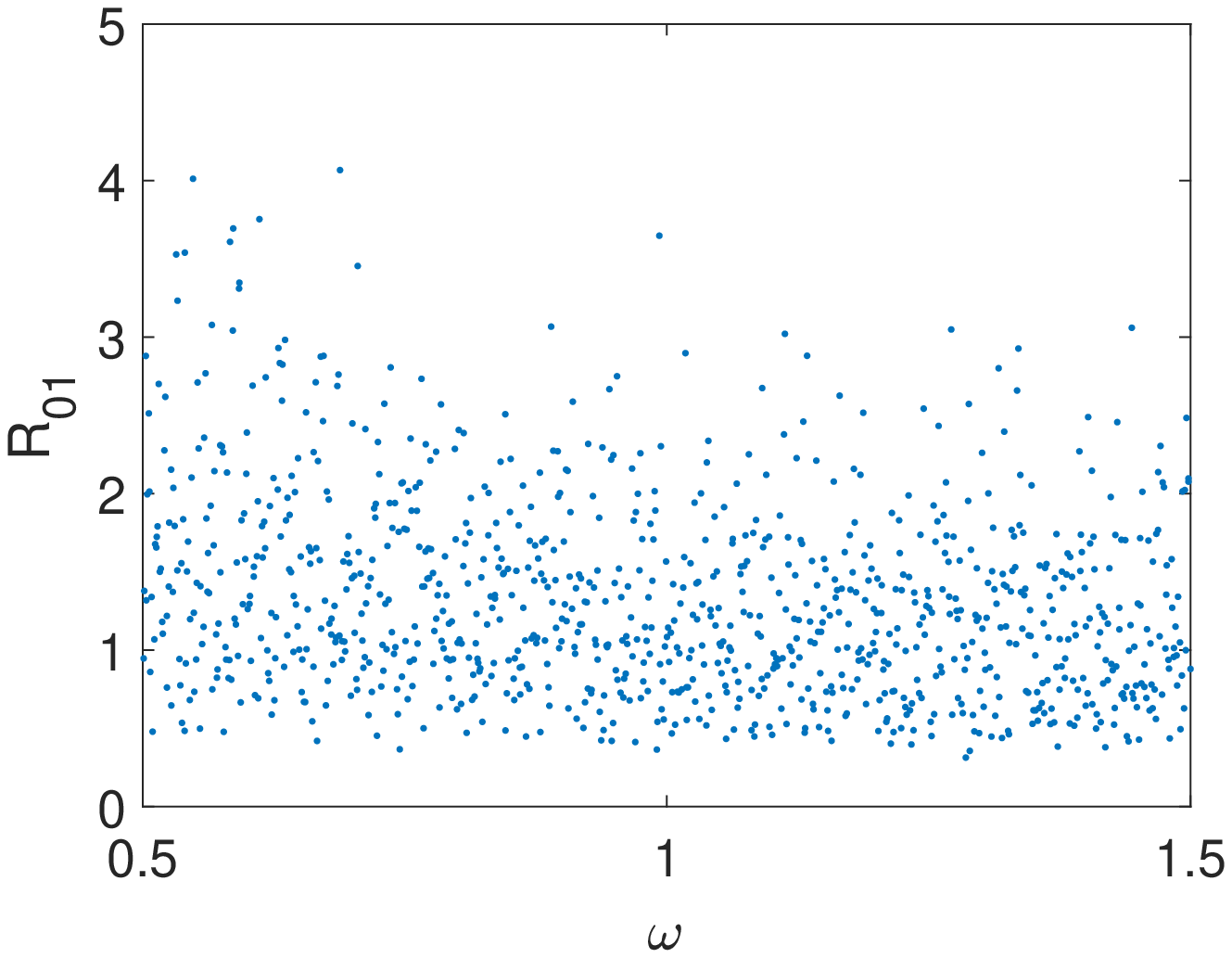}
	(d)\includegraphics[scale=0.5]{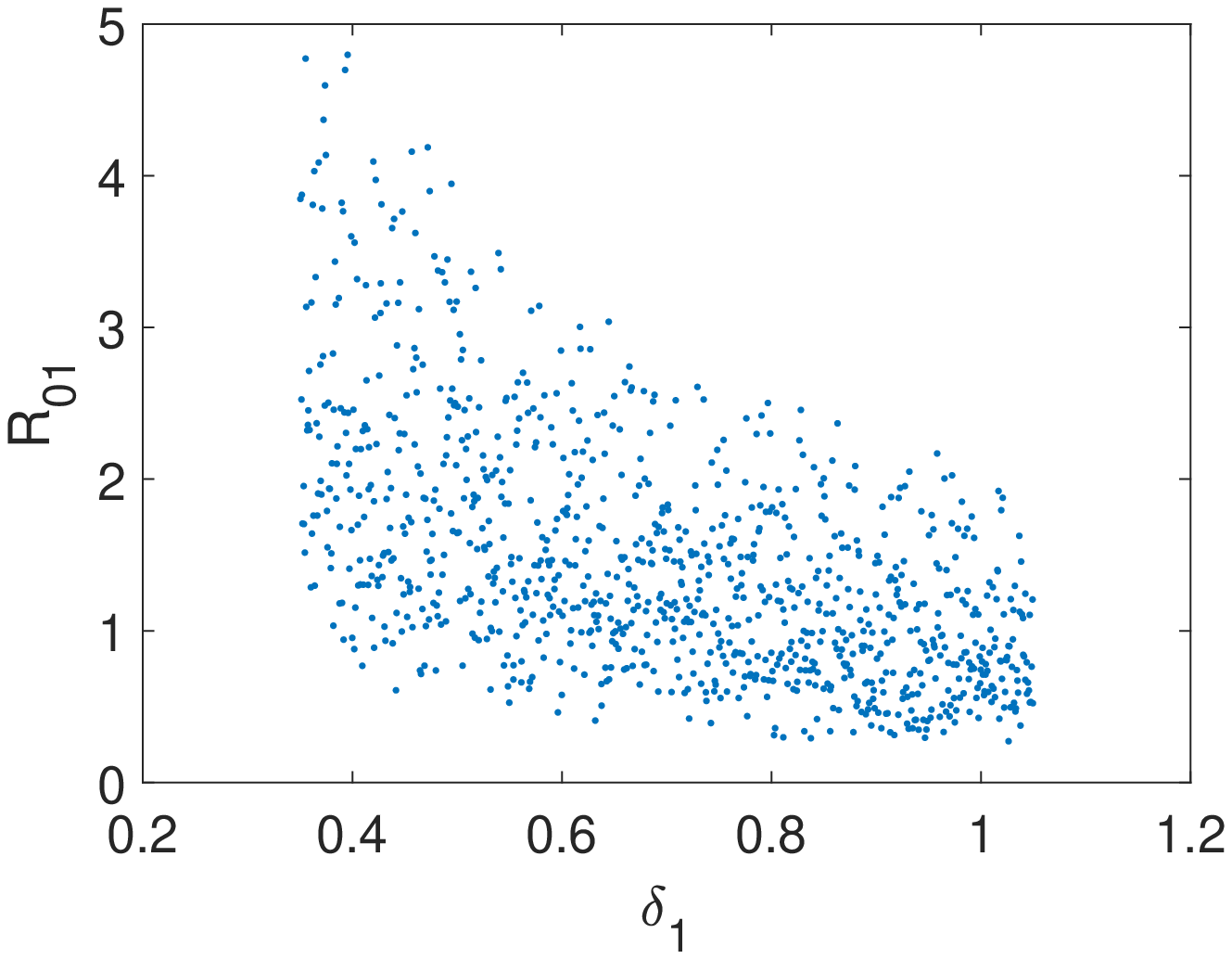}
	(e)\includegraphics[scale=0.5]{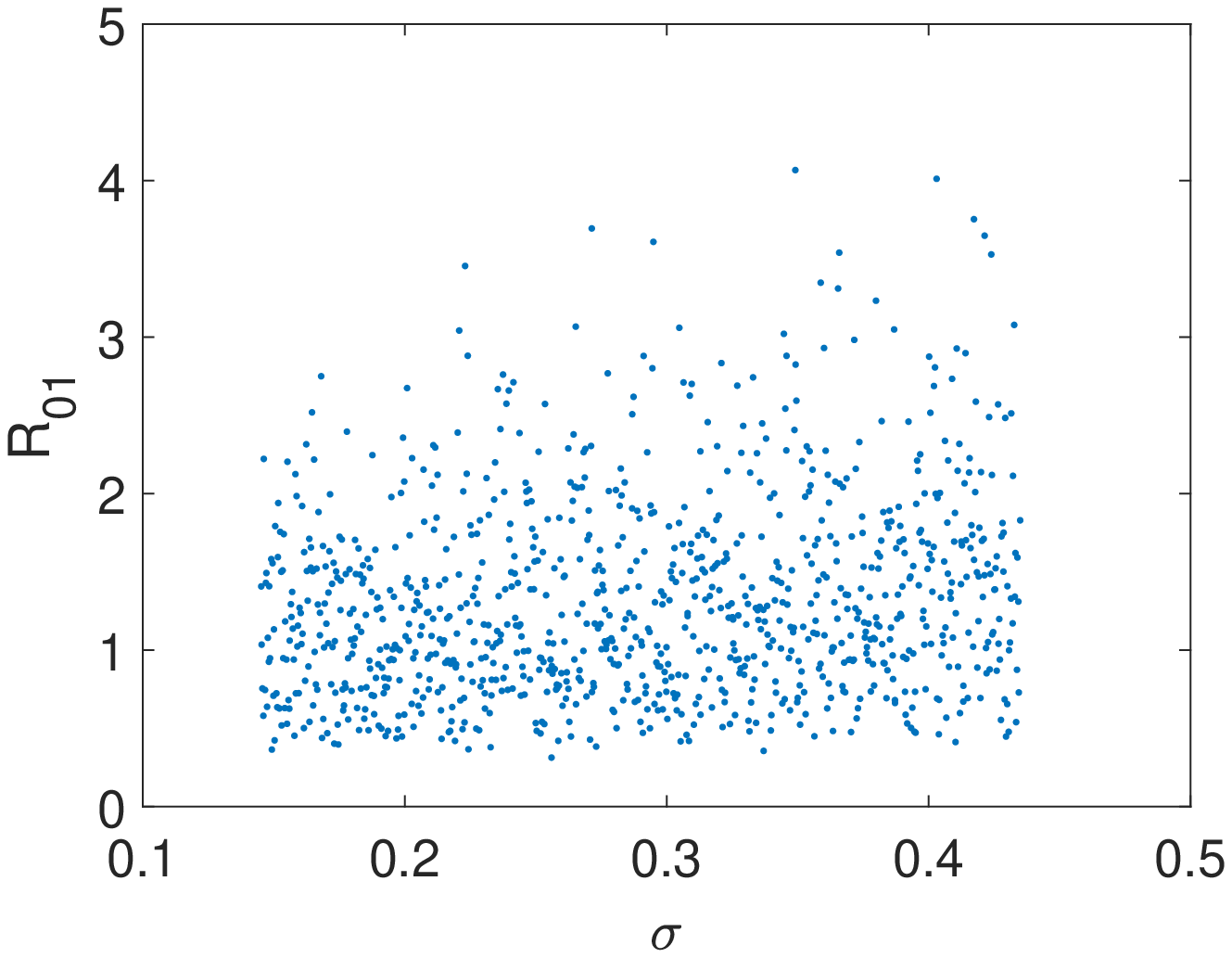}
	(f)\includegraphics[scale=0.5]{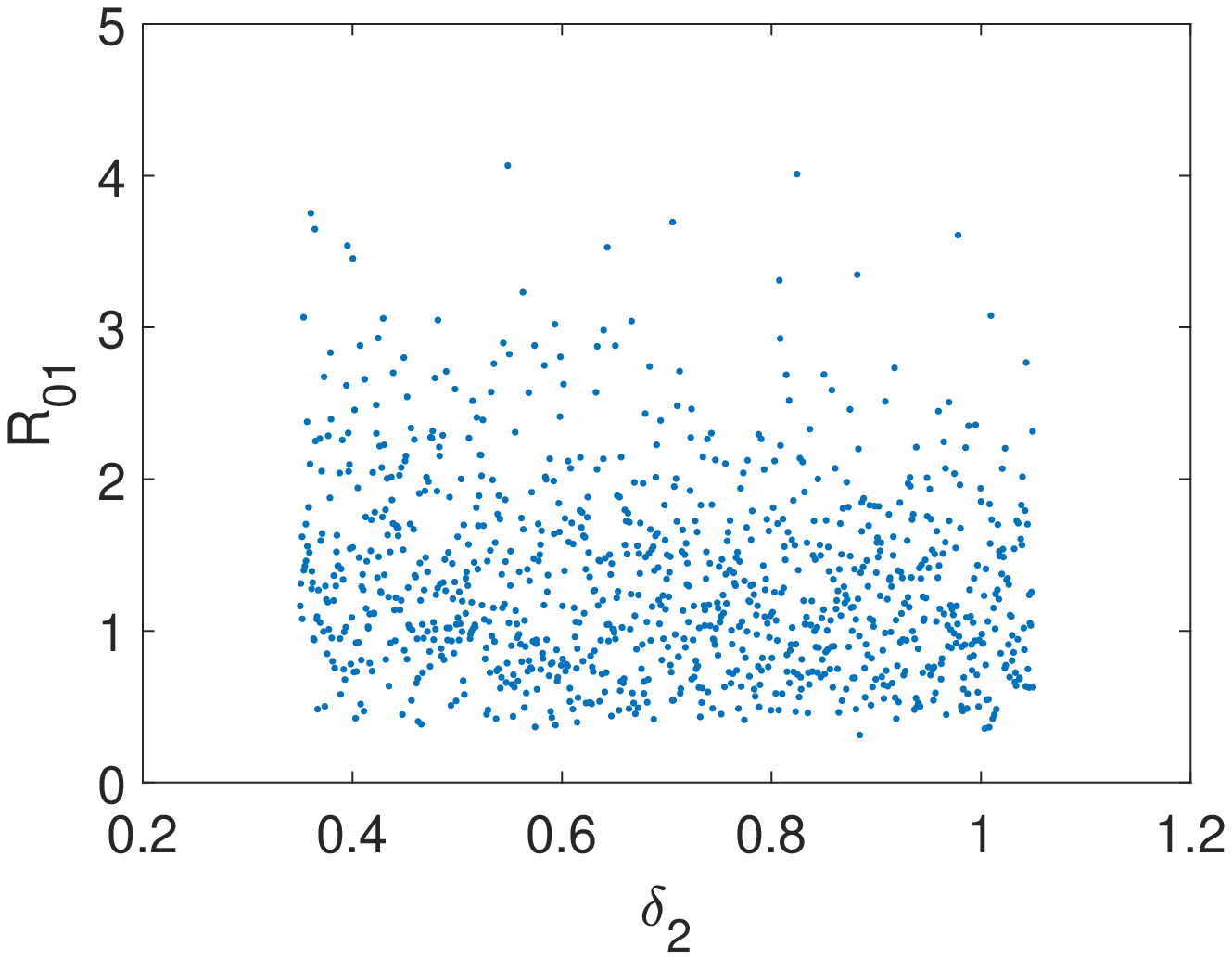}
	\caption{Scatter plots of the basic reproduction number $(R_{01})$ with respect to the parameters ($\beta, N_0, \omega, \delta_1, \sigma$ and $\delta_2$).}\label{sct}
\end{figure}

\subsection{Impact of parameters on the outbreak's peak, cumulative cases, basic reproduction number $(R_{01})$ and final size}\label{5.5} 
This section assesses the impact of testing rate $(\omega)$, testing efficacy $(1-\sigma)$ and transmission rate $(\beta)$ on the weekly new cases at the peak, cumulative cases, and the basic reproduction number $(R_{01})$. We simulate model \eqref{modeloutbreak}, using the baseline parameter values given in Table \ref{fitval} and \ref{estmdini}. In this work, we mainly want to emphasize the impact of testing rate and efficacy, but the transmission rate is a very important parameter in the disease epidemic. The transmission rate could be controlled by various control strategies, such as vaccination, using face masks, social distancing, public awareness, etc. Thus it is necessary to assess the impact of different values of $\beta$ on the model outcomes. \\\\
\textbf{Assessment of the impact of the testing rate:} 
First, we investigate the impact of $\omega$ on the weekly new cases of the early outbreak of COVID-19 in India shown in Figure \ref{impctomega}(a). Hence, we utilize five different values for $\omega$, where $\omega = 1$ is the baseline value as given in Table \ref{fitval}, while other four values are $10\%$, $20\%$, $30\%$, and $40\%$ increment in the baseline value. From Figure \ref{impctomega}(a), it can be observed that increasing the testing rate can dramatically decrease the number of infected individuals at the outbreak's peak and delay the outbreak's peak. It is because the confirmed positive tested individuals move to the isolation/quarantine compartment and do not participate in the disease transmission. The blue curve represents the weekly new cases for the baselines values given in Table \ref{fitval}. If the testing rate is increased by 10\% from its baseline value, the weekly new cases at the peak are decreased by 19.93\% and delayed the peak by four weeks (red curve). Further, If the testing rate is increased by 20\% from its baseline value, the weekly new cases at the peak are decreased by 37.63\% and delayed the peak by eight weeks (yellow curve). Furthermore, If the testing rate is increased by 30\% from its baseline value, the weekly new cases at the peak are decreased by 52.90\% and delayed the peak by fourteen weeks (purple curve). Similar findings are recorded for the cumulative cases represented in Figure  \ref{impctomega}(b). The testing rate also significantly reduces the basic reproduction number ($R_{01}$), represented in Figure \ref{impctomega}(c).\\\\
\noindent \textbf{Assessment of the impact of the testing efficacy:} We assess the impact of testing efficacy ($1-\sigma$) on the weekly new cases, cumulative cases, and $R_{01}$ shown in Figure \ref{impctsigma}. For this we plot the weekly new cases for the testing efficacies 55\%, 60\%, 71\% (baseline value, $1-\sigma = 0.71$), 80\%, and 85\% in Figure \ref{impctsigma}(a). The increasing testing efficacy intensely decreases the newly detected cases and cumulative cases and delays the outbreak's peak. Figure \ref{impctsigma}(a) shows that if the testing efficacy is increased by 12.67\% from its baseline value, i.e., $1-\sigma = 80\%$, the weekly new cases at the peak are decreased by 59.05\% and delayed the peak by 15 weeks (red curve). On the other hand, if the testing efficacy is decreased by 15.49\% from its baseline value, i.e., $1-\sigma = 60\%$, the weekly new cases at the peak are increased by 60.67\%, and the peak appeared eight weeks earlier (purple curve). Further, if the testing efficacy is decreased by 22.53\% from its baseline value, i.e., $1-\sigma = 85\%$, the weekly new cases at the peak are increased by 78.88\%, and the peak appeared by 11 weeks earlier (green curve). Similar outcomes in the cumulative cases are recorded with the increasing and decreasing testing efficacy shown in Figure \ref{impctsigma}(b). If there are many infected individuals in the community, then perfection in the testing results may detect more infected individuals and move them into the isolation/quarantine compartment. The perfection in testing outcomes also substantially reduces the basic reproduction number ($R_{01}$), as shown in Figure \ref{impctsigma}(c).\\\\ 
\noindent \textbf{Assessment of the impact of different values of the transmission rate:} The outcomes of assessment for distinct values of transmission rate $(\beta)$ are depicted in Figure \ref{impctbeta}. Figure \ref{impctbeta}(a) shows the variation in the parameter $\beta$ for weekly new cases, where the yellow curve is for the baseline values of the parameters given in Table \ref{fitval}. We observe that a 10\% increment in the baseline value of $\beta$ increases the weekly new cases at the peak by 97.47\% (purple curve). Additionally, a 20\% increment in the baseline value of $\beta$ could increase the weekly new cases at the peak by 205.89\% (green curve). On the other hand, a 10\% reduction in the baseline value of $\beta$ could reduce the weekly new cases at the peak by 40.23\% (red curve) and delay the peak time by 16 weeks. Akin outcomes are noted for the cumulative cases for the different values of $\beta$ illustrated in \ref{impctbeta}(b). The basic reproduction number $(R_{01})$ also increases with the increasing value of $\beta$ depicted in Figure \ref{impctbeta}(c).

\begin{figure}[H]
	(a)\includegraphics[scale=0.45]{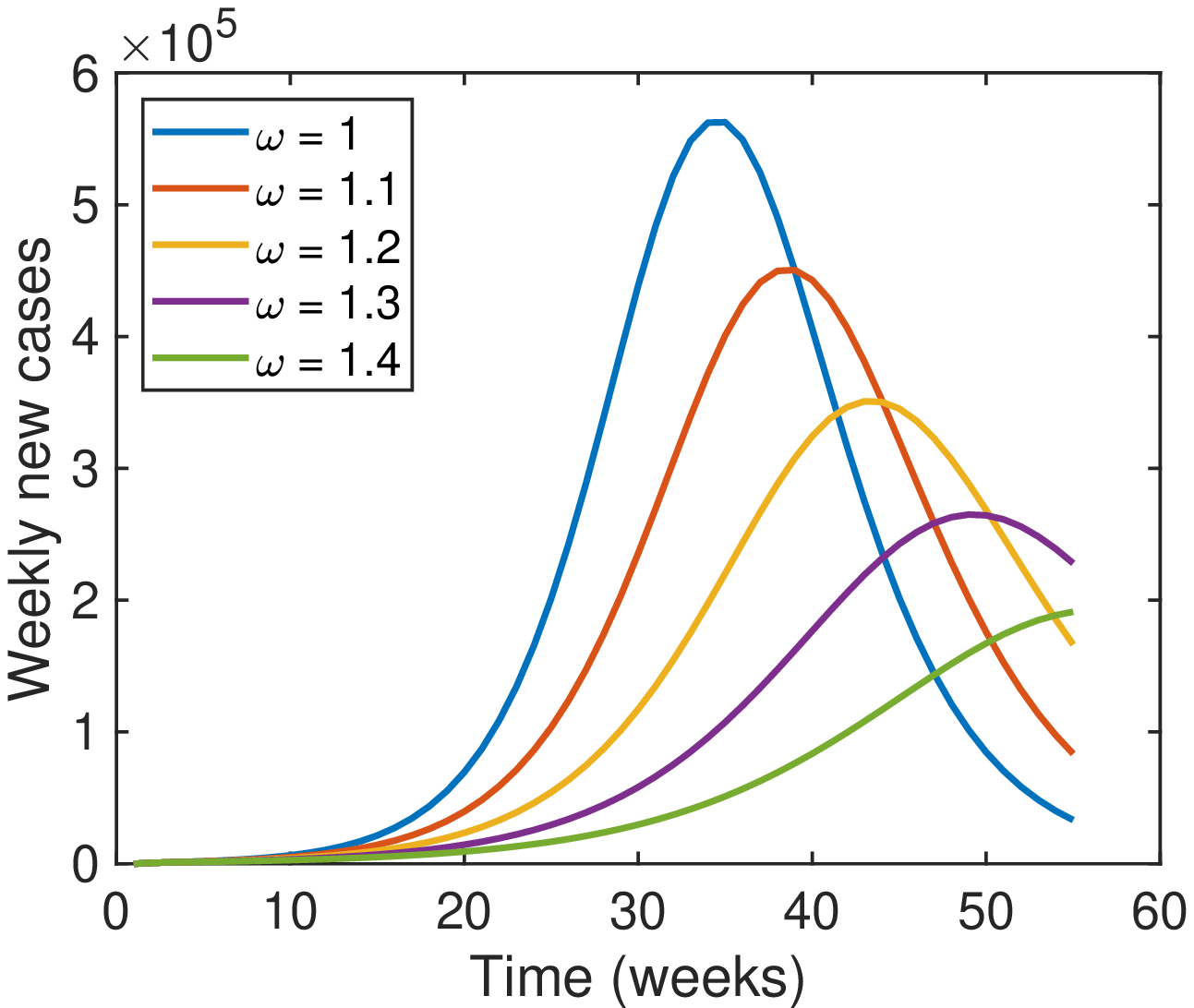}
	(b)\includegraphics[scale=0.45]{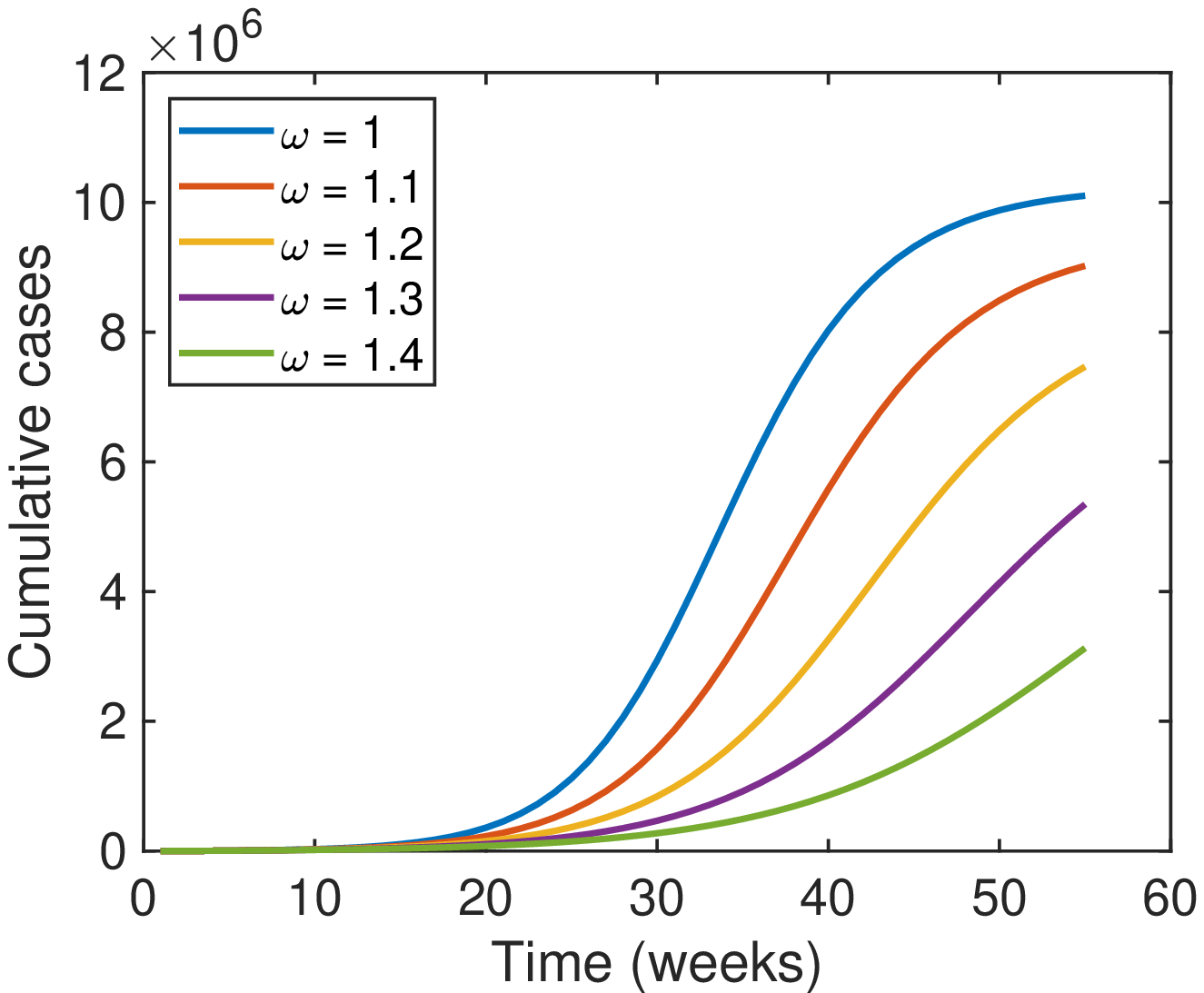}
	(c)\includegraphics[scale=0.45]{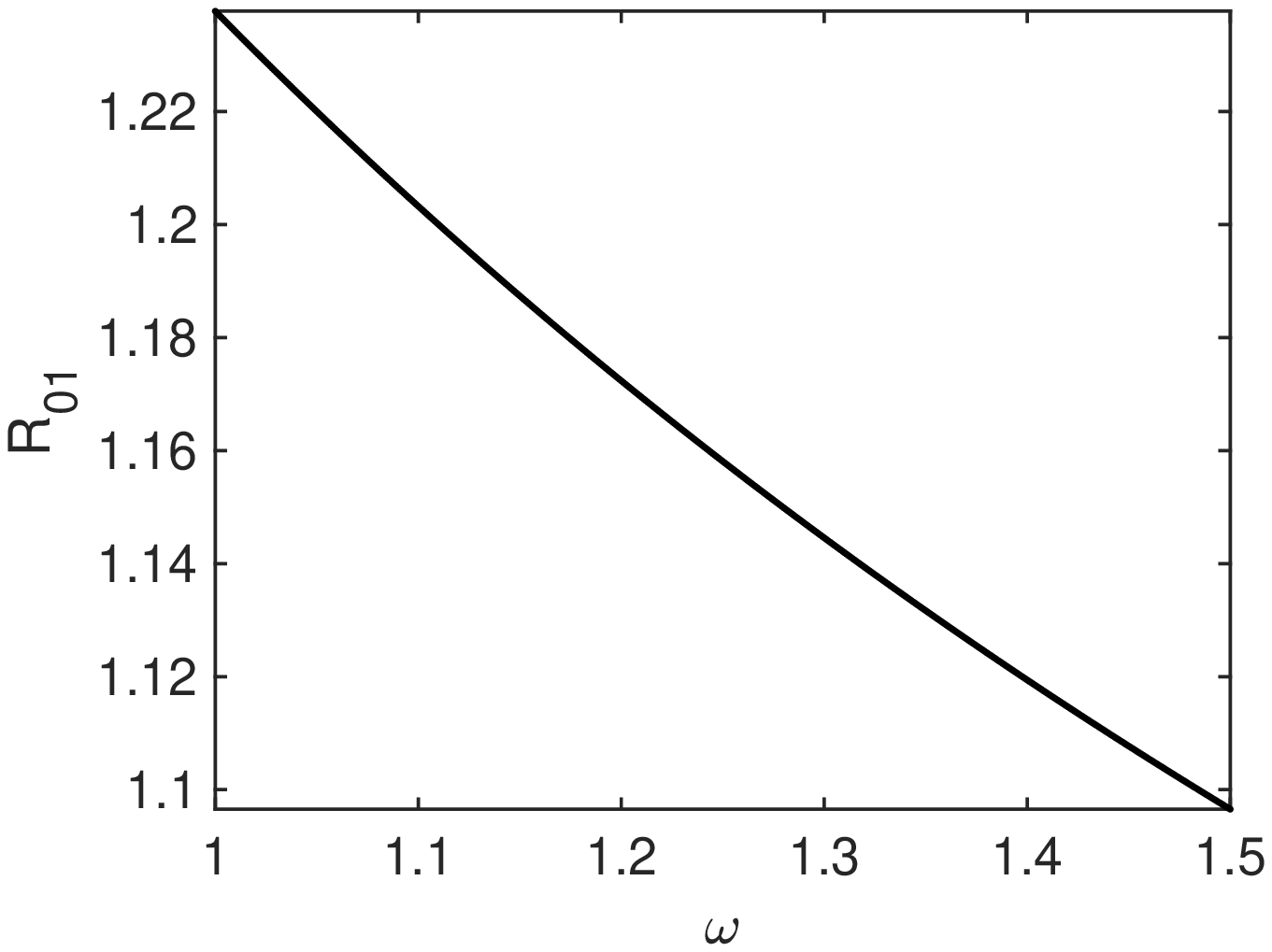}
	\caption{Simulations of the model \eqref{modeloutbreak} exhibiting the impact of testing rate $(\omega)$. (a) Figure depicts weekly new cases for different values of $\omega$. (b) Cumulative cases for different values of $\omega$. (c) The basic reproduction number, $R_{01}$, with respect to $\omega$.}\label{impctomega}
\end{figure}  

\begin{figure}[H]
	(a)\includegraphics[scale=0.45]{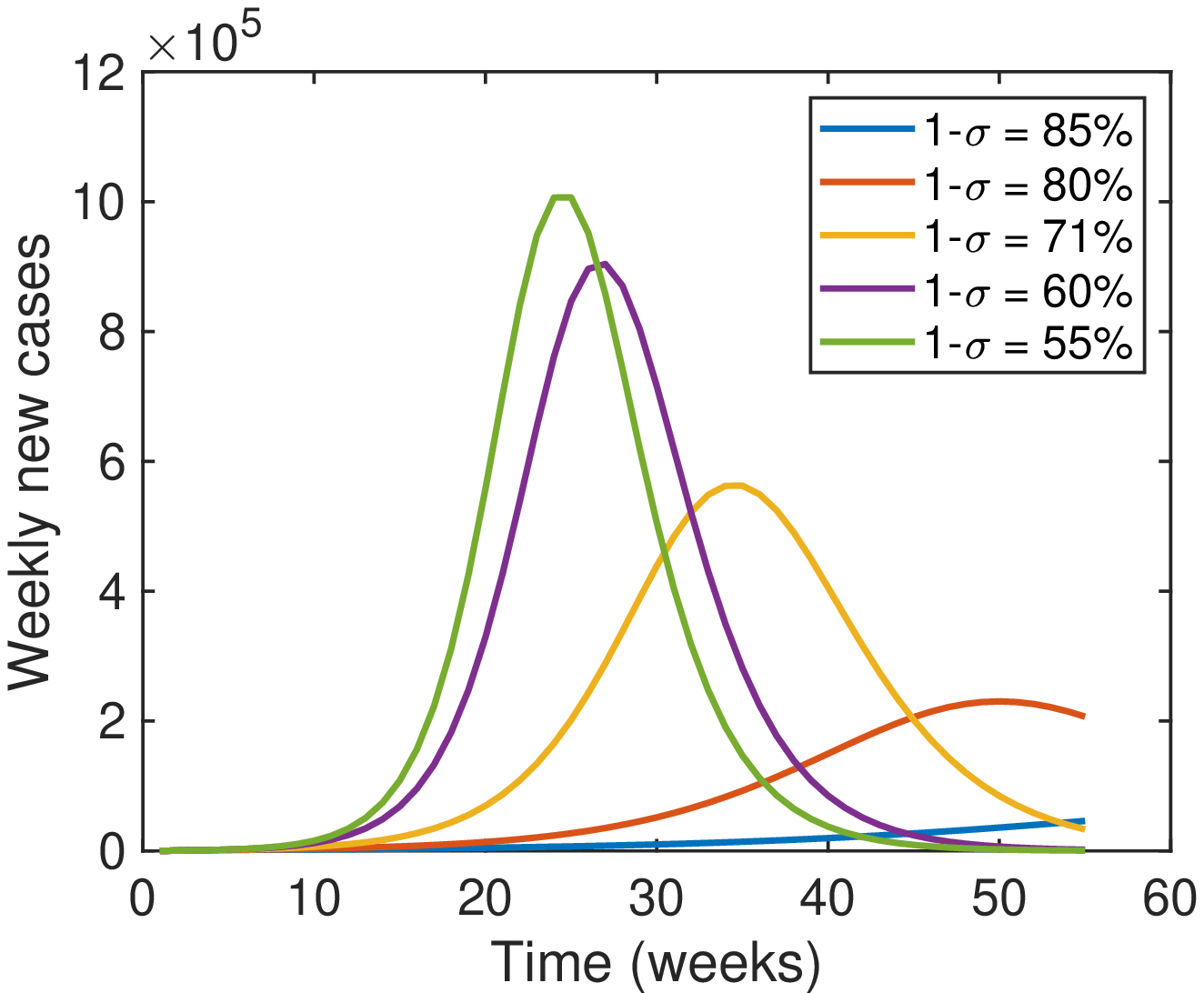}
	(b)\includegraphics[scale=0.45]{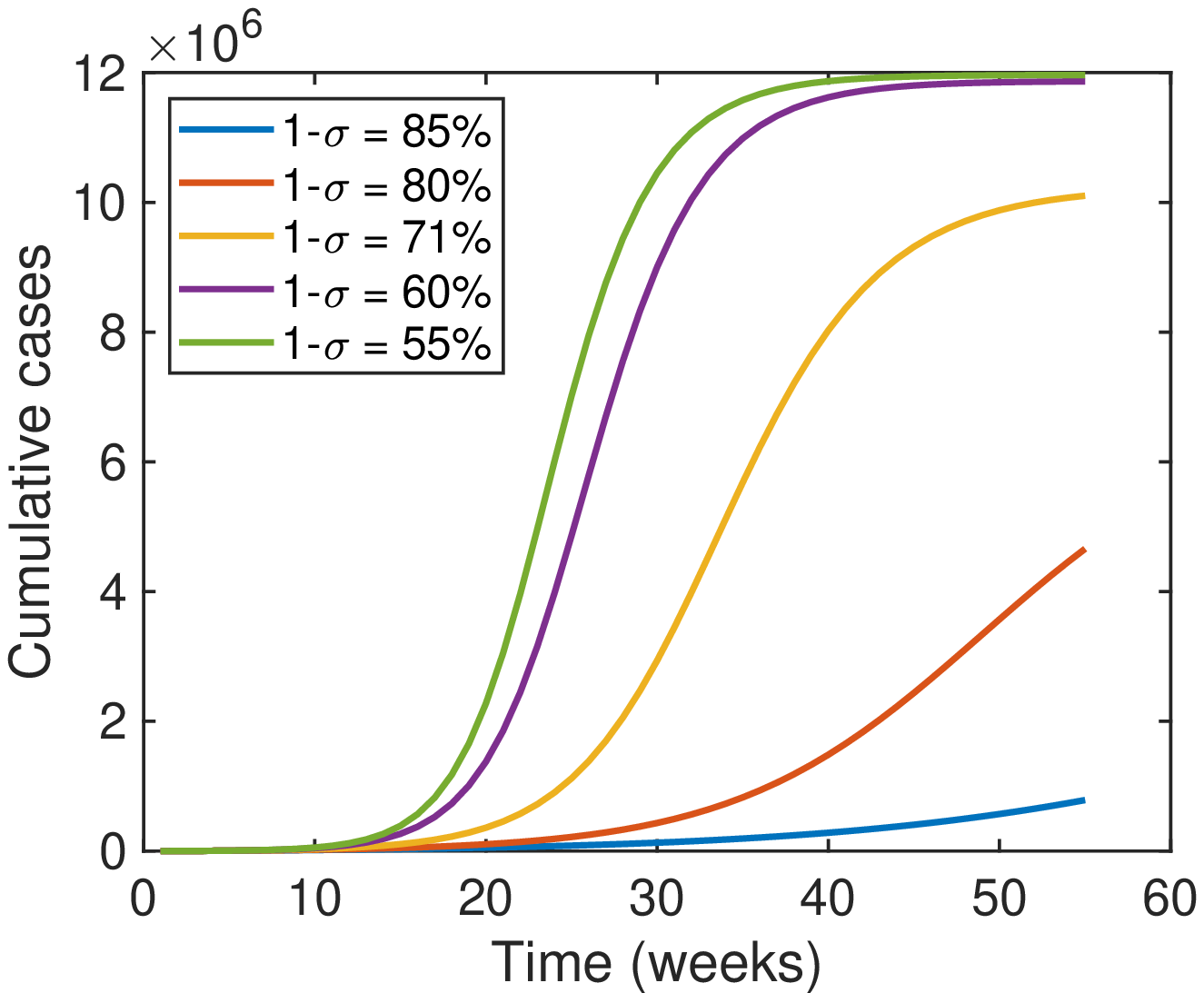}
	(c)\includegraphics[scale=0.45]{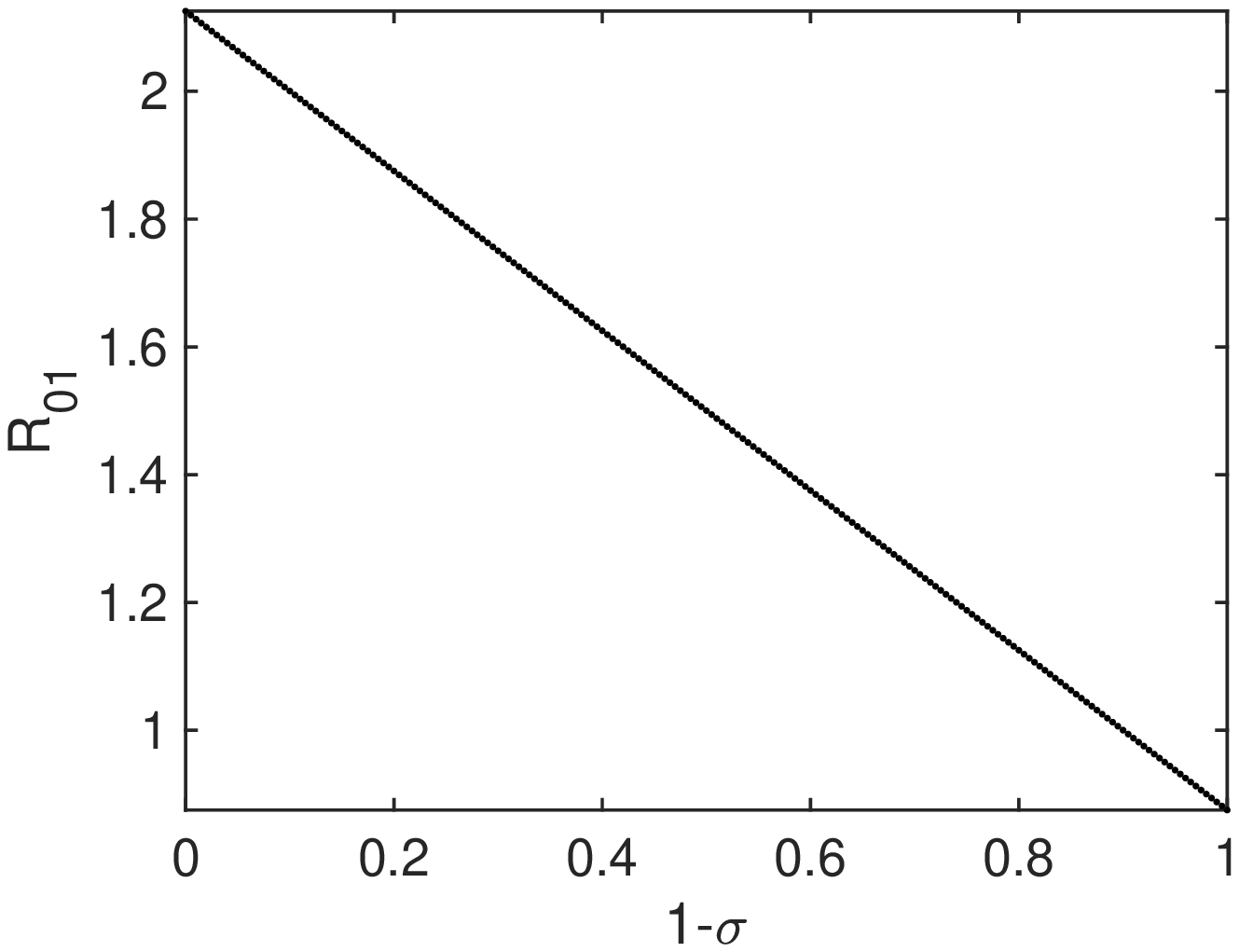}
	\caption{Simulations of the model \eqref{modeloutbreak} showing the impact of testing efficacy $(1-\sigma)$. (a) Weekly new cases for different values of $1-\sigma$. (b) Cumulative cases for different values of $1-\sigma$. (c) $R_{01}$ with respect to $1-\sigma$.}\label{impctsigma}
\end{figure}

\begin{figure}[H]
	(a)\includegraphics[scale=0.45]{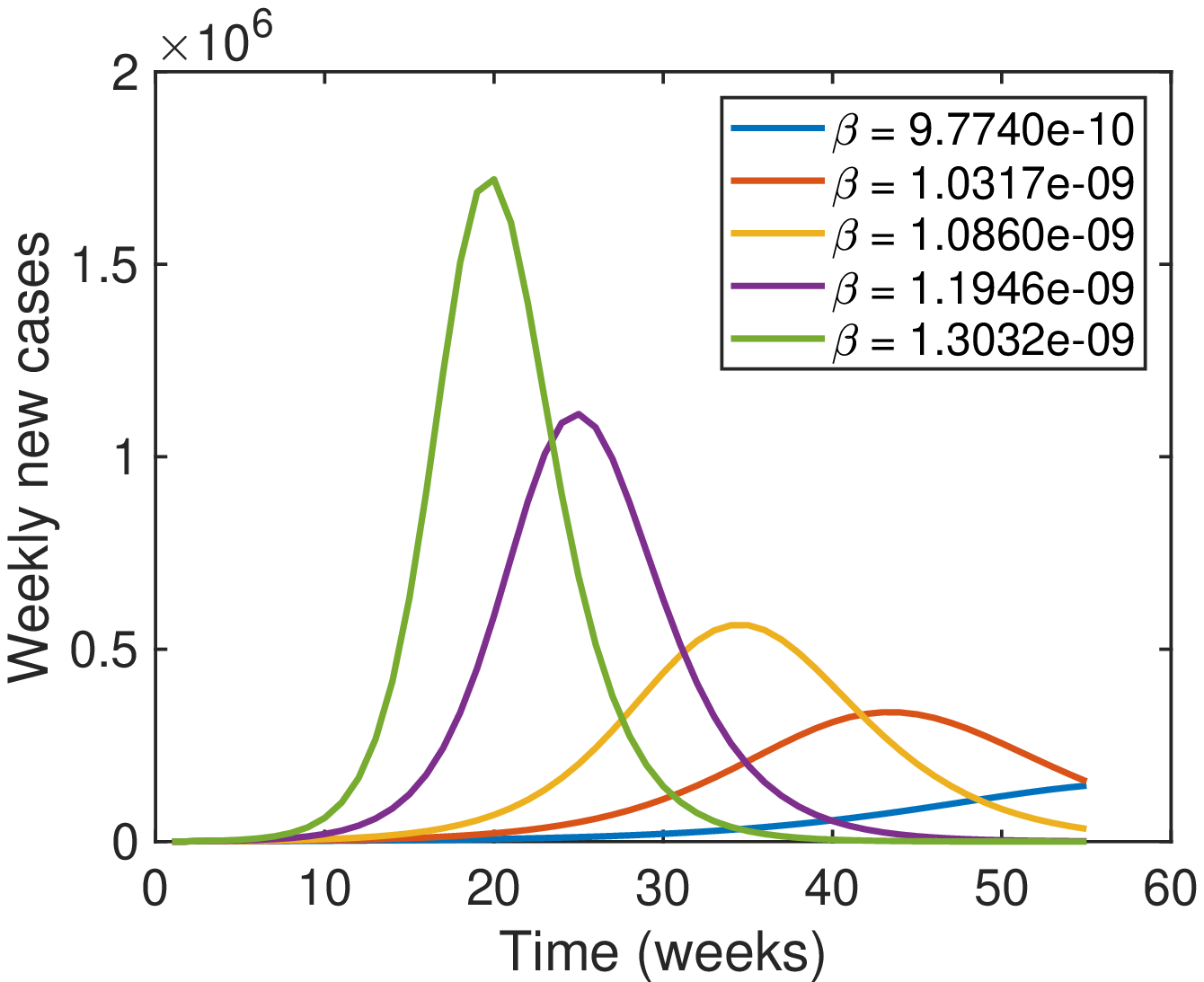}
	(b)\includegraphics[scale=0.45]{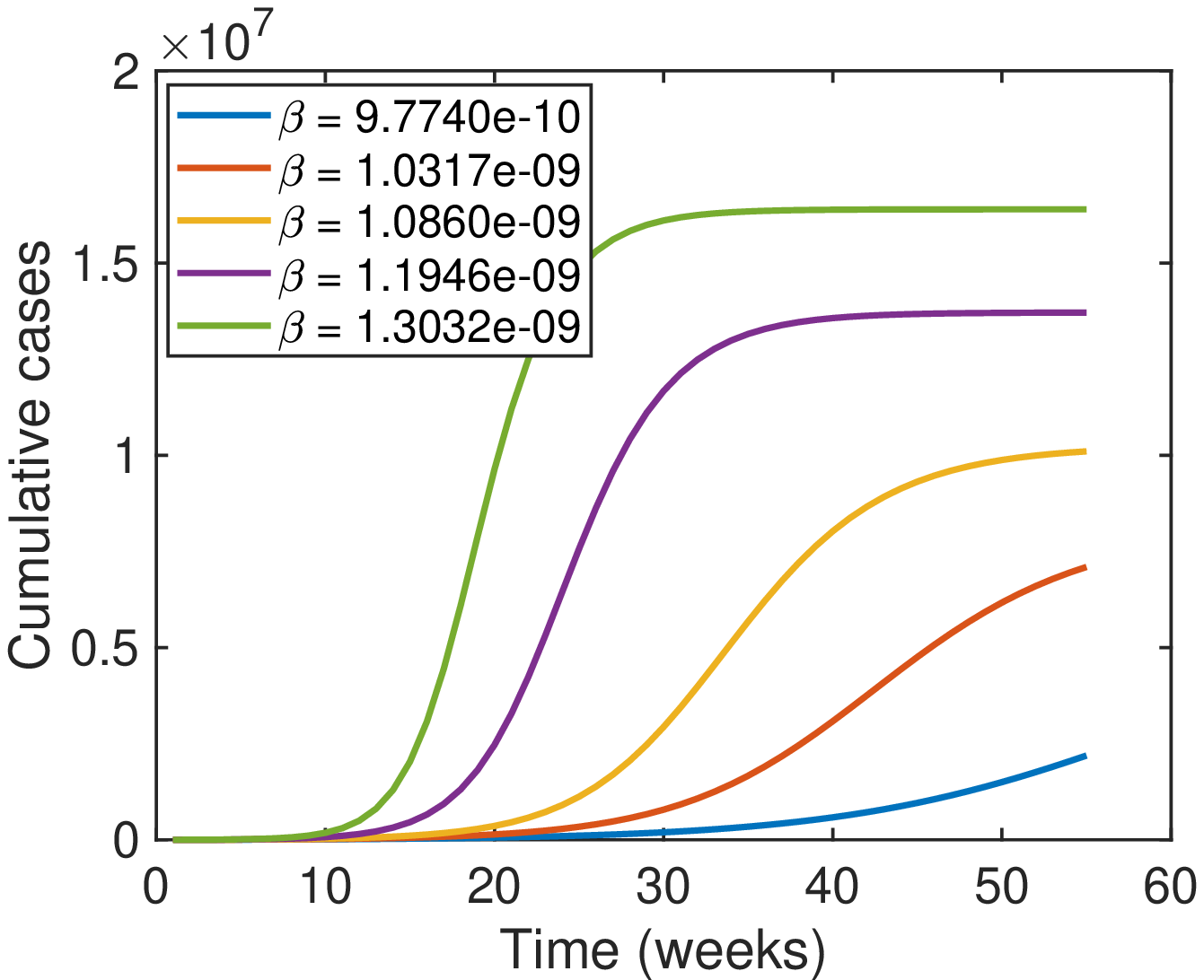}
	(c)\includegraphics[scale=0.45]{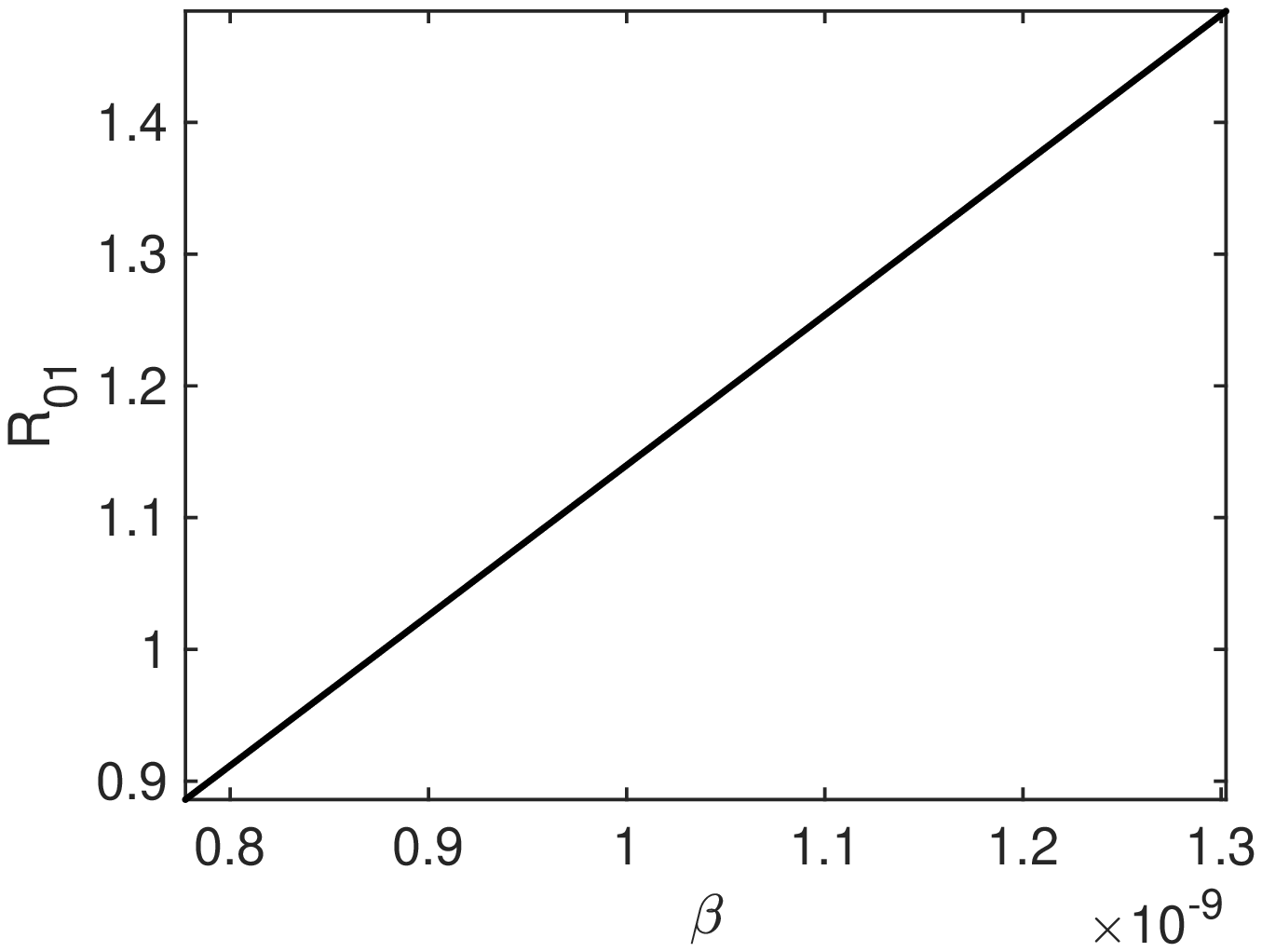}
	\caption{Simulations of the model \eqref{modeloutbreak} representing the impact of different values of the transmission rate $(\beta)$. (a) The trend of weekly new cases for distinct values of $\beta$. (b) Cumulative cases for different values of $\beta$. (c) $R_{01}$ with respect to $\beta$.}\label{impctbeta}
\end{figure}

\begin{figure}[H]
	(a)\includegraphics[scale=0.45]{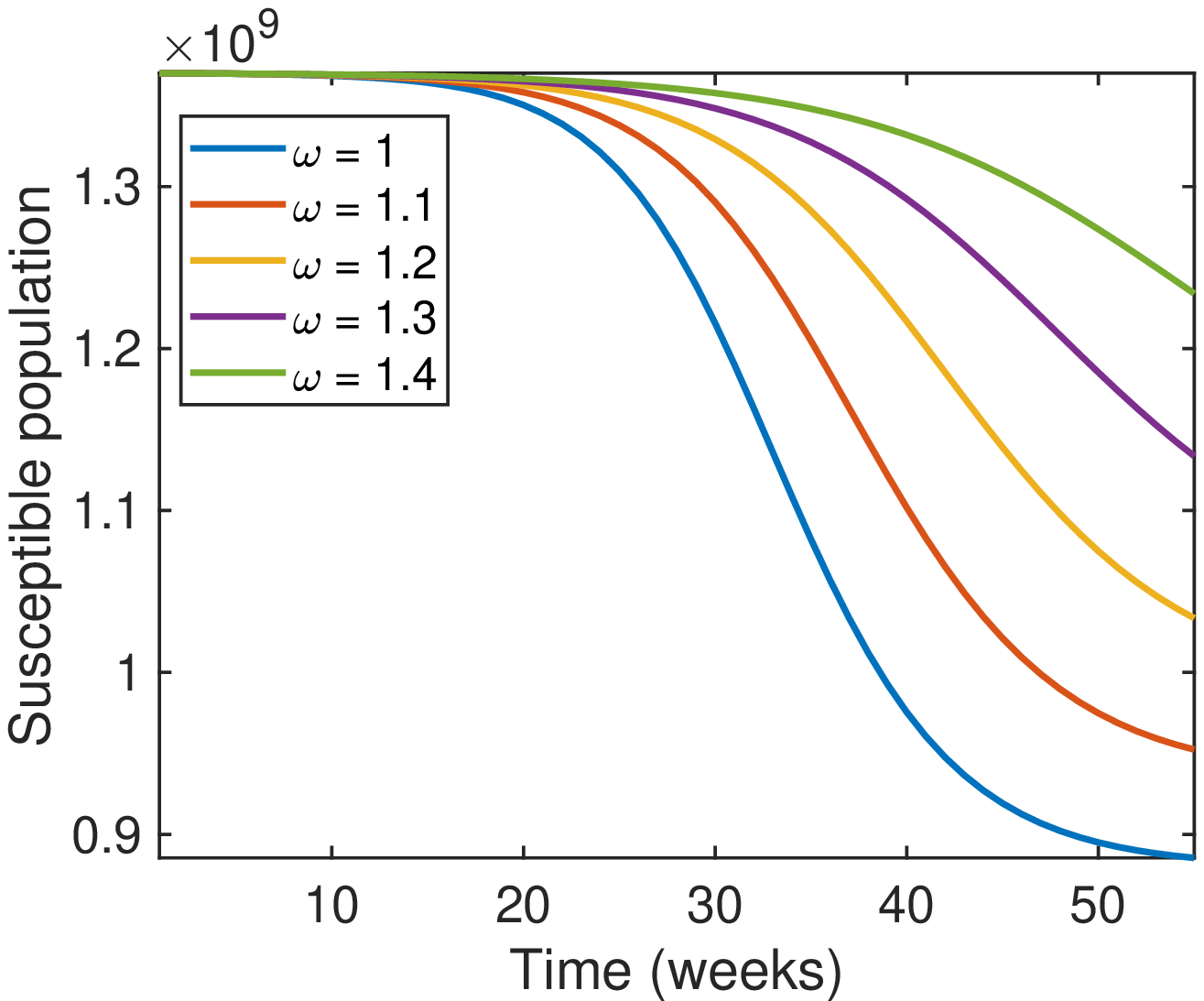}
	(b)\includegraphics[scale=0.45]{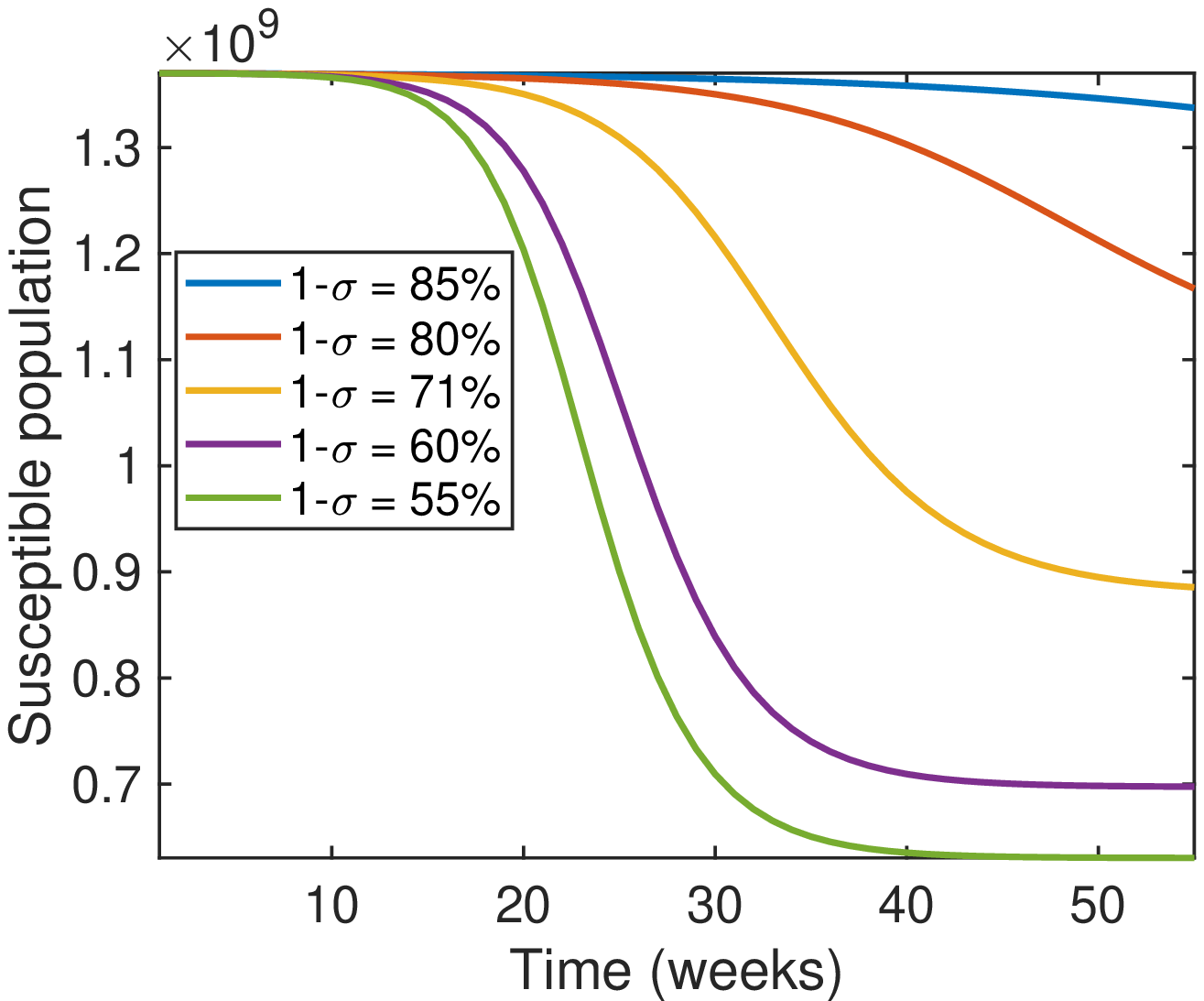}
	(c)\includegraphics[scale=0.45]{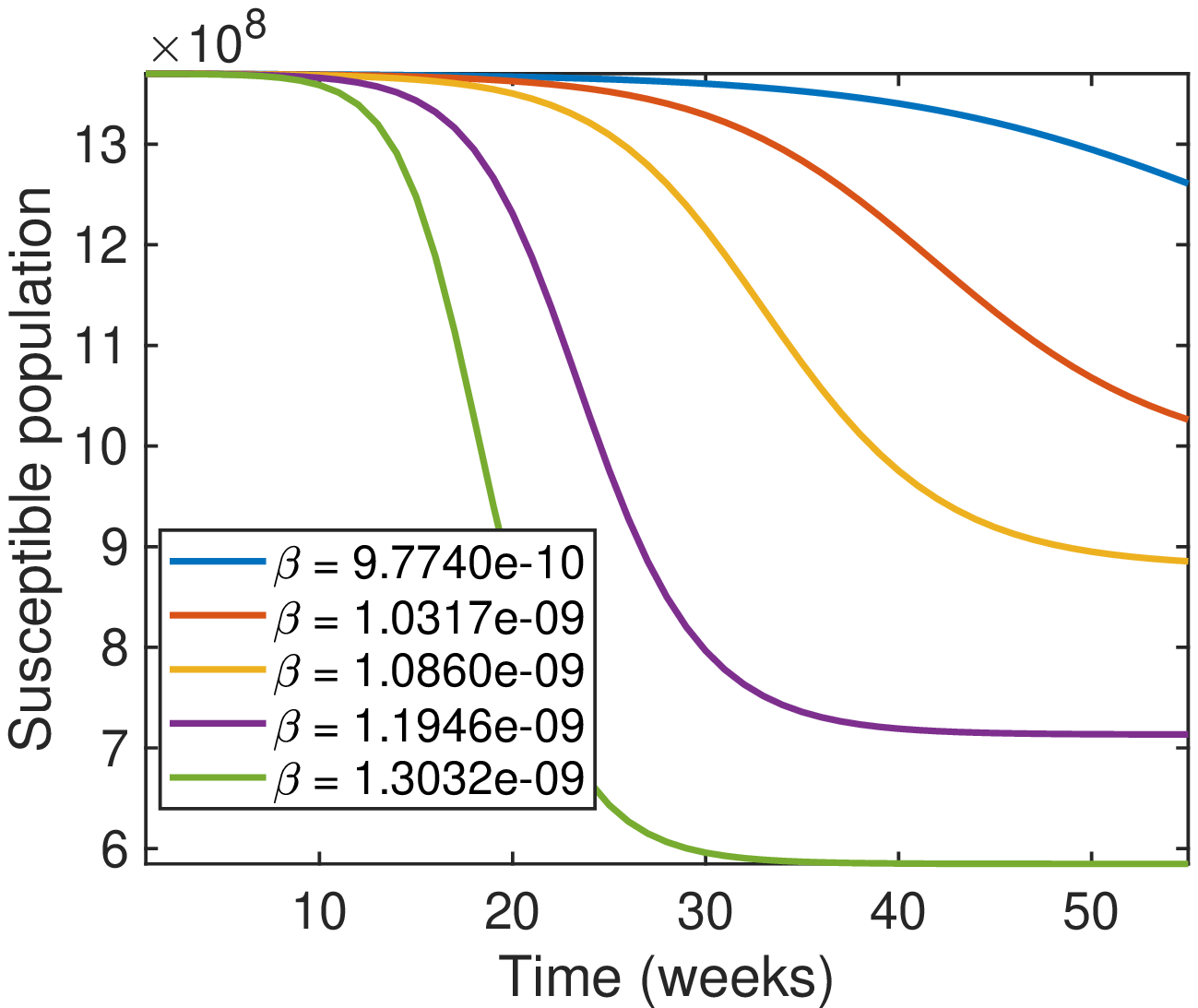}
	\caption{Simulations show the final size relations of the early COVID-19 epidemic in India, i.e., susceptible individuals, $S_{\infty}$, who escaped the epidemic for the different values of (a) $\omega$, (b)  $1-\sigma$, and (c) $\beta$.}\label{finalsize}
\end{figure}

\noindent \textbf{Severity of the epidemic:} Numerical simulations for the final size of the epidemic are shown in Figure \ref{finalsize}. For the baseline parameter values from Table \ref{fitval} and initial conditions in Table \ref{estmdini}, the final size of susceptible population is $S_{\infty} = 885470000$. We assess the impact of increasing testing rate and efficacy as well as the variation in transmission rate on the final size of susceptibles. From Figure \ref{finalsize}(a), we observe that if the testing rate is increased by 10\%, 20\%, 30\%, and 40\% from its baseline value, then the final susceptible individuals are increased by 7.55\%, 16.71\%, 28.02\%, and 39.35\%, respectively. Figure \ref{finalsize}(b) reveals that increasing testing efficacy by 12.67\% and 19.71\% from its baseline value increases the final susceptible population by 31.79\% and 51.04\%, respectively, while decreasing testing efficacy by 15.49\% and 22.53\% from its baseline value reduces the final susceptible individuals by 21.21\% and 28.80\%, respectively. These results imply that increasing the testing rate and efficacy reduces the epidemic's severity, i.e., more individuals escape the epidemic. Figure \ref{finalsize}(c) exposes the impact of variation in the transmission rate. If the transmission rate increases by 10\% and 20\% from its baseline value, the final susceptible individuals decrease by 19.43\% and 33.96\%, respectively. In contrast, if it decreases by 5\% from its baseline value, the final number of susceptible individuals increases by 15.90\%. These results ensure that the disease is more severe for high transmission rate, which is quite obvious.    
\\\\
\noindent \textbf{Combined impact of the parameters:} 
We also observe the combined impact of the testing rate ($\omega$), testing efficacy ($1-\sigma$), and transmission rate ($\beta$) on $R_{01}$ by plotting the surface plot in Figure \ref{surfplot}. Since $R_{01}$ is directly proportional to $\beta$ therefore it quickly decreases with the decreasing value of $\beta$, however $\omega$ and $1-\sigma$ also have noteworthy  influence, illustrated in Figure \ref{surfplot}(a) and (b). Figure \ref{surfplot}(c) compares the testing rate and efficacy, which shows that when testing efficacy is high, then a higher testing rate reduces $R_{01}$ more, while if the testing efficacy is low, then the high testing rate keeps $R_{01}$ high, implies that testing rate is most significant when the testing efficacy is high. Figure \ref{surfplot}(c) also exhibits that the testing efficacy considerably reduces the basic reproduction number compared to the testing rate. 

In addition, Figure \ref{contourplot} demonstrates the impact of the testing rate and efficacy for different transmission rates in reducing the basic reproduction number, $R_{01}$. These contour plots exhibit that the testing efficacy plays a more significant role than the testing rate. These contour plots also reveal that if the transmission rate increases, then $R_{01}$ increases for the lower efficacy. Hence, testing efficacy needs to be increased to fetch the basic reproduction number ($R_{01}$) less than unity when the transmission rate is high. 

\begin{figure}[H]
	(a)\includegraphics[scale=0.45]{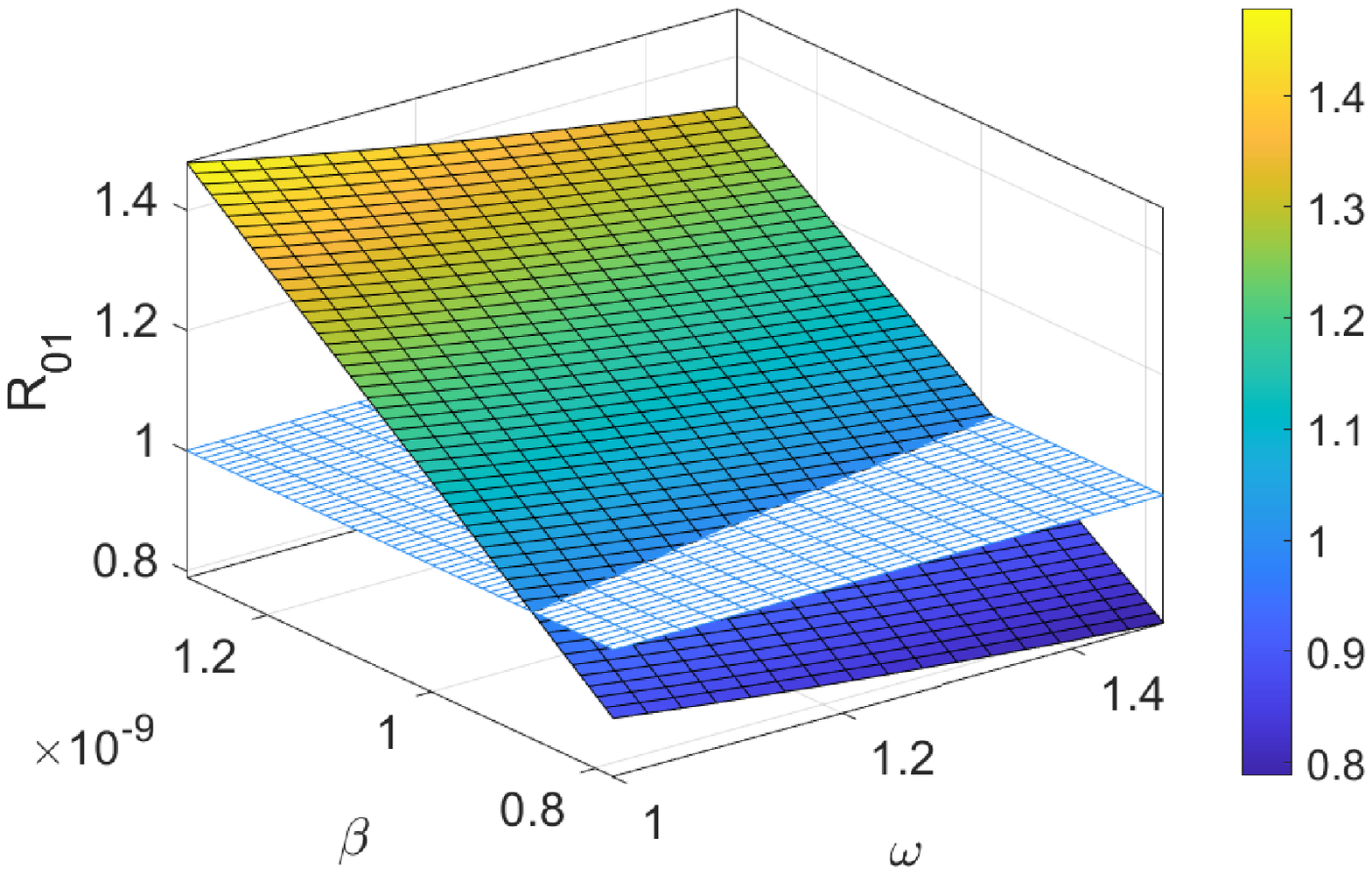}
	(b)\includegraphics[scale=0.42]{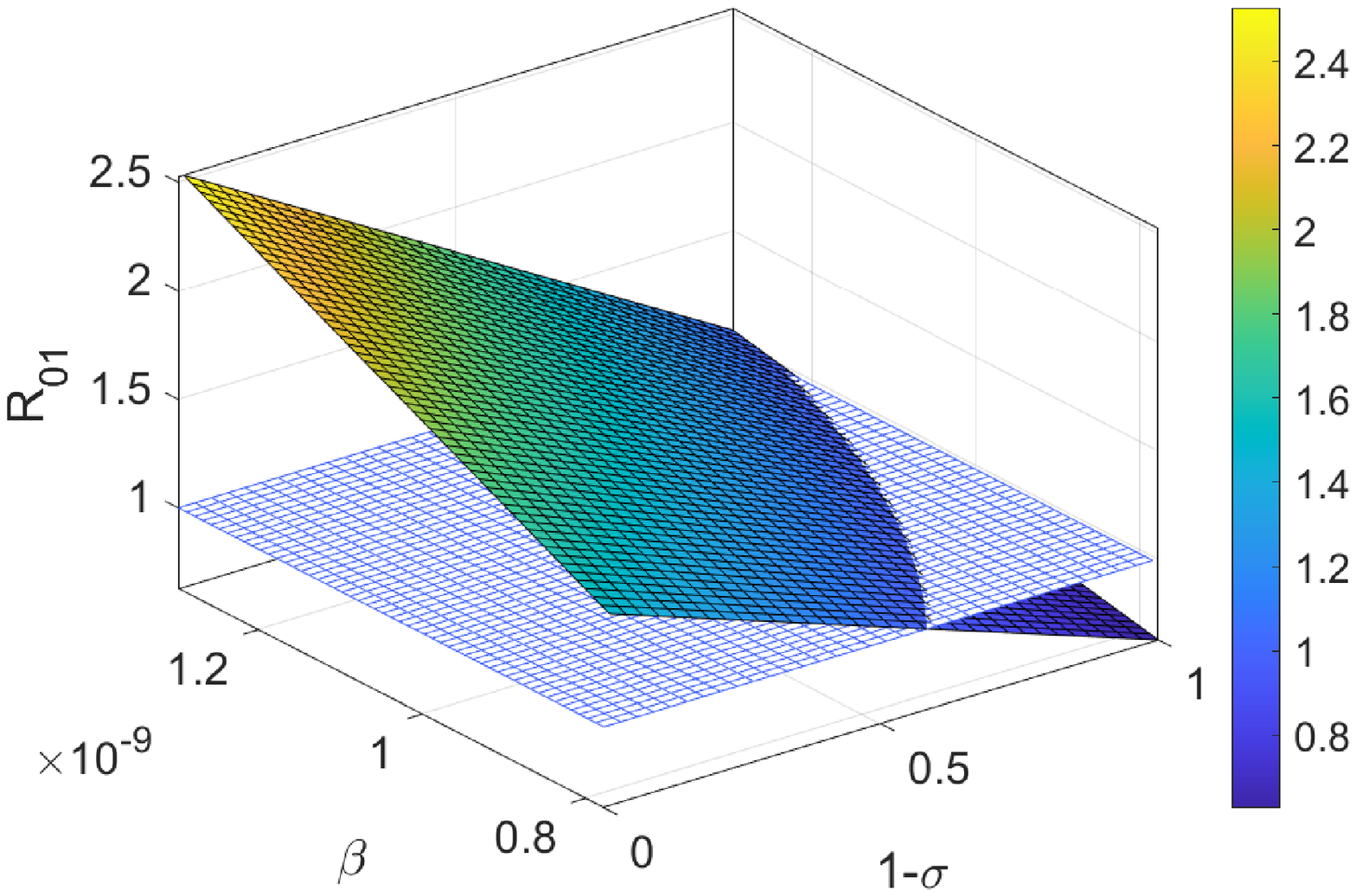}
	(c)\includegraphics[scale=0.41]{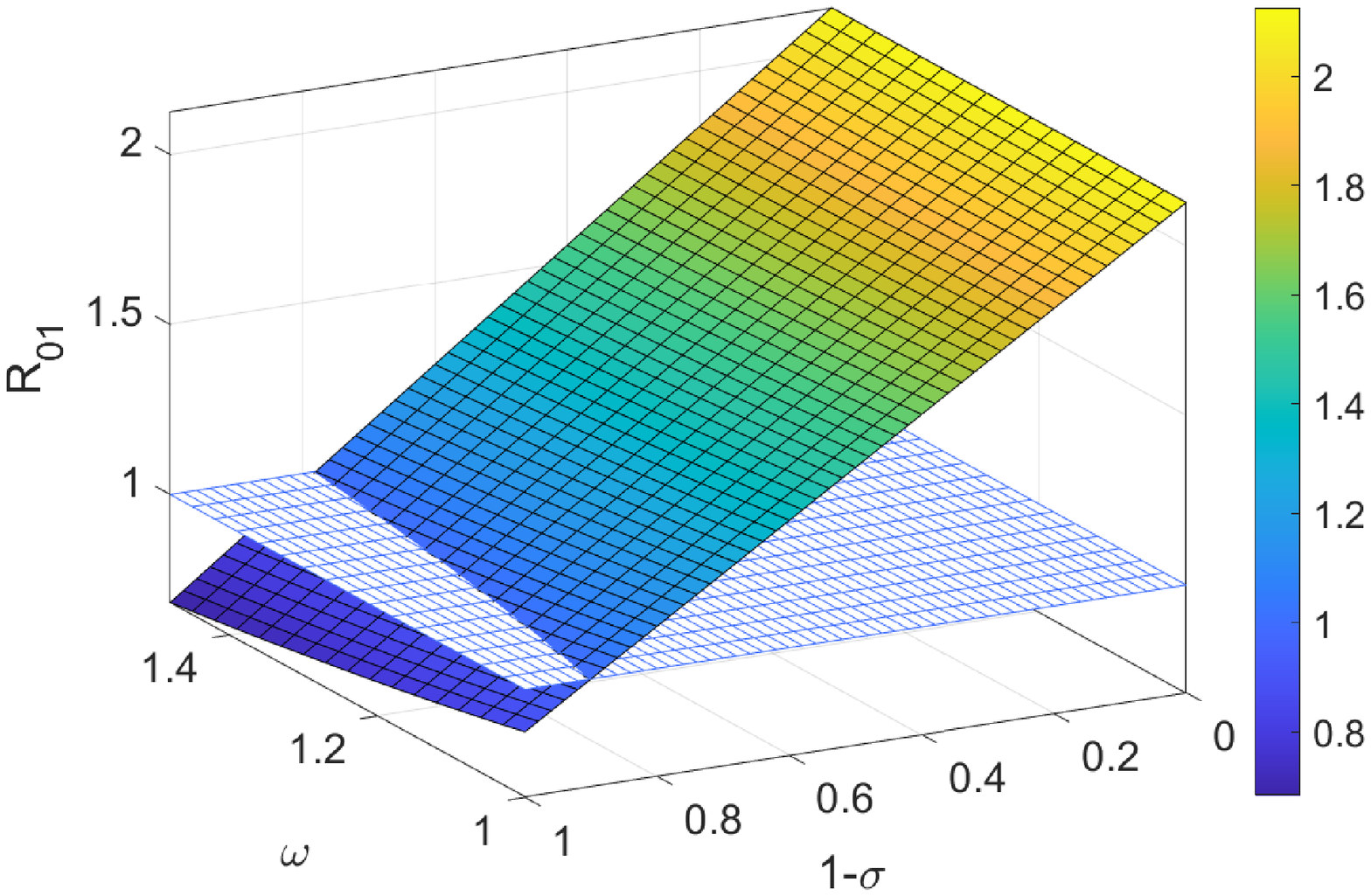}
	\caption{Surface plots demonstrate the impacts of different combinations of the parameters by considering $R_{01}$ as a function of (a) $\beta$ and $\omega$, (b) $\beta$ and $\sigma$,  (c) $\sigma$ and $\omega$.}\label{surfplot}
\end{figure}

\begin{figure}[H]
	\includegraphics[scale=1]{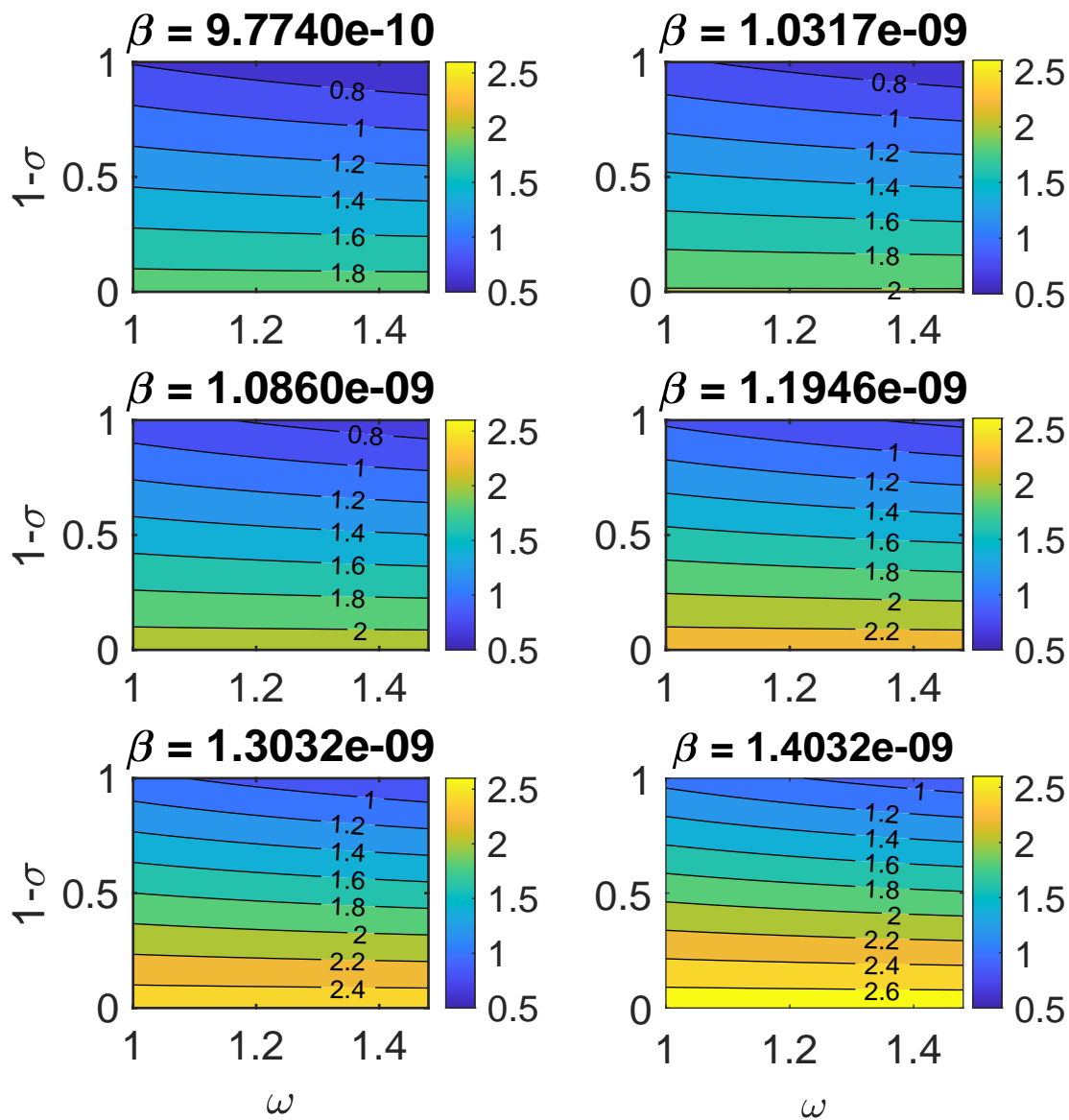}
	\caption{Effectiveness of testing rate ($\omega$) and efficacy ($1-\sigma$) for different transmission rates ($\beta$) in reducing the basic reproduction number, $R_{01}$.}\label{contourplot} 
\end{figure}

\section{Discussion and conclusions}\label{diss}
The control of infectious diseases depends on many factors or non-pharmaceutical/pharmaceutical intervention strategies, such as vaccination, drugs, isolation, use of face masks, social distancing, awareness, etc. A lots of studies have already been done with several intervention strategies  \cite{BAJ2020,BUG2020,BUG2022,FEN2007,WANG2022,Young2019,elasticity,FER2020,GUM2004,BAG2013,HOD2017,KUC2020,DAV2020} and many more, but a few studies \cite{VIL2017,CHI2020,STU2021,GRI2021,SAL2014} have included the testing strategy yet. It is an essential first step to diagnosis the disease, and an accurate estimation of the burden of any disease is crucial to notify the pandemic response. Because of the imperfection in testing, case counts cannot seize the entire burden of the disease. Also, sometimes testing is largely limited to persons with moderate to severe symptoms due to inadequate test accessibility \cite{WU2020}. To capture the scenario related to an imperfect testing strategy, we proposed a compartmental non-linear ordinary differential equation model. The model comprises numerous crucial epidemiological characteristics of a contagious disease, and is intended and utilized to evaluate the possible outcomes of an imperfect testing strategy. It must be remarked that the proposed model could also be utilized to analyze other contagious diseases that follow transitions between compartments, as displayed in Figure \ref{schematic}. Some of the main mathematical and epidemiological findings of this study include the following: 
\begin{itemize}
	\item[(i)] The dynamics of the proposed model \eqref{model}, in the absence of immigration of infectives, i.e., $q=0,$ is entirely governed by the basic reproduction number, $R_0$. The model with $q=0$ has a globally asymptotically stable disease-free equilibrium whenever $R_0<1$ and a unique globally asymptotically stable endemic equilibrium whenever $R_0 > 1$ implying that the disease can be eradicated from the population if the testing has adequate efficacy to maintain (and bring) $R_0$ to a value less than one. This outcome has vital public health inferences since imperfect testing is probably cost-effective. It imposes substantial socio-economic burden in the community, only correct results on the detection move to the quarantine/isolation compartment. The model with $q \neq 0,$ does not have a disease-free equilibrium, and disease never dies out from the community. For this case, a unique endemic equilibrium is globally asymptotically stable, which has been proved by constructing a suitable Lyapunov function.     
	\item[(ii)] A single outbreak model has also been analyzed, including an outbreak's peak and final size relation. The parameter estimation has been done by the maximum likelihood method with respect to the data of the early COVID-19 outbreak in India. We explored the practical identifiability of the single outbreak model \eqref{modeloutbreak}. We fit two parameters because fixing more and fitting fewer parameters makes the model identifiable. We found that our estimated parameters are practically identifiable; it could be seen by the profile likelihoods in Figure \ref{costfun}. These profile likelihoods give the 95\% confidence bounds and a single optimum for each parameter. Moreover, if we try to estimate more parameters in the model, the model may not be practically identifiable. For instance, some previous studies \cite{KAO2018, EIS2013} concerned the issue of the  unidentifiability of the models. The complex model is often more likely to be unidentifiable. 
	\item[(iii)] The basic reproduction number increases more with imperfect testing in comparison to perfect testing. However, the basic reproduction number decreases with the increasing testing efficacy and the testing rate, which is analyzed mathematically for the model \eqref{model} (in subsection \ref{difsce}) as well as numerically for the model \eqref{modeloutbreak} (in subsection \ref{5.5}).  
	\item[(iv)] When observing the impact of the parameters on the weekly new cases, cumulative cases, and basic reproduction number, we note that the transmission rate $\beta$ is highly sensitive to the weekly new cases as well as cumulative cases. To this end, reducing the contact rate may significantly be helpful to decrease the disease burden. The testing rate is also highly sensitive, but in a positive way, it can detect more infected individuals and move them into the isolation/quarantine compartment, which helps reduce the overall disease burden.    
	\item[(v)] We notice that high testing rate and efficacy are helpful in reducing the basic reproduction number as well as the outbreak's peak, which implies that these strategies degrade the disease burden and secondary infections. We obtained that increasing testing rate and efficacy also increase the final number of susceptible individuals, i.e., more susceptible individuals escaped the epidemic which reduces the severity of the epidemic. Lessening the transmission rate decreases the peak size and delays the peak time, implying that intervention strategies that reduce the transmission rate are always helpful in eradicating the disease. It is also explored that the testing rate is more significant when testing efficacy is high in reducing secondary infection. By plotting the contour plots in Figure \ref{contourplot}, it is also revealed that for the increased transmission rate, testing efficacy needs to be improved to bring the basic reproduction number to less than one.       
	\item[(vi)] For the proposed model, testing is always advantageous to the society, although its inclusive impression upsurges with mounting rate and efficacy. It is quite a constructive point since it is realized that imperfect testing can sometimes develop detrimental outcomes for the community.         
\end{itemize}

Overall, this study confirms that an imperfect testing strategy can increase the basic reproduction number and consequently endemicity in the community. The testing rate and efficacy are essential and can make the basic reproduction number less than one; however, these are more effective if combined with other intervention strategies that reduce the transmission rate. It should be declared that the proposed model is relatively simple. We did not analyze the general model \eqref{model} (with demographic effect and waning immunity rate) computationally. The same results as for the single outbreak model \eqref{modeloutbreak} (in Section \ref{5.5}) could be found for the testing rate and efficacy concerning endemicity (for long-term disease dynamics) and the basic reproduction number, $R_0$ for the model \eqref{model}. However, we have found that the recruitment of infected individuals increases the endemicity (in Figure \ref{eqlibrm}(b)), and the disease never dies out in the community. 
The model \eqref{modeloutbreak} is also a simplified interpretation of the role of an imperfect testing strategy in disease dynamics and only considers a single outbreak and one viral strain. 
%The model may require to refine in the future as more information (and knowledge) becomes available.

\section*{Acknowledgments}
The research work of Sarita Bugalia is supported by the Council of Scientific \& Industrial Research (CSIR), India [File No. 09/1131(0025)/2018-EMR-I]. The research work of Jai Prakash Tripathi is supported by the Science and Engineering Research Board (SERB), India [File No. ECR/2017/002786].

\section*{Appendix} \label{apndxa}
\textbf{Proof of Theorem \ref{thm3.1}:} From system \eqref{model}, we have
\begin{equation*}
\begin{aligned}
\frac{dS}{dt}\Big|_{S=0,I\geq 0, T_n \geq 0, J\geq 0, R\geq 0}&=(1-p-q)\Lambda + \xi R >0,\\
\frac{dI}{dt}\Big|_{S> 0,I= 0, T_n \geq 0, J\geq 0, R\geq 0}&= q\Lambda + \beta S T_n > 0,\\
\frac{dT_n}{dt}\Big|_{S> 0,I\geq 0, T_n = 0, J\geq 0, R\geq 0}&= \sigma \omega I \geq 0,\\
\frac{dJ}{dt}\Big|_{S> 0,I\geq 0, T_n \geq 0, J=0, R\geq 0}&= (1- \sigma) (1-\theta) \omega I \geq 0,\\
\frac{dR}{dt}\Big|_{S> 0,I\geq 0, T_n \geq 0, J\geq 0, R= 0}&= p\Lambda + \rho J > 0.
\end{aligned}
\end{equation*}
The above calculation shows that all the rates are non-negative on the boundary planes of $\mathbb{R}_+^5$. Therefore, we can deduce that the direction of a vector field is inward from the boundary planes. Thus, whenever the system begins in a non-negative $\mathbb{R}_+^5$, all the solutions remain in the positive region. 

Consider $N(t) = S(t)+I(t)+T_n(t)+J(t)+R(t)$ the total population size corresponding to the system \eqref{model}. Adding equations of system \eqref{model} gives the equation for the total population:
\begin{equation*}
\frac{dN}{dt} = \frac{dS}{dt}+\frac{dI}{dt}+\frac{dT_n}{dt}+\frac{dJ}{dt}+\frac{dR}{dt},
\end{equation*}
which yields 
\begin{equation*}
\frac{dN}{dt}=\Lambda - \mu N -\delta_1 I - \theta \omega I - \delta_2 T_n - \delta_3 J \implies \Lambda - (\mu + \delta_1 + \theta \omega + \delta_2 + \delta_3)N \leq \frac{dN}{dt} \leq \Lambda - \mu N.
\end{equation*}
By integrating both sides, we obtain 
{\footnotesize \begin{equation*}
	\frac{\Lambda}{\mu + \delta_1 + \theta \omega + \delta_2 + \delta_3}+\left(N(0) - \frac{\Lambda}{\mu + \delta_1 + \theta \omega + \delta_2 + \delta_3} \right)e^{-(\mu + \delta_1 + \theta \omega + \delta_2 + \delta_3)t}\leq N(t) \leq \frac{\Lambda}{\mu}+\left(N(0) - \frac{\Lambda}{\mu} \right)e^{-\mu t}.
	\end{equation*} }
Taking $t \rightarrow \infty$, we obtain
{\footnotesize \begin{equation*}
	\frac{\Lambda}{\mu + \delta_1 + \theta \omega + \delta_2 + \delta_3} \leq \liminf\limits_{t \rightarrow \infty}N(t) \leq \limsup\limits_{t \rightarrow \infty}N(t) \leq \frac{\Lambda}{\mu} \implies \frac{\Lambda}{\mu + \delta_1 + \theta \omega + \delta_2 + \delta_3} \leq N(t) \leq \frac{\Lambda}{\mu},
	\end{equation*}}
with $\limsup \limits_{t\rightarrow \infty} N(t) = N_0 = \frac{\Lambda}{\mu}$ if and only if $\limsup \limits_{t\rightarrow \infty} I(t) = 0$, $\limsup \limits_{t\rightarrow \infty} T_n(t) = 0$, and $\limsup \limits_{t\rightarrow \infty} J(t) = 0$. Subsequently, in the absence of infection $(I = T_n = J=0)$, the total population, $N$, approaches the carrying capacity, $N_0$, asymptotically; and in the existence of infection, the total population is less than or equal to $N_0$. Hence, whenever the trajectory begins inside the region of $\Omega$, it remains in the region. On the contrary, if the trajectory begins outside of $\Omega$, then it will move into the region and remain in it. Therefore, the $\omega$-limit sets of system \eqref{model} are enclosed in $\Omega$. It implies that $\Omega$ is a positively invariant set for system \eqref{model}. Hence the proof is completed. 
\\
\textbf{Proof of Theorem \ref{gbstabend}:} We first consider the case $q \neq 0$ and
\begin{align*}
x=\frac{S}{S^*},~ y=\frac{I}{I^*},~ z=\frac{T_n}{T_n^*},~ u=\frac{J}{J^*},~ v = \frac{R}{R^*},
\end{align*}
and by using the Eqs. of the model \eqref{model}, the model \eqref{model} can be  reshaped in the following form:
\begin{equation}\label{modellyp1}
\begin{aligned}
x' &= x\Big[\frac{(1-p-q)\Lambda}{S^*}\Big(\frac{1}{x}-1\Big) - \beta I^* (y-1) - \beta T_n^* (z-1) + \frac{\xi R^*}{S^*}\Big(\frac{v}{x} -1 \Big) \Big],\\
y' & =  y\Big[\frac{q \Lambda}{I^*}\Big(\frac{1}{y}-1\Big) + \beta S^* (x-1) + \frac{\beta S^* T_n^*}{I^*}\Big(\frac{xz}{y}-1 \Big) \Big],\\
z' &= z \frac{\sigma \omega I^*}{T_n^*}\Big[\frac{y}{z}-1 \Big],\\
u' &= u  \frac{(1-\sigma) \omega I^*}{J^*}\Big[\frac{y}{u}-1 \Big],\\
v' & = v \Big[\frac{p \Lambda}{R^*}\Big(\frac{1}{v} - 1\Big) + \frac{\rho J^*}{R^*}\Big(\frac{u}{v} - 1\Big) \Big].
\end{aligned}
\end{equation}
Further, we ponder the subsequent Lyapunov function
\begin{equation*}
\begin{aligned}
Z = S^* (x-1-\ln x) + I^* (y-1-\ln y) + T_n^* (z-1-\ln z) + J^* (u-1-\ln u) + R^* (v-1-\ln v).
\end{aligned}
\end{equation*}
Differentiating the above function $Z$ with respect to $t$ along the solutions of system \eqref{model}, yields
\begin{equation*}
\begin{aligned}
Z' =& (x-1) \Big[(1-p-q)\Lambda \Big(\frac{1}{x}-1\Big) - \beta S^* I^* (y-1) - \beta S^* T_n^* (z-1) + \xi R^*\Big(\frac{v}{x} -1 \Big) \Big] \\
&+ (y-1) \Big[q \Lambda \Big(\frac{1}{y}-1\Big) + \beta S^* I^* (x-1) + \beta S^* T_n^* \Big(\frac{xz}{y}-1 \Big) \Big] +(z-1) \Big[\sigma \omega I^* \Big(\frac{y}{z}-1\Big) \Big] \\
& + (u-1) \Big[ (1-\sigma) \omega I^* \Big(\frac{y}{u}-1 \Big) \Big] + (v-1) \Big[p \Lambda\Big(\frac{1}{v} - 1\Big) + \rho J^*\Big(\frac{u}{v} - 1\Big) \Big] \\
= & 2\Lambda + \xi R^* + \omega I^* + \rho J^* - x \Big[ (1-p-q)\Lambda - \beta S^* T_n^* + \xi R^* \Big] - \frac{1}{x} (1-p-q)\Lambda \\
& - y \Big[q\Lambda + \beta S^* T_n^* - \omega I^* \Big] - z \Big[\sigma \omega I^* -\beta S^* T_n^* \Big] - v \Big[ p\Lambda + \rho J^* - \xi R^* \Big]  - \frac{v}{x} \xi R^*\\
& - \frac{1}{y} q \Lambda - \frac{xz}{y} \beta S^* T_n^* - \frac{y}{z} \sigma \omega i^* - u \Big[ (1-\sigma)\omega I^* - \rho J^* \Big] - \frac{1}{v}p\Lambda - \frac{u}{v} \rho J^* - \frac{y}{u} (1-\sigma) \omega I^*\\
&=: G(x,y,z,u,v).
\end{aligned}
\end{equation*}
It further yields
\begin{equation*}
\begin{aligned}
G(x,y,z,u,v) =& b_1 \Big(2-x-\frac{1}{x}\Big) + b_2 \Big(2-y-\frac{1}{y}\Big) + b_3 \Big(2-v-\frac{1}{v}\Big) + b_4 \Big(3- \frac{1}{y} - z - \frac{y}{z}\Big) \\
& + b_5 \Big(3- x - \frac{1}{v} -\frac{v}{x}\Big) + b_6 \Big(3- \frac{1}{x} - \frac{xz}{y} -\frac{y}{z}\Big) + b_7 \Big(4- \frac{1}{y} - v - \frac{y}{u} - \frac{u}{v}\Big)\\
& + b_8 \Big(3- u - \frac{y}{u} -\frac{1}{y}\Big) + b_9 \Big(4- \frac{1}{v} - \frac{v}{x} - \frac{y}{z} - \frac{xz}{y}\Big),
\end{aligned}
\end{equation*} 
where
\begin{align*}
b_2 &= q\Lambda + \beta S^* T_n^* - \omega I^*, \\
b_3 &= p\Lambda - \xi R^*,\\
b_4 &=  \sigma \omega I^* - \beta S^* T_n^*,\\
b_5 & = (1-p-q)\Lambda - \beta S^* T_n^* + \xi R^* - b_1,\\
b_6 & = (1-p-q)\Lambda - b_1,\\
b_7 & = \rho J^*,\\
b_8 & =(1-\sigma)\omega I^* - \rho J^*,\\
b_9 & = \beta S^* T_n^* - (1-p-q)\Lambda + b_1.
\end{align*}
To assure that $b_5, b_6$ and $b_9$ are all nonnegative, $b_1$ should fulfill the ensuing inequalities:
\begin{align}\label{con_b1}
(1-p-q)\Lambda - \beta S^* T_n^* \leq b_1 \leq \min \left\lbrace (1-p-q)\Lambda, (1-p-q)\Lambda + \xi R^* - \beta S^* T_n^*  \right\rbrace. 
\end{align}
Notice that the inequality $(1-p-q)\Lambda - \beta S^* T_n^* < \min \left\lbrace (1-p-q)\Lambda, (1-p-q)\Lambda + \xi R^* - \beta S^* T_n^*  \right\rbrace$ is always true. 

The constrained condition for $b_1$, the inequality \eqref{con_b1}, exposes that the obtainable value of $b_1$ is non-unique, then the associated coefficients 
$b_k(k = 5, 6, 9)$ are also non-unique. Thus, the form of the associated function $G(x,y,z,u,v)$ is also non-unique. To establish the different expressions of $G(x,y,z,u,v)$ for the various parameter values under the condition \eqref{con_b1}, we separate the feasible region of all parameter values into three subregions, as follows:\\
(R1): $(1-p-q)\Lambda > \beta S^* T_n^*$; \quad (R2): $(1-p-q)\Lambda = \beta S^* T_n^*$; \quad (R3): $(1-p-q)\Lambda < \beta S^* T_n^*$.
Then we may select a specific case of $b_k(k = 5, 6, 9)$ for each subregions.

For Case (R1), we choose that $b_1 = (1-p-q)\Lambda - \beta S^* T_n^*$, $b_2 = q\Lambda + \beta S^* T_n^*- \omega I^*$, $b_3 = p \Lambda - \xi R^*$, $b_4 = \sigma \omega I^* - \beta S^* T_n^*$, $b_5 = \xi R^*$, $b_6 = \beta S^* T_n^*$, $b_7 = \rho J^*$, $b_8 = (\mu + \delta_3)J^*$, $b_9 = 0$, then the associated function $G(x,y,z,u,v)$ is 
\begin{equation*}
\begin{aligned}
G(x,y,z,u,v) =& ((1-p-q)\Lambda - \beta S^* T_n^*) \Big(2-x-\frac{1}{x}\Big) + (q\Lambda + \beta S^* T_n^*- \omega I^*) \Big(2-y-\frac{1}{y}\Big)\\
& + (p \Lambda - \xi R^*) \Big(2-v-\frac{1}{v}\Big) + (\sigma \omega I^* - \beta S^* T_n^*) \Big(3- \frac{1}{y} - z - \frac{y}{z}\Big) \\
& + \xi R^* \Big(3- x - \frac{1}{v} -\frac{v}{x}\Big) + \beta S^* T_n^* \Big(3- \frac{1}{x} - \frac{xz}{y} -\frac{y}{z}\Big) \\
& + \rho J^* \Big(4- \frac{1}{y} - v - \frac{y}{u} - \frac{u}{v}\Big) + (\mu + \delta_3)J^* \Big(3- u - \frac{y}{u} -\frac{1}{y}\Big).
\end{aligned}
\end{equation*} 
For Case (R2), we choose that $b_1 = 0$, $b_2 = q\Lambda + \beta S^* T_n^*- \omega I^*$, $b_3 = p \Lambda - \xi R^*$, $b_4 = \sigma \omega I^* - \beta S^* T_n^*$, $b_5 = \xi R^*$, $b_6 = (1-p-q)\Lambda$, $b_7 = \rho J^*$, $b_8 = (\mu + \delta_3)J^*$, $b_9 = 0$, then the associated function $G(x,y,z,u,v)$ is
\begin{equation*}
\begin{aligned}
G(x,y,z,u,v) =& (q\Lambda + \beta S^* T_n^*- \omega I^*) \Big(2-y-\frac{1}{y}\Big) + (p \Lambda - \xi R^*) \Big(2-v-\frac{1}{v}\Big) \\
&+ (\sigma \omega I^* - \beta S^* T_n^*) \Big(3- \frac{1}{y} - z - \frac{y}{z}\Big) + \xi R^* \Big(3- x - \frac{1}{v} -\frac{v}{x}\Big) \\
&+ (1-p-q)\Lambda \Big(3- \frac{1}{x} - \frac{xz}{y} -\frac{y}{z}\Big) + \rho J^* \Big(4- \frac{1}{y} - v - \frac{y}{u} - \frac{u}{v}\Big)\\
& + (\mu + \delta_3)J^* \Big(3- u - \frac{y}{u} -\frac{1}{y}\Big).
\end{aligned}
\end{equation*} 
For Case (R3), we choose that $b_1 = 0$, $b_2 = q\Lambda + \beta S^* T_n^*- \omega I^*$, $b_3 = p \Lambda - \xi R^*$, $b_4 = \sigma \omega I^* - \beta S^* T_n^*$, $b_5 = \beta S^* I^* + \mu S^*$, $b_6 = (1-p-q)\Lambda$, $b_7 = \rho J^*$, $b_8 = (\mu + \delta_3)J^*$, $b_9 = \beta S^* T_n^* - (1-p-q)\Lambda$, then the associated function $G(x,y,z,u,v)$ is
\begin{equation*}
\begin{aligned}
G(x,y,z,u,v) =& (q\Lambda + \beta S^* T_n^*- \omega I^*) \Big(2-y-\frac{1}{y}\Big) + (\mu + \delta_3)J^* \Big(3- u - \frac{y}{u} -\frac{1}{y}\Big) \\
&+ (\sigma \omega I^* - \beta S^* T_n^*) \Big(3- \frac{1}{y} - z - \frac{y}{z}\Big) + (\beta S^* I^* + \mu S^*) \Big(3- x - \frac{1}{v} -\frac{v}{x}\Big)\\
& + (1-p-q)\Lambda \Big(3- \frac{1}{x} - \frac{xz}{y} -\frac{y}{z}\Big) + \rho J^* \Big(4- \frac{1}{y} - v - \frac{y}{u} - \frac{u}{v}\Big)\\
& + (p \Lambda - \xi R^*) \Big(2-v-\frac{1}{v}\Big) + (\beta S^* T_n^* - (1-p-q)\Lambda) \Big(4- \frac{1}{v} - \frac{v}{x} - \frac{y}{z} - \frac{xz}{y}\Big).
\end{aligned}
\end{equation*} 
By using the property that the geometric mean is less than or equal to the arithmetic mean, $G(x,y,z,u,v) \leq 0,$ and the equality is true only for $x=y=v=z=u=1$ i.e., $$\left\lbrace (x,y,z,u,v) \in \Omega : G(x,y,z,u,v) = 0 \right\rbrace \equiv \left\lbrace (x,y,z,u,v): x=y=v=u=z=1 \right\rbrace, $$ which corresponds to the set 
$\Omega' = \left\lbrace (S,I, T_n, J, R): S = S^*, I=I^*, T_n = T_n^*, J = J^*, R = R^* \right\rbrace \subset \Omega.$ It is evident to see that the maximum invariant set of \eqref{model} on the set $\Omega'$ is the singleton $\left\lbrace E^* \right\rbrace $, then the equilibrium $E^*$ is globally stable in $\Omega$ by LaSalle's Invariance Principle \cite{LaSalle1976}.

The global stability of the endemic equilibrium for the case $q=0,$ can be proved by a similar method.
\end{document}